\documentclass[twoside,11pt]{article}

\usepackage{amsmath}
\usepackage{float}
\RequirePackage{bm}
\usepackage{natbib}
\usepackage{multirow}
\usepackage{amsmath}
\usepackage{wrapfig}
\usepackage{dsfont}
\usepackage{tabularx}
\usepackage{subcaption}
\usepackage{enumitem}
 
\bibpunct{(}{)}{;}{a}{,}{,}
 
\usepackage{jmlr2e}
\usepackage{mkolar_definitions}

\usepackage{mathtools}
\usepackage{color}
\mathtoolsset{showonlyrefs=true}

\graphicspath{{fig/}}
   
\usepackage{xargs}
\usepackage[colorinlistoftodos,prependcaption,textsize=tiny]{todonotes}
\newcommandx{\unsure}[2][1=]{\todo[linecolor=red,backgroundcolor=red!25,bordercolor=red,#1]{#2}}
\newcommandx{\change}[2][1=]{\todo[linecolor=blue,backgroundcolor=blue!25,bordercolor=blue,#1]{#2}}
\newcommandx{\info}[2][1=]{\todo[linecolor=OliveGreen,backgroundcolor=OliveGreen!25,bordercolor=OliveGreen,#1]{#2}}
\newcommandx{\improvement}[2][1=]{\todo[linecolor=Plum,backgroundcolor=Plum!25,bordercolor=Plum,#1]{#2}}

\let\hat\widehat
\let\tilde\widetilde
\def\T{\top}

\newcommand{\vpx}{\varphi(x)}

\newcommand{\vpax}{\varphi_1(x)}
\newcommand{\vpbx}{\varphi_2(x)}
\newcommand{\vpa}{\varphi_1}
\newcommand{\vpb}{\varphi_2}
\newcommand{\vpaxi}[1][]{\varphi_{1#1}(x_i)}
\newcommand{\vpbxi}[1][]{\varphi_{2#1}(x_i)}

\usepackage{lastpage}
\jmlrheading{21}{2020}{1-\pageref{LastPage}}{5/19; Revised
3/20}{5/20}{19-383}{Ming Yu, Varun Gupta, and Mladen Kolar}

\ShortHeadings{Simultaneous Inference for Pairwise Graphical Models}{Yu, Gupta, and Kolar}
\firstpageno{1}

\begin{document}

\title{Simultaneous Inference for Pairwise Graphical Models with Generalized Score Matching}

\author{\name Ming Yu \email mingyu@chicagobooth.edu \\
       \name Varun Gupta \email varun.gupta@chicagobooth.edu  \\
       \name Mladen Kolar \email mladen.kolar@chicagobooth.edu  \\
       \addr Booth School of Business \\
       The University of Chicago\\
       Chicago, IL 60637, USA
}

\editor{Jie Peng}

\maketitle

\begin{abstract}

Probabilistic graphical models provide a flexible yet parsimonious
framework for modeling dependencies among nodes in networks.  There is
a vast literature on parameter estimation and consistent model
selection for graphical models.  However, in many of the applications,
scientists are also interested in quantifying the uncertainty
associated with the estimated parameters and selected models, which
current literature has not addressed thoroughly.  In this paper, we
propose a novel estimator for statistical inference on edge parameters
in pairwise graphical models based on generalized Hyv\"arinen scoring
rule.  Hyv\"arinen scoring rule is especially useful in cases where
the normalizing constant cannot be obtained efficiently in a closed
form, which is a common problem for graphical models, including Ising
models and truncated Gaussian graphical models.  Our estimator allows
us to perform statistical inference for general graphical models
whereas the existing works mostly focus on statistical inference for
Gaussian graphical models where finding normalizing constant is
computationally tractable.  Under mild conditions that are typically
assumed in the literature for consistent estimation, we prove that our
proposed estimator is $\sqrt{n}$-consistent and asymptotically normal,
which allows us to construct confidence intervals and build hypothesis
tests for edge parameters. Moreover, we show how our proposed method
can be applied to test hypotheses that involve a large number of model
parameters simultaneously.  We illustrate validity of our estimator
through extensive simulation studies on a diverse collection of
data-generating processes.

\end{abstract}

\begin{keywords}
generalized score matching, high-dimensional inference, probabilistic graphical models, simultaneous inference
\end{keywords}

\section{Introduction}
\label{sec:introduction}

Undirected probabilistic graphical models are widely used to explore
and represent dependencies between random variables
\citep{Lauritzen1996Graphical}. They have been used in areas ranging
from computational biology to neuroscience and finance.  An undirected
probabilistic graphical model consists of an undirected graph
$G = (V,E)$, where $V = \{1, \ldots, p\}$ is the vertex set and
$E \subset V \times V$ is the edge set, and a random vector
$X = (X_1, \ldots, X_p) \in \Xcal^p \subseteq \RR^P$. Each coordinate
of the random vector $X$ is associated with a vertex in $V$ and the
graph structure encodes the conditional independence assumptions
underlying the distribution of $X$. In particular, $X_a$ and $X_b$ are
conditionally independent given all the other variables if and only if
$(a,b) \not\in E$, that is, the nodes $a$ and $b$ are not adjacent in
$G$.  One of the fundamental problems in statistics is that of
learning the structure of $G$ from {\it i.i.d.}~samples from $X$ and
quantifying the uncertainty of the estimated structure.
\cite{Drton2016Structure} provides a recent review of algorithms for
learning the structure, while \cite{Jankova2018Inference} provides an
overview of statistical inference in Gaussian graphical models.

Gaussian graphical models are a special case of undirected
probabilistic graphical models and have been widely studied in the
machine learning literature. Suppose that $X \sim \Ncal(\mu, \Sigma)$.
In this case, the conditional independence graph is determined by the
pattern of non-zero elements of the inverse of the covariance matrix
$\Omega = \Sigma^{-1} = (\omega_{ab})$. In particular, $X_a$ and $X_b$
are conditionally independent given all the other variables in $X$ if
and only if $\omega_{ab}$ and $\omega_{ba}$ are both zero. This simple
relationship has been fundamental for the development of rich literature
on Gaussian graphical models and has facilitated the development of fast
algorithms and inferential procedures \citep[see, for
example,][]{Dempster1972Covariance, Drton04model, Meinshausen2006High,
  Yuan2007Model, Friedman2008Sparse, rothman08spice, Yuan2010High,
  Sun2012Sparse, Cai2011Constrained}.

In this paper, we consider a more general, but still tractable, class
of pairwise interaction graphical models with densities belonging to
an exponential family
$\Pcal = \{ p_\theta(x) \mid \theta \in \Theta \}$ with natural
parameter space $\Theta$:
\begin{multline}
  \label{eq:logdensity}
  \log p_\theta(x) = 
  \sum_{a \in V} \sum_{k \in [K]}\theta_{a}^{(k)} t_{a}^{(k)}(x_a) \\
  +
  \sum_{(a,b)\in E} \sum_{l \in [L]}\theta_{ab}^{(l)} t_{ab}^{(l)}(x_a,x_b)
  - \Psi(\theta) + \sum_{a\in V}h_a(x_a), 
  \quad x \in \Xcal \subseteq \RR^p.
\end{multline}
The functions $t_{a}^{(k)}$, $t_{ab}^{(l)}$ are the sufficient statistics
and $\Psi(\theta)$ is the log-partition function. We assume throughout
the paper that the support of the densities is either $\Xcal = \RR^p$
or $\Xcal = \RR^p_+$ and $\Pcal$ is dominated by Lebesgue measure on
$\RR^p$. To simplify the notation, for a log-density of the form given
in \eqref{eq:logdensity} we will write
$$
 \log p_\theta(x) = \theta^\top  t(x) - \Psi(\theta) + h(x),
$$
where $\theta \in \RR^s$ and $t(x) : \RR^{p} \mapsto \RR^s$ with
$s = L \cdot {p \choose 2} + p\cdot K$. The natural parameter space has
the form
$\Theta = \{ \theta \in \RR^s \mid \Psi(x) =
\log\int_{\Xcal}\exp(\theta^\top  t(x)dx)<\infty\}$.  Under the model in
\eqref{eq:logdensity}, there is no edge between $a$ and $b$ in the
corresponding conditional independence graph if and only if
$\theta_{ab}^{(1)}=\cdots=\theta_{ab}^{(L)}=0$. The model in
\eqref{eq:logdensity} encompasses a large number of graphical models
studied in the literature as we discuss in Section~\ref{sec:related}.
\citet{Lin2015High} studied estimation of parameters in model
\eqref{eq:logdensity}, however, the focus of this paper, as we discuss
next, is on performing statistical inference---constructing honest
confidence intervals and statistical tests---for parameters in
\eqref{eq:logdensity}.

The focus of the paper is on the inferential analysis about parameters
in the model given in \eqref{eq:logdensity}, as well as the Markov
dependencies between observed variables. Our inference procedure does
not rely on the oracle support recovery properties of the estimator
and is therefore uniformly valid in a high-dimensional regime and
robust to model selection mistakes, which commonly occur in ultra-high
dimensional setting. Our approach is based on Hyv\"arinen
generalized scoring rule estimate of $\theta$ in
\eqref{eq:logdensity}. The same procedure was used in
\citet{Lin2015High}, however, rather than focusing on consistent model
selection, we use the initial estimator to construct a regular linear
estimator \citep{Vaart1998Asymptotic}. We establish Bahadur type
representation for our final regular estimator that is robust to model
selection mistakes and valid for a big class of data generating
distributions. The purpose of establishing a Bahadur representation is
to approximate an estimate by a sum of independent random variables,
and hence prove the asymptotic normality of the estimator for
\eqref{eq:logdensity}, allowing us to conduct statistical inference on
the model parameters \citep[see][]{Bahadur1966note}.  In particular,
we show how to construct confidence intervals for a parameter in the
model that have nominal coverage and also propose a statistical test
for existence of edges in the graphical model with nominal size.
These results complement existing literature, which is focused on
consistent model selection and parameter recovery, as we review in the
next section. Furthermore, we develop a methodology for constructing
simultaneous confidence intervals for all the parameters in the model
\eqref{eq:logdensity} and apply this methodology for testing the
parameters in the differential network\footnote{We adopt the notion
  used in \citet{Li2007Finding} and \citet{Danaher2011Joint} and
  define the differential network as a difference between parameters
  of two graphical models.}.  The main idea here is to use the
Gaussian multiplier bootstrap to approximate the distribution of the
maximum coordinate of the linear part in the Bahadur representation.
Appropriate quantile obtained from the bootstrap distribution is used
to approximate the width of the simultaneous confidence intervals and
the cutoff values for the tests for the parameters of the differential
network.

\subsection{Main Contribution}

This paper makes two major contributions to the literature on
statistical inference for graphical models.  First, compared to
previous work on high-dimensional inference in graphical models
\citep{Ren2013Asymptotic, Barber2015ROCKET, Wang2016Inference,
  Jankova2014Confidence}, this is the first work on statistical
inference in models where computing the log-partition function is
intractable.  Existing works mostly focus on Gaussian graphical models
with a tractable normalizing constant, whereas our method can be
applied to more general models, as we discuss in
Section~\ref{sec:ExpoGM}.  Second, we apply our proposed method to
simultaneous inference on all edges connected to a specific node.  Our
simultaneous inference procedure can be used to
\begin{enumerate}[topsep=0pt,itemsep=-1ex,partopsep=1ex,parsep=1ex]
\item test whether a node is isolated in a graph; that is, whether
  it is conditionally independent with all the other nodes;
\item estimate the support of the graph by setting an appropriate
  threshold on the proposed estimators; and
\item test for the difference between graphical models where we have
  observations of two graphical models with the same nodes and we
  would like to test whether the local connectivity pattern for a specific
  node is the same in the two graphs.
\end{enumerate}

Once again, the existing approaches cannot deal with simultaneous testing with an intractable normalizing constant.
Moreover, most of the existing work impose a sparsity condition on the inverse of Hessian and focus on $L = 1$ only. 
Here we relax the sparsity condition on the inverse Hessian 
and show how to perform inference for a general $L$.

\subsection{Related Work}
\label{sec:related}

Our work straddles two areas of statistical learning which have
attracted significant research of late: model selection and estimation
in high-dimensional graphical models, and high-dimensional inference.
We briefly review the literature most relevant to our work, and refer
the reader to two recent review articles for a comprehensive overview
\citep{Drton2016Structure,Jankova2018Inference}.
\citet{Drton2016Structure} focuses on structure learning in graphical
models, while \citet{Jankova2018Inference} reviews inference in
Gaussian graphical models.

We start by reviewing the literature on learning structure of
probabilistic graphical models.  Much of the research effort has
focused on learning structure of Gaussian graphical models where the
edge set $E$ of the graph $G$ is encoded by the non-zero elements of
the precision matrix $\Omega = \Sigma^{-1}$. The literature here
roughly splits into two categories: global and local methods.  Global
methods typically estimate the precision matrix by maximizing
regularized Gaussian log-likelihood \citep{Yuan2007Model,
  rothman08spice, Friedman2008Sparse, dAspremont2008First,
  Ravikumar2011High, fan09network, Lam2009Sparsistency}, while 
local methods estimate the graph structure by learning 
the neighborhood or Markov blanket of each node separately
 \citep{Meinshausen2006High, Yuan2010High,
  Cai2011Constrained, Liu2012TIGER, Zhao2014Calibrated}. Extensions to
more general distributions in Gaussian and elliptical families are
possible using copulas, as the graph structure within these families
is again determined by the inverse of the latent correlation matrix
\citep{Liu2009Nonparanormal:, Liu2012High, Xue2012Regularized,
  Liu2012Transelliptical, Fan2014High}.

Once we depart from the Gaussian distribution and related families,
learning the conditional independence structure becomes more
difficult, primarily owing to computational intractability of
evaluating the log-partition function.  A computationally tractable
alternative to regularized maximum likelihood estimation is
regularized pseudo-likelihood which was studied in the context of
learning structure of Ising models in
\citet{Hoefling2009Estimation}, \citet{ravikumar09high}, and
\citet{Xue2012Nonconcave}. Similar methods were developed in the study
of mixed exponential family graphical models, where a node's
conditional distribution is a member of an exponential family
distribution, such as Bernoulli, Gaussian, Poisson or
exponential. See \cite{Guo2011Joint}, \cite{Guo2011Asymptotic},
\cite{Lee2012Learning}, \cite{Cheng2013High},
\cite{Yang2012Graphical}, and \cite{Yang2014Mixed} for more details.

More recently, score matching estimators have been investigated for
learning the structure of graphical models in high-dimensions when the
normalizing constant is not available in a closed-form
\citep{Lin2015High,Yu2018Graphical}.  Score matching was first
proposed in \citet{hyvarinen2005estimation} and subsequently extended
for binary models and models with non-negative data in
\citet{Hyvaerinen2007Some}. It offers a computational advantage when
the normalization constant is not available in a closed-form, making
likelihood based approaches intractable, and is particularly appealing
for estimation in exponential families as the objective function is
quadratic in the parameters of interest. \citet{Sun2015Learning}
develop a method based on score matching for learning conditional
independence graphs underlying structured infinite-dimensional
exponential families. \citet{Forbes2013Linear} investigated the use of
score matching for the inference of Gaussian linear models in
low-dimensional settings.  However, despite its power, there have not
been results on inference in high-dimensional models using score
matching. As one of our contributions in this paper, we build on the
prior work on estimation using generalized score matching and develop
an approach to statistical inference for high-dimensional graphical
models. In particular, we construct a novel $\sqrt{n}$-consistent
estimator of parameters in~\eqref{eq:logdensity}.  This is the first
procedure that can obtain a parametric $\sqrt{n}$ rate of convergence for an edge
parameter in a graphical model where computing the normalizing
constant is intractable.

Next, we review the literature on high-dimensional inference, focusing on
work related to high-dimensional undirected graphical models.
\citet{Liu2013Gaussian} developed a procedure that estimates
conditional independence graph from Gaussian observations and controls
false discovery rates asymptotically. \citet{Wasserman2014Berry}
develop confidence guarantees for undirected graphs under minimal
assumptions by developing Berry-Esseen bounds on the accuracy of Normal
approximation. \citet{Ren2013Asymptotic},
\citet{Jankova2014Confidence}, and \citet{Jankova2017Honest} develop
methods for constructing confidence intervals for edge parameters in
Gaussian graphical models, based on the idea of debiasing the $\ell_1$
regularized estimator developed in \citep{Zhang2011Confidence,
  Geer2013asymptotically, Javanmard2013Confidence}. A related approach
was developed for edge parameters in mixed graphical models whose node
conditional distributions belong to an exponential family in
\citet{Wang2016Inference}. \citet{Wang2014Inference} develop
methodology for performing statistical inference in time-varying and
conditional Gaussian graphical models, while \citet{Barber2015ROCKET}
and \citet{Lu2015Posta} develop methods for semi-parametric copula
models.  We contribute to the literature on high dimensional inference
by demonstrating how to construct regular estimators for probabilistic
Graphical models whose normalizing constant is intractable. Our
estimators are robust to model selection mistakes and allows us to
perform valid statistical inference for edge parameters in a large
family of data generating distributions.

Finally, we contribute to the literature on simultaneous inference in
high-dimensional models. \citet{Zhang2014Simultaneous} and
\citet{Dezeure2017High} develop methods for performing simultaneous
inference on all the coefficients in a high-dimensional linear
regression. 
In the same setting,
\citet{Zhao2014General} use a multiplier bootstrap approach to construct
robust simultaneous confidence intervals.
\citet{chang2018confidence} applies it to the
simultaneous inference of Gaussian graphical models.  These procedures
allow for the dimensionality of the vector to be exponential in the
sample size and rely on bootstrap to approximate the quantile of the
test statistic. We extend these ideas to the high dimensional
graphical model setting and show how we can build simultaneous
hypothesis tests on the neighbors of a specific node.

A conference version of this paper was presented in the Annual
Conference on Neural Information Processing Systems 2016
\citep{Yu2016Statistical}. Compared to the conference version, in this
paper we extend the results in the following ways.  First, we extend
the results to include the generalized score matching method
\citep{Yu2018Graphical, Yu2018Generalized} in place of the original
score matching method.  This generalized form of the score matching
method allows us to improve the estimation accuracy and obtain better
inference results for non-negative data. In the conference
  version, we made an assumption that the inverse of the population
  Hessian matrix, see Section~\ref{sec:asympt-norm-estim}, is
  (approximately) sparse. We relax this sparsity condition and develop
  an inference procedure that is valid even if the sparsity condition
  is violated, but the inverse of the Hessian matrix has bounded
  columns in the $\ell_1$ norm.  Moreover, instead of focusing on a
single edge as in the conference version, in this work we propose a
procedure for simultaneous inference for all edges connected to a
specific node. This allows us to build hypothesis tests for a broad
class of applications, including testing of isolated nodes, support
recovery, and testing the difference between two graphical models.
Furthermore, while the conference version focused on the case
 where $L=1$ in~\eqref{eq:logdensity}, here we extend the results to
 a general choice of $L$.  Lastly, we run additional experiments to
demonstrate the effectiveness of our proposed method.

\subsection{Notation}

We use $[n]$ to denote the set $\{1,\ldots,n\}$. For a vector $a \in \RR^n$,
we let ${\rm supp}(a) = \{j\ :\ a_j \neq 0\}$ be the support set (with
an analogous definition for matrices
$A \in \mathbb{R}^{n_1\times n_2}$), $\|a\|_q$, $q \in [1,\infty)$,
the $\ell_q$-norm defined as
$\|a\|_q = (\sum_{i\in[n]} |a_i|^q)^{1/q}$ with the usual extensions
for $q \in \{0,\infty\}$, that is, $\|a\|_0 = |{\rm supp}(a)|$ and
$\|a\|_\infty = \max_{i\in[n]}|a_i|$.
For a vector $x$, $x_{M}$ is a
sub-vector of $x$ with components corresponding to the set $M$, and
$x_{-ab}$ is the sub-vector with component corresponding to edge $\{a,b\}$ omitted. 
For a matrix $A \in \RR^{m \times n}$, denote $\|A\|_q = \sup\{ \|Ax\|_q: x \in \RR^n, \|x\|_q=1\}$ as the induced $\ell_q$ norm. 
In particular, $\|A\|_{\infty }=\max _{1\leq i\leq m}\sum _{j=1}^{n}|a_{ij}|$. 
We also use $\|A\|_{\max} = \max_{jk} |a_{jk}|$ to denote the maximum component of $A$.
We define $\mathbb{E}_n$ as the empirical mean of $n$ samples:
$\mathbb{E}_n[f(x_i,\theta)] = \frac{1}{n} \sum_{i=1}^n
f(x_i,\theta)$. For two sequences of numbers $\{a_n\}_{n=1}^\infty$
and $\{b_n\}_{n=1}^\infty$, we use $a_n = \Ocal(b_n)$, or $a_n \lesssim b_n$ to denote that
$a_n \leq Cb_n$ for some finite positive constant $C$, and for all $n$
large enough. 
We use $a_n \lesssim_P b_n$ to denote that $a_n \lesssim b_n$ happens with high probability.
The notation $a_n = o(b_n)$ is used to denote that
$a_nb_n^{-1}\xrightarrow{n\rightarrow\infty} 0$.  
We denote $a_n \longrightarrow_D \mathcal A$ as convergence in distribution to a fixed distribution $\Acal$ and $a_n \longrightarrow_P a$ as convergence in probability to a constant $a$.
We denote
$a \circ b = (a_1b_1, ..., a_pb_p)$ for $a,b \in \RR^p$.  For any
function $f:\mathbb{R}^p \to \mathbb{R}$, we use
$\nabla f(x) = \cbr{ \partial/(\partial x_j) f(x) }_{j \in [p]}$ to
denote the gradient, and
$\Delta f(x) = \sum_{j\in[p]} \partial^2/(\partial x_j^2) f(x)$ to
denote the Laplacian operator on $\RR^p$.  Note that both the gradient
and the Laplacian are \emph{with respect to $x$}.

\subsection{Organization of the Paper}

The remainder of this paper is structured as follows. We begin in
Section \ref{sec:background} with background on exponential family
pairwise graphical model, score matching method, and a brief review of
statistical inference in high dimensional models.  In Section
\ref{sec:methodology} we describe the construction of our
novel estimator for a single edge parameter based on a three-step
procedure, for $L = 1$.  Section \ref{sec:asympt-norm-estim} provides theoretical
results and Section \ref{sec:relaxation_L1} discusses the relaxation of sparsity condition on the inverse of population Hessian matrix.
Section \ref{sec:simultaneous} extends the procedure to
simultaneous inference for all edges connected to some specific node.
In Section \ref{sec:general_L} we extend our results to general $L$.
We provide experimental results for synthetic datasets and a real
dataset in Sections \ref{sec:experiments_synthetic} and
\ref{sec:experiments_real} respectively.  Section \ref{sec:conclusion}
provides conclusion and discussion.

\section{Background}
\label{sec:background}

We begin with reviewing exponential family pairwise graphical models
in Section \ref{sec:ExpoGM}, and then introduce the score matching and
generalized score matching methods in Section
\ref{sec:score-matching}. Finally we provide a brief overview of
statistical inference for high dimensional models in
Section~\ref{sec:inference}.

\subsection{Exponential Family Pairwise Graphical Models}
\label{sec:ExpoGM}

Throughout the paper we focus on the case where
\[
\Pcal = \{ p_\theta(x) \mid \theta \in \Theta \}  
\]
is an exponential family with log-densities given in
\eqref{eq:logdensity}, which frequently appear in graphical
modeling. There are $K$ sets of sufficient statistics
$\{t_a^{(k)}\}_{k \in [K]}$ for each $a \in V$ that depend on the
individual nodes and $L$ sets of sufficient statistics for each $(a,b) \in {V \choose 2}$ that allow for
pairwise interactions of different types. Conditional independence
graph underlying a distribution $p_\theta \in \Pcal$ has no edge
between vertices $a$ and $b$ if and only if
$\theta_{ab}^{(1)} = \ldots = \theta_{ab}^{(L)} = 0$. A special case of the
model given in \eqref{eq:logdensity} are pairwise interaction models
with log-densities
\begin{equation}
  \label{eq:pairwise}
  \log p_\theta(x) =   
  \sum_{(a,b)\in E} \theta_{ab} t_{ab}(x_a,x_b) - \Psi(\theta) + h(x), 
  \quad x \in \Xcal \subseteq \RR^p,
\end{equation}
where $t_{ab}(x_a, x_b)$ are sufficient statistics that depend only on
$x_a$ and $x_b$.  In what follows, we will consider models that either
has the form given in~\eqref{eq:pairwise} or the more general form
given in \eqref{eq:logdensity}.

A number of well-studied distributions have the above discussed
form. We provide some examples below, including examples where the
normalizing constant $\Psi(\theta)$ cannot be obtained in closed-form.

\paragraph{\emph{Gaussian graphical models.}} The most studied example of a
probabilistic graphical model is the case of the Gaussian graphical
model. Suppose that the random variable $X$ follows the centered
multivariate Gaussian distribution with covariance $\Sigma$ and
precision matrix $\Omega = \Sigma^{-1} = (\omega_{ab})$. The
log-density is given as
\begin{equation}
  \label{eq:gaussian_density}
p(x ; \Omega) \propto \exp \cbr{ -\frac 12 x^\top \Omega x },
\end{equation}
 the support of the density is $\Xcal = \RR^{p}$ and the sufficient
statistics take the form $t_{ab}(x_a,x_b) = x_ax_b$.

\paragraph{\emph{Non-negative Gaussian.}}  Our second example of a distribution
with the log-density of the form in \eqref{eq:pairwise} is that of a
non-negative Gaussian random vector. The probability density function
of a non-negative Gaussian random vector $X$ is proportional to that
of the corresponding Gaussian vector given in
\eqref{eq:gaussian_density}, but restricted to the non-negative orthant.
Here the support of the density is $\Xcal = \RR_+^{p}$.  The
conditional independence graph is determined the same way as in the
Gaussian graphical model case through the non-zero pattern of the
elements in the precision matrix $\Omega$. The normalizing constant in
this family has no closed-form and hence maximum likelihood estimation
of $\Omega$ is intractable.

\paragraph{\emph{Normal conditionals.}} Our third example is taken from
\citet{Lin2015High}. See also \citet{AndrewGelman1991Note} and
\citet{Arnold1999Conditional}. Consider the family of distributions
with densities of the form
\begin{equation*}
  p(x ; \Theta^{(1)}, \Theta^{(2)}, \eta, \beta) \propto
  \exp \cbr{
    \sum_{a \neq b} \Theta_{ab}^{(2)}x_a^2x_b^2 +
    \sum_{a \neq b} \Theta_{ab}^{(1)}x_ax_b +
    \sum_{a \in V} \eta_a x_a^2 + 
    \sum_{a \in V} \beta_ax_a },\, x \in \RR^{p},
\end{equation*}
where the matrices $\Theta^{(1)},\Theta^{(2)} \in \RR^{p \times p}$
are symmetric interaction matrices with a zero diagonal. Members of
this family have Normal conditionals, but the densities themselves
need not be unimodal. The conditional independence graph does not
contain an edge between vertices $a$ and $b$ if and only if both
$\Omega_{ab}^{(1)}$ and $\Omega_{ab}^{(2)}$ are equal to zero. In
contrast to the Gaussian graphical models, the conditional dependence
may also express itself in the variances.

\paragraph{\emph{Conditionally specified mixed graphical models.}} In general,
specifying multivariate distributions is difficult, since in a given
problem it might not be clear what class of graphical models to use.
On the other hand, specifying univariate distributions is an easier
task.  \citet{Chen2013Selection} and \citet{Yang2013Graphical}
explored ways of specifying multivariate joint distributions via
univariate exponential families. Consider a conditional density of
the form
\begin{equation}
  \label{eq:conditional_model}
  p(x_a \mid (x_b, b\neq a) ; \theta_a) =
  \exp\cbr{f_a(x_a) + \sum_{b \neq a} \theta_{ab}B_a(x_a)B_b(x_b) - \Psi_a(\eta_a)},
  \quad x_a \in \Xcal_a,
\end{equation}
where $\eta_a = \eta_a(\theta_a, f_a, (x_b)_{b\neq a})$ and
$B_a(\cdot)$ are known functions for each $a \in V$.  Suppose that for
a random vector $X$, each coordinate $X_a$ follows the conditional
density of the form in \eqref{eq:conditional_model} with
$\theta_{ab} = \theta_{ba}$ for all $a, b \in V$. Then
\citet{Chen2013Selection} and \citet{Yang2013Graphical} showed that
there exists a joint distribution of $X$ compatible with the
conditional densities and that it is of the form
\begin{equation*}
  p(x ; \Theta) \propto \exp \cbr{
    \sum_{a \in V} f_a(x_a) + \frac12 \sum_{a\in V}\sum_{b \neq a} \theta_{ab} B_a(x_a)B_b(x_b)
  }, \quad x \in \Xcal.
\end{equation*}
In particular, the joint density above is of the form given in
\eqref{eq:logdensity}, with pairwise interaction sufficient statistics
given as $t_{ab}(x_a, x_b) = B_a(x_a)B_b(x_b)$. When the support of
the distribution is $\Xcal = \RR^p$ or $\Xcal = \RR_+^p$, the
parameters of the distribution can be efficiently estimated using
score matching. In the case of unknown function $B_a(\cdot)$,
\citet{Suggala2017Expxorcist} explored nonparametric estimation via
basis expansion and fitted parameters using
pseudo-likelihood. Developing a valid statistical inference procedure
for this nonparametric setting is beyond the scope of the current
work.

As an example of a conditionally specified model, that we will return to
later in the paper, consider exponential graphical models where the
node-conditional distributions follow an exponential distribution.
For a random vector $X$ described by an exponential graphical model, the density
function is given by
\begin{equation*}
  p(x ; \Theta) \propto \exp \cbr{ -\sum_{a \in V} \theta_ax_a - \sum_{a \neq b} \theta_{ab}x_ax_b },
  \quad x \in \RR^p_+.
\end{equation*}
Note that the variable takes only non-negative values. To ensure that
the distribution is valid and normalizable, the natural parameter
space $\Theta$ consists of matrices whose elements are
positive. Therefore, one can only model negative dependencies via the
exponential graphical model.

\paragraph{\emph{Exponential square-root graphical model. }}
As our last example, consider the exponential square-root graphical model 
\citep{Inouye2016Square} with density function given by
\begin{equation*}
  p(x ; \eta, K) \propto \exp \cbr{ -\sqrt{x}^\top K \sqrt{x} + 2\eta^\top\sqrt{x} },
  \quad x \in \RR^p_+.
\end{equation*}
This square-root graphical model is a multivariate generalizations of
univariate exponential family distributions that can capture the
positive dependency among nodes.  Specifically, it assumes only a mild
condition on the parameter matrix, but allows for almost arbitrary
negative and positive dependencies. We refer to
\citet{Inouye2016Square} for details on parameter estimation with
nodewise regressions and likelihood approximation methods.

\subsection{Score Matching}
\label{sec:score-matching}

In this section we briefly review the score matching method proposed
in \cite{hyvarinen2005estimation,Hyvaerinen2007Some} and the
generalized score matching for non-negative data proposed in
\cite{Yu2018Graphical}. 

\subsubsection{Score Matching}
\label{sec:score_matching}

A scoring rule $S(x,Q)$ is a real-valued function that quantifies the
accuracy of $Q \in \Pcal$ being the distribution from which an
observed realization $x \in \Xcal$ may have been sampled.  There are a
large number of scoring rules that correspond to different decision
problems \cite{Parry2012Proper}. Given $n$ independent realizations of
$X$, $\{x_i\}_{i\in[n]}$, one finds optimal score estimator
$\hat Q \in \Pcal$ that minimizes the empirical score
\begin{equation}
  \label{eq:score_minimization}
  \hat Q = \arg\min_{Q \in \Pcal} \EE_n\sbr{ S(x_i, Q) }.
\end{equation}

When $\Xcal = \RR^p$ and $\Pcal$ consists of twice differentiable
densities with respect to Lebesgue measure, the Hyv\"arinen scoring
rule \citep{hyvarinen2005estimation} is given as
\begin{equation}
  \label{eq:score_matching}
  S(x, Q) = \frac 12 \big\| \nabla \log q(x) \big\|_2^2 + \Delta \log q(x), 
\end{equation}
where $q$ is the density of $Q$ with respect to Lebesgue measure on
$\Xcal$.  We would like to emphasize that this gradient and Laplacian
are \emph{with respect to $x$}. In this way we get rid of the
normalizing constant which does not depend on $x$.  This scoring rule
is convenient for learning models that are specified in an
unnormalized fashion or whose normalizing constant is difficult to
compute. The score matching rule is proper \citep{Dawid2007geometry},
that is, $\EE_{X \sim P} S(X, Q)$ is minimized over $\Pcal$ at
$Q = P$. Suppose the density $q$ of $Q \in \Pcal$ is twice
continuously differentiable and satisfies
\[
  \EE_{X \sim P} \norm{\nabla \log q(X)}_2^2 < \infty, \qquad
  \text{for all $P, Q \in \Pcal$}
\]
and 
\[
  q(x) \text{ and } \norm{\nabla q(x)}_2 \text{ tend to zero as $x$
    approaches the boundary of $\Xcal$ }.
\]
Then the Fisher divergence between $P, Q \in \Pcal$,
\[
  D(P, Q) = \int p(x) \|\nabla \log q(x) - \nabla \log p(x)\|_2^2 d x,
\]
where $p$ is the density of $P$, is induced by the score matching rule
\citep{hyvarinen2005estimation}.
The gradients in the equation above can be thought of as gradients
with respect to a hypothetical location parameter, evaluated at the
origin \citep{hyvarinen2005estimation}.

For a parametric exponential family
$\Pcal = \{p_\theta \mid \theta \in \Theta\}$ with densities given in
\eqref{eq:logdensity}, minimizing \eqref{eq:score_minimization} with
the scoring rule in \eqref{eq:score_matching} can be done in a closed
form \citep{hyvarinen2005estimation,Forbes2013Linear}.  An estimator
$\hat \theta$ obtained in this way can be shown to be asymptotically
consistent \citep{hyvarinen2005estimation}, however, in general it
will not be efficient \citep{Forbes2013Linear}.

\subsubsection{Generalized Score Matching for Non-Negative Data}

The score matching method in Section \ref{sec:score_matching} does not work for non-negative data, since the assumption that $q(x)$ and $||\nabla q(x)||_2$ tend to 0 at the boundary breaks down. 
To solve this problem, \citet{Hyvaerinen2007Some} proposed a generalization of the score
matching approach to the case of non-negative data.

When
$\Xcal = \RR^p_+$ the non-negative score matching loss (analogous to the Fisher divergence $D(P,Q)$) is defined as
\begin{equation*}
\label{eq:score_matching_nonnegative}
J_+(P,Q) = \int_{\RR^p_+} p(x) \cdot \big\| \nabla\log p(x) \circ x - \nabla\log q(x) \circ x \big\|_2^2 dx.
\end{equation*}
The scoring rule for non-negative data that induces $J_+(P,Q)$ is given as
\begin{equation}
  \label{eq:score_matching_nonnegative:sample}
S_+(x, Q) = {\sum_{a \in V} \left[2x_a \frac{\partial \log q(x)}{\partial x_a} + x_a^2\frac{\partial^2 \log q(x)}{\partial x_a^2} + \frac{1}{2}x_a^2 \left(\frac{\partial \log q(x)}{\partial x_a}\right)^2 \right]}.
\end{equation}
For exponential families, the non-negative score matching loss again
can be obtained in a closed form and the estimator is consistent and
asymptotically normal under suitable conditions \citep{Hyvaerinen2007Some}.

\citet{Yu2018Graphical} proposed the generalized score matching for
non-negative data to improve the estimation efficiency of the
procedure based on the scoring rule in
\eqref{eq:score_matching_nonnegative:sample}.  Let
$\ell_1, ..., \ell_p: \RR_+ \to \RR_+$ be positive and differentiable
functions and set
\[
  \ell(x) = \big( \ell_1(x_1), \ldots, \ell_p(x_p) \big).
\]
The
generalized $\ell$-score matching loss is defined as
\begin{equation*}  
J_\ell(P,Q) = \int_{\RR^p_+} p(x) \cdot \big\| \nabla\log p(x) \circ \ell^{1/2}(x) - \nabla\log q(x) \circ \ell^{1/2}(x) \big\|_2^2 dx,
\end{equation*}
where $\ell^{1/2}(x) = \big( \ell_1^{1/2}(x_1), \ldots, \ell_p^{1/2}(x_p) \big)$. 
Suppose the following regularity conditions are satisfied
\begin{equation}
  \label{eq:condition:partial_integration}
\begin{aligned}
\lim_{x_j \to \infty} p(x) \ell_j(x_j) \nabla_j \log q(x) = 0 ~~~~ \forall x_{-j} \in \RR^{p-1}_{+}, ~\forall p \in \Pcal_{+}, \\
\lim_{x_j \to 0} p(x) \ell_j(x_j) \nabla_j \log q(x) = 0 ~~~~ \forall x_{-j} \in \RR^{p-1}_{+}, ~\forall p \in \Pcal_{+}, \\
\EE_{X \sim \Pcal_{+}} \Big[ \| \nabla\log q(X) \circ \ell^{1/2}(X) \|_2^2 \Big] < +\infty, \\
\EE_{X \sim \Pcal_{+}} \Big[ \| (\nabla\log q(X) \circ \ell(X))' \|_1 \Big] < +\infty.
\end{aligned}
\end{equation}
Under the condition~\eqref{eq:condition:partial_integration}, the
scoring rule corresponding to the generalized $\ell$-score matching
loss is given as
\begin{equation*}
\label{eq:score_matching_nonnegative:sample_l}
  S_\ell(x, Q) = {\sum_{a \in V} \left[ \ell'_a(x_a) \frac{\partial \log q(x)}{\partial x_a} + \ell_a(x_a)\frac{\partial^2 \log q(x)}{\partial x_a^2} + \frac{1}{2}\ell_a(x_a) \left(\frac{\partial \log q(x)}{\partial x_a}\right)^2 \right]}.
\end{equation*}
The regularity condition~\eqref{eq:condition:partial_integration} is
required for applying integration by parts and Fubini-Tonelli theorem in order to show consistency of the score-matching estimator.

Note that by choosing $\ell_j(x) = x^2$, for all $j$, one recovers the
original score matching formulas for non-negative data
in~\eqref{eq:score_matching_nonnegative:sample}.  The advantage of
this generalized score matching rule is that by choosing an
increasing, but slowly growing $\ell(x)$ (for example,
$\ell(x) = \log(x+1)$), one does not need to estimate high moments of
the underlying distribution, which leads to better practical
performance and improved theoretical guarantees. See
\citet{Yu2018Graphical} for details.

\subsubsection{Score matching for probabilistic graphical models} 

Score matching has been successfully applied in the context of
probabilistic graphical models. \citet{Forbes2013Linear} studied score
matching to learn Gaussian graphical models with symmetry
constraints. \citet{Lin2015High} proposed a regularized score matching
procedure to learn conditional independence graph in a
high-dimensional setting by minimizing
\[
  \EE_n\sbr{\overline S(x_i, \theta)} + \lambda \|\theta\|_1,
\]
where the loss function $\overline S(x_i, \theta)$ is either
$S(x_i, Q_\theta)$ defined in \eqref{eq:score_matching} or
$S_+(x_i, Q_\theta)$ defined in
\eqref{eq:score_matching_nonnegative:sample}. For Gaussian models,
$\ell_1$-norm regularized score matching  is a simple, yet efficient
method, which coincides with the method in \cite{Liu2015Fast}.
\citet{Yu2018Graphical} improved on the approach of
\citet{Lin2015High} and studied regularized generalized $\ell$-score
matching of the form
\[
  \EE_n\sbr{S_\ell(x_i, Q_\theta)} + \lambda \|\theta\|_1.
\]
Applied to data generated from a multivariate truncated normal
distribution, the conditional independence graph can be recovered with
the same number of samples that are needed for recovery of the
structure of a Gaussian graphical model.  \citet{Sun2015Learning}
develop a score matching estimator for learning the structure of
nonparametric probabilistic graphical models, extending the work on
estimation of infinite-dimensional exponential families
\citep{Sriperumbudur2013Density}.  In Section \ref{sec:methodology},
we present a new estimator for components of $\theta$
in~\eqref{eq:logdensity} that is consistent and asymptotically normal,
building on \citet{Lin2015High} and \citet{Yu2018Graphical}.

\subsection{Statistical Inference}
\label{sec:inference}

We briefly review how to perform statistical inference for low
dimensional parameters in a high-dimensional model.  In many
statistical problems, the unknown parameter $\beta \in \RR^p$ can be
partitioned as $\beta = (\alpha, \eta)$, where $\alpha$ is a scalar of
interest and $\eta$ is a $(p - 1)$ dimensional nuisance parameter. Let
$\beta^* = (\alpha^*, \eta^*)$ denote the true unknown parameter.  In
a high-dimensional setting, where the sample size $n$ is much smaller
than the dimensionality $p$ of the parameter $\beta$, it is common to
impose structural assumptions on $\beta^*$. For example in several
applications, it is common to assume that the true parameter $\beta^*$
is sparse. Indeed, we will work under this assumption as well.

Let us denote the empirical negative log-likelihood  by
$$
\Lcal(\beta) = \frac 1n \sum_{i=1}^n \Lcal_i(\beta), 
$$
where $\Lcal_i(\beta)$ is the negative log-likelihood for the
$i^{th}$ observation. Let
$I = \EE\sbr{\nabla^2 \Lcal(\beta)}$ denote the information
matrix and denote
the partition of $I$ corresponding to $\beta = (\alpha, \eta)$ as
\begin{equation}
I = 
 \begin{pmatrix}
  I_{\alpha\alpha} & I_{\alpha\eta} \\
  I_{\eta\alpha} & I_{\eta\eta}
 \end{pmatrix}.
\end{equation}
The partial information matrix
of $\alpha$ is denoted as
$I_{\alpha|\eta} = I_{\alpha\alpha} - I_{\alpha\eta} I_{\eta\eta}^{-1}
I_{\eta\alpha}$.

Consider for the moment a low-dimensional setting.  In order to
perform statistical inference about $\alpha^*$, one can use the
{\it profile partial score function} defined as 
\[
U(\alpha) = \nabla_\alpha \Lcal\big( \alpha, \hat\eta(\alpha) \big), 
\]
where $\hat\eta(\alpha) = \arg\min_\eta \Lcal(\alpha, \eta)$
is the maximum partial likelihood estimator for $\eta$ with a fixed
parameter $\alpha$.  Under the null hypothesis that
$\alpha^* = \alpha^0$, we have that \citep{Vaart1998Asymptotic}
\[
\sqrt n U\rbr{\alpha^0} \longrightarrow_D N(0, I_{\alpha|\eta}^*).
\]
Therefore, one can reject the null hypothesis for large values of
$U\rbr{\alpha^0}$. However, in a high-dimensional setting, the
estimator $\hat\eta(\alpha)$ is no longer $\sqrt{n}$-consistent and we
have to modify the approach above. In particular, we will show how to
modify the profile partial score function to allow for valid inference
in a high-dimensional setting based on a sparse estimator of
$\hat\eta(\alpha)$.

Without loss of generality, assume that $\alpha^0 = 0$.  For any
estimator $\tilde \eta$, Taylor's expansion theorem gives
\begin{equation}
\label{eq:Taylor}
\sqrt n \nabla_\alpha \Lcal(0, \tilde \eta) =
\sqrt n \nabla_\alpha \Lcal(0, \eta^*) + \sqrt n \nabla_{\alpha\eta} \Lcal(0, \eta^*) \cdot (\tilde \eta - \eta^*) + \textsf{rem}, 
\end{equation}
where \textsf{rem} is the remainder $o(\tilde \eta-\eta^*)$ term. The first term
$\sqrt n \nabla_\alpha \Lcal(0, \eta^*)$ in \eqref{eq:Taylor}
converges to a normal distribution under suitable assumptions using
the central limit theorem (CLT).  The distribution of the second term,
however, is in general intractable to obtain. This is due to the fact
that the distribution of $\tilde \eta$ depends on the selected model.
Unless we are willing to assume stringent and untestable conditions
under which it is possible to show that the true model can be
selected, the limiting distribution of $\tilde \eta$ cannot be
estimated even asymptotically \citep{Leeb2007Can}.  To overcome this
issue, one needs to modify the profile partial score function, so that
its limiting distribution does not depend on the way the nuisance
parameter is estimated.

\cite{Ning2014General} introduced the following decorrelated score function
\begin{equation*}
  U(\alpha, \eta) =
  \nabla_\alpha \Lcal(\alpha, \eta)
  - w^T \nabla_\eta \Lcal(\alpha, \eta), 
\end{equation*}
where $w = I_{\alpha\eta}I_{\eta\eta}^{-1}$.  The decorrelated score
function $U(\alpha, \eta)$ is uncorrelated with the nuisance score
functions $\nabla_\eta \Lcal(\alpha, \eta)$ and, therefore,
its limiting distribution will not depend on the model selection
mistakes incurred while estimating $\eta^*$. In particular,
$U(\alpha^0, \tilde \eta)$ is indeed asymptotically normally
distributed under the null hypothesis, as long as $\tilde \eta$ is a
good enough estimator of $\eta^*$, but not necessarily
$\sqrt{n}$-consistent estimator. Based on the asymptotic normality of
the decorrelated score function, we can then build confidence
intervals for $\alpha^*$ and perform hypothesis testing.

In practice, the vector $w$ is unknown and needs to be estimated.  A
number of methods have been proposed for its estimation in the
literature.  For example, \citet{Ning2014General} use a Dantzig
selector-like method, \citet{Belloni2012Inference} proposed the double
selection method, while \citet{Zhang2011Confidence},
\citet{Geer2013asymptotically}, and \citet{Javanmard2013Confidence}
use a lasso based estimator. 
See also \citet{Dezeure2017High}, \citet{Zhang2014Simultaneous} for simultaneous inference, \citet{Taylor2014Exact}, \citet{yang2016selective} for post selective inference, \citet{li2019statistical}, \citet{cao2019estimation}, and \citet{cao2019synthetic} for for synthetic control, etc.
In this paper, we adopt the double
selection procedure of \citet{Belloni2012Inference}. Details will be
given in Section \ref{sec:methodology}.

\section{Methodology}
\label{sec:methodology}

In this section, we present a new procedure that constructs a
$\sqrt{n}$-consistent estimator of an element $\theta_{ab}$ of
$\theta$. Our procedure involves three steps that we detail below.  We
start by introducing some additional notation and then describe the
procedure for the case where $\Xcal = \RR^p$. Extension to
non-negative data is given at the end of the section.
Throughout this section we consider $L = 1$ only, so that the
  parameter of interest $\theta_{ab}$ is a scalar. Extensions to
  general $L$ is discussed later in Section \ref{sec:general_L}.

For fixed indices $a, b \in [p]$, let
\[
  q^{ab}_\theta(x) := q^{ab}_\theta(x_a, x_b \mid x_{-ab})
\]
be the conditional density of $(X_a, X_b)$ given $X_{-ab} =
x_{-ab}$. In particular,
\begin{equation}
\label{eq:conditional_log_density}
  \log q^{ab}_\theta(x) = \dotp{\theta^{ab}}{\vpx} 
  - \Psi^{ab}(\theta, x_{-ab})  
  + h^{ab}(x),
\end{equation}
where $\theta^{ab} \in \RR^{s'}$, with $s' = 2K+2p-3$, is the part
of the vector $\theta$ corresponding to
$\left\{ \theta_{a}^{(k)}, \theta_b^{(k)} \right\}_{k\in[K]}$,
$\cbr{\theta_{ac},\theta_{bc}}_{c\in-ab}$, and $\theta_{ab}$; and
$\vpx = \varphi^{ab}(x) \in \RR^{s'}$ is the corresponding vector of
sufficient statistics $\left\{t_a^{(k)}(x_a), t_b^{(k)}(x_b)\right\}_{k\in[K]}$, $\left\{ t_{ac}(x_a,x_c), t_{bc}(x_b,x_c) \right\}_{c\in-ab}$, and $t_{ab}(x_a,x_b)$.  Here $\Psi^{ab}(\theta, x_{-ab})$ is the
log-partition function of the conditional distribution and
$h^{ab}(x) = h_a(x_a) + h_b(x_b)$.  
Let $\nabla_{ab}$ and $\Delta_{ab}$ be the gradient and Laplacian operators,
respectively, with respect to $x_a$ and $x_b$ defined as:
\begin{align*}
\nabla_{ab} f(x) & = \Big( (\partial/\partial x_a) f(x),
  (\partial/\partial x_b) f(x) \Big)^\top \in \RR^2, \\
 \Delta_{ab} f(x) & = \Big( (\partial^2/ \partial x_a^2) +
  (\partial^2/ \partial x_b^2) \Big) f(x).  
\end{align*}

With this notation, we introduce the following scoring rule
\begin{equation}
\begin{aligned}
  \label{eq:conditional_score}
  S^{ab}(x, \theta) 
  =  
    \frac 12    \big\|\nabla_{ab} \log q^{ab}_{\theta}(x) \big\|_2^2 
    + \Delta_{ab} \log q^{ab}_{\theta}(x) 
  = \frac 12 \theta^\top \Gamma(x) \theta + \theta^\top g(x) + c(x),
\end{aligned}
\end{equation}
where the constant term $c(x) = \frac{1}{2} \|\nabla h^{ab}(x) \|^2 + \Delta h^{ab}(x) $, and
\begin{align*}
  \Gamma(x) = \vpax\vpax^\top + \vpbx\vpbx^\top \quad\text{ and } \quad
  g(x) = \vpax h_1^{ab}(x) + \vpbx h_2^{ab}(x) + \Delta_{ab}\vpx
\end{align*}
with $\vpa = (\partial/\partial x_a) \varphi$,
$\vpb = (\partial/\partial x_b) \varphi$,
$h_1^{ab} = (\partial/\partial x_a) h^{ab}$, and
$h_2^{ab} = (\partial/\partial x_b) h^{ab}$.

This scoring rule is related to the one in \eqref{eq:score_matching},
however, rather than using the density $q_\theta$ in evaluating the
parameter vector, we only consider the conditional density
$q_\theta^{ab}$.  We will use this conditional scoring rule to create
an asymptotically normal estimator of an element $\theta_{ab}$. Our
motivation for using this estimator comes from the fact that the
parameter $\theta_{ab}$ can be identified from the conditional
distribution of $(X_a, X_b) \mid X_{M_{ab}}$ where
\[
  M_{ab} := \{ c \mid (a,c) \in E \text{ or } (b,c) \in E \}
\]
is the Markov blanket of $(X_a, X_b)$. Furthermore, the optimization
problems arising in steps 1-3 below can be solved much more
efficiently, as the scoring rule in~\eqref{eq:conditional_score}
involves fewer parameters.

We are now ready to describe our procedure for estimating $\theta_{ab}$,
which proceeds in three steps. 

\paragraph{\emph{Step 1:}} We find a pilot estimator of $\theta^{ab}$ by solving the 
following program 
\begin{equation}
  \label{eq:estimation}
\begin{aligned}
  \hat \theta^{ab} 
  & =  \arg\min_{\theta \in \RR^{s'}} \ \EE_n\sbr{S^{ab}(x_i, \theta)} +  \lambda_{1} \norm{\theta}_1,
\end{aligned}  
\end{equation}
where $\lambda_1$ is a tuning parameter. Let
$\hat M_1 = {\rm supp}(\hat \theta^{ab}) := \{ (c,d) \mid \hat \theta^{ab}_{cd}
\neq 0 \} $.

Since we are after an asymptotically normal estimator of
$\theta_{ab}$, one may think that it is sufficient to find
$\tilde \theta^{ab} = \arg\min\{\EE_n\sbr{S^{ab}(x_i, \theta)} \mid
{\rm supp}(\theta) \subseteq \hat M_1 \}$ and appeal to results of
\citet{Portnoy1988Asymptotic}, who has established asymptotic normality 
for $M$-estimators with increasing number of parameters. Unfortunately, this is not the
case. Since $\tilde \theta$ is obtained via a model selection
procedure, it is irregular and its asymptotic distribution cannot be
estimated \citep{Leeb2007Can,Poetscher2009Confidence}. Therefore, we
proceed to create a regular estimator of $\theta_{ab}$ in steps 2 and
3. The idea is to create an estimator $\tilde \theta_{ab}$ that is
insensitive to first-order perturbations of other components of
$\tilde \theta^{ab}$, which we consider as nuisance components. The
idea of creating an estimator that is robust to perturbations of
nuisance has been recently used in \citet{Belloni2012Inference},
however, the approach goes back to the work of
\citet{Neym1959Optimal}.

\paragraph{\emph{Step 2:}} Let $\hat \gamma^{ab}$ be a minimizer of 
\begin{equation}
  \label{eq:estimation:step2}
  \begin{aligned} 
     \frac 12 \EE_n[
      (\vpaxi[,ab]-\vpaxi[,-ab]^\top \gamma)^2 + 
      (\vpbxi[,ab]-\vpbxi[,-ab]^\top \gamma)^2 
    ]
    + \lambda_{2} \norm{\gamma}_1,
  \end{aligned}
\end{equation}
where $\lambda_2$ is a tuning parameter. Let
$\hat M_2 = {\rm supp}(\hat \gamma^{ab}) := \{ (c,d) \mid \hat
\gamma^{ab}_{cd} \neq 0 \} $.  The intuition here is that the vector
$(1, -\hat \gamma^{ab,\top})^\top $ approximately computes a row, up
to a constant, of the inverse of the Hessian in \eqref{eq:estimation}.

\paragraph{\emph{Step 3:}}
Let $\tilde M = \{(a,b)\} \cup \hat M_1 \cup \hat M_2$.
We obtain our estimator as a solution to the following 
program
\begin{equation}
  \label{eq:estimation:step3}
\begin{aligned}
  \tilde \theta^{ab} 
  & = \arg\min_\theta \ \EE_n\sbr{ S^{ab}(x_i, \theta) }
  \qquad \text{s.t.}\quad {\rm supp}(\theta) \subseteq \tilde M.
\end{aligned}  
\end{equation}
Our estimator of $\theta_{ab}$ is the coordinate $ab$ of
$\tilde \theta^{ab}$---which we denote as $\tilde \theta_{ab}$.  Motivation for this
procedure will be clear from the proof of Theorem~\ref{thm:main} given
in the next section.

\paragraph{\emph{Extension to non-negative data.}}

For non-negative data, the procedure is similar. In place of the score
rule in \eqref{eq:conditional_score}, we will use a conditional score
rule based on the generalized $\ell$-score rule. We define the
following scoring rule
\begin{equation}
\label{eq:S_l_nonnegative}
S^{ab}_\ell(x,\theta) = \frac{1}{2} \theta^\top  \Gamma_\ell(x) \theta +
\theta^\top  g_\ell(x)
\end{equation}
with
$$ 
\Gamma_\ell(x) = \ell_a(x_a) \cdot \vpax\vpax^\top  + \ell_b(x_b) \cdot \vpbx\vpbx^\top  
$$
and
\begin{equation*}
\begin{aligned}
 g_\ell(x) = 
 \ell_a(x_a) \vpax h_1^{ab}(x) + \ell_b(x_b) \vpbx h_2^{ab}(x) &+ \ell_a(x_a) \varphi_{11}(x) + \ell_b(x_b) \varphi_{22}(x) \\
 &\qquad + \ell'_a(x_a) \vpax + \ell'_b(x_b) \vpbx.
\end{aligned}
\end{equation*}
Here $\varphi_{11} = (\partial^2/\partial x_a^2) \varphi$, and
$\varphi_{22} = (\partial^2/\partial x_b^2) \varphi$. Now we can
define 
\begin{equation}
\label{eq:tilde_phi_nonnegative}
\tilde \varphi_1 = \ell^{1/2}_a(x_a) \vpa \quad \text{and}\quad \tilde \varphi_2 = \ell^{1/2}_b(x_b) \vpb.
\end{equation}
Then
$\Gamma_\ell(x) = \tilde \varphi_1(x)\tilde \varphi_1(x)^\top + \tilde
\varphi_2(x)\tilde \varphi_2(x)^\top $, which is of the same form as
\eqref{eq:conditional_score} with $\tilde \varphi_1$ and
$\tilde \varphi_2$ replacing $\vpa$ and $\vpb$, respectively.  Thus
our three-step procedure for non-negative data can be written as
follows.  For notation consistency, we omit the subscript $\ell$ on
the estimator $\theta$ and support $M$.

\paragraph{\emph{Step 1:}} We find a pilot estimator of $\theta^{ab}$ by solving 
\begin{equation}
  \label{eq:estimation_nonnegative}
\begin{aligned}
  \hat \theta^{ab} 
  & =  \arg\min_{\theta \in \RR^{s'}} \ \EE_n\sbr{S_\ell^{ab}(x_i, \theta)} +  \lambda_{1} \norm{\theta}_1,
\end{aligned}  
\end{equation}
where $\lambda_1$ is a tuning parameter and $S_\ell^{ab}$ is defined in \eqref{eq:S_l_nonnegative}. Let
$\hat M_1 = {\rm supp}(\hat \theta^{ab})$.

\paragraph{\emph{Step 2:}} Let $\hat \gamma^{ab}$ be a minimizer of 
\begin{equation}
  \label{eq:estimation:step2_nonnegative}
  \begin{aligned} 
     \frac 12 \EE_n \big[
      (\tilde \varphi_{1,ab}(x_i) - \tilde \varphi_{1,-ab}(x_i)^\top \gamma)^2 + 
      (\tilde \varphi_{2,ab}(x_i) - \tilde \varphi_{2,-ab}(x_i)^\top \gamma)^2 
    \big]
    + \lambda_{2} \norm{\gamma}_1,
  \end{aligned}
\end{equation}
where $\lambda_2$ is a tuning parameter and $\tilde\varphi_1, \tilde\varphi_2$ are defined in \eqref{eq:tilde_phi_nonnegative}. Let
$\hat M_2 = {\rm supp}(\hat \gamma^{ab})$.

\paragraph{\emph{Step 3:}}
Let $\tilde M = \{(a,b)\} \cup \hat M_1 \cup \hat M_2$.
We obtain our estimator as a solution to the following 
program
\begin{equation}
  \label{eq:estimation:step3_nonnegative}
\begin{aligned}
  \tilde \theta^{ab} 
  & = \arg\min_\theta \ \EE_n\sbr{ S_\ell^{ab}(x_i, \theta) }
  \qquad \text{s.t.}\quad {\rm supp}(\theta) \subseteq \tilde M.
\end{aligned}  
\end{equation}
Our estimator of $\theta_{ab}$ is the coordinate $ab$ of
$\tilde \theta^{ab}$---which we denote as $\tilde \theta_{ab}$.

\section{Asymptotic Normality of the Estimator}
\label{sec:asympt-norm-estim}

In this section, we outline the main theoretical properties of our
estimator.  We start by providing high-level conditions that allow us
to establish properties of each step in the procedure.

\paragraph{\emph{Assumption {\bf M}.}} We are given $n$ i.i.d. samples
$\{x_i\}_{i\in[n]}$ from $p_{\theta^*}$ of the form in
\eqref{eq:logdensity}.  Let
\begin{equation}
  \label{eq:estimation:step2:true}
  \begin{aligned}
    \gamma^{ab,*} &= \arg\min_{\gamma} \
    \EE[
      (\vpaxi[,ab]-\vpaxi[,-ab]^\top \gamma)^2 + 
      (\vpbxi[,ab]-\vpbxi[,-ab]^\top \gamma)^2 
    ]
  \end{aligned}
\end{equation}
and 
$$
\eta_{1i} = \vpaxi[,ab]-\vpaxi[,-ab]^\top \gamma^{ab,*} ~~~\text{and}~~~ \eta_{2i} = \vpbxi[,ab]-\vpbxi[,-ab]^\top \gamma^{ab,*} ~~~\text{for}~ i \in [n].
$$
We assume that the parameter vector $\theta^*$ is sparse
with $|{\rm supp}(\theta^{ab,*})| \ll n$; and the vector $\gamma^{ab,*}$ is sparse with
$|{\rm supp}(\gamma^{ab,*})| \ll n$.  

Let $m = |{\rm supp}(\theta^{ab,*})| \vee |{\rm supp}(\gamma^{ab,*})|$.
The assumption {\bf M} supposes that the parameter to be estimated is
sparse, which makes estimation in the high-dimensional setting
feasible. An extension to the approximately sparse parameter is
possible but technically cumbersome, and does not provide additional
insights into the problem. One of the benefits of using the
conditional score to learn parameters of the model is that the sample
size will only depend on the size of ${\rm supp}(\theta^{ab,*})$ and
not on the sparsity of the whole vector $\theta^*$ as in
\cite{Lin2015High}. The second part of the assumption states that the
inverse of the population Hessian is approximately sparse, which is a
reasonable assumption for a number of models, since the Markov
blanket of $(X_a, X_b)$ is small under the sparsity assumption on
$\theta^{ab,*}$. We relax the sparsity assumption in Section~\ref{sec:relaxation_L1}.

The vector $\gamma^{ab,*}$ is determined by the model
  \eqref{eq:conditional_log_density} and parameter $\theta^*$, and is
  therefore not a free parameter.  For the Gaussian graphical model,
  it can be shown that the sparsity of $\theta^{ab,*}$ implies the
  sparsity of $\gamma^{ab,*}$. That is, assumption {\bf M} holds when
  the columns of the precision matrix are sparse.  For a general
  model, it may not be easy to explicitly verify the exact sparsity of
  $\gamma^{ab,*}$, since the calculation of $\gamma^{ab,*}$ involves
  calculation of possibly intractable moments, especially when using
  generalized score matching with $\ell(x) = \log(x+1)$ for
  non-negative data.  For normal conditionals and exponential
  graphical model, we verify numerically (in Section
  \ref{sec:experiments_synthetic}) that the sample version of
  $\gamma^{ab,*}$ behaves approximately like a sparse vector when $n$
  is large enough.  These indicate that assumption {\bf M} is
  reasonable, at least in an approximately sparse version.  For
  general models, the sparsity condition on $\gamma^{ab,*}$ could be
  violated and, therefore, we discuss how to relax it in Section
  \ref{sec:relaxation_L1}.

\vspace{1mm}
Our next condition assumes that the Hessian in \eqref{eq:estimation}
and \eqref{eq:estimation:step2} is well conditioned.

\paragraph{\emph{Assumption {\bf SE}.}} Let
\[
\phi_{-}(s, A) = \inf\cbr{ \delta^\top  A \delta / \norm{\delta}_2^2
  \mid 1 \leq \norm{\delta}_0 \leq s}
\]
and
\[
\phi_{+}(s, A) = \sup\cbr{ \delta^\top  A \delta / \norm{\delta}_2^2
  \mid 1 \leq \norm{\delta}_0 \leq s}
\]
denote the minimal and maximal $s$-sparse eigenvalues of a
semi-definite matrix $A$, respectively. We assume
\[
\phi_{\min} 
  \leq \phi_{-}(m \cdot \log n, \EE\sbr{\Gamma(x_i)}) 
  \leq \phi_{+}(m \cdot \log n, \EE\sbr{\Gamma(x_i)})
  \leq \phi_{\max},
\]
where $0 < \phi_{\min} \leq \phi_{\max} < \infty$.

Assumption {\bf SE} imposes the sparse eigenvalue condition on
  the population quantity.  A lower bound on the population Hessian is
  required even in a low dimensional setting in order to prove
  asymptotic normality of an estimator.  See, for example,
  \citet[][]{Forbes2013Linear} where the population Hessian is assumed
  to be invertible. An upper bound on the Hessian matrix is also
  commonly assumed in the literature on graphical models and
  high-dimensional inference \citep[see, for
  example,][]{Yang2013Graphical,Belloni2013Least}.  We use the upper
  bound on the Hessian to control the size of the estimated support in
  steps 1 and 2 of the procedure.

For Gaussian graphical model, assumption {\bf SE} is satisfied
  with non-degenerate covariance matrix.  For general models,
  assumption {\bf SE} puts restrictions on the model parameter in a
  way that is hard to handle explicitly.  Note that related work
  imposes stronger assumption on the sample Fisher information matrix
  directly. See, for example, conditions (C1) and (C2) in \citet{Yang2013Graphical}. 

For the upper bound of the sparse eigenvalue, we remark that
  the mean of $\varphi(x)$ could be non-zero. For the Gaussian
  graphical model, if there is a non-zero mean $\mu$, then the
  components of $\varphi_1(x)$ and $\varphi_2(x)$ would instead be
  $x - \mu$. Therefore the sparse eigenvalue would not explode.  In
  practice, we subtract the empirical mean and only need to consider the
  centered case.  For other models, existing works assume boundedness
  of the first and second order moments of all the components of $x$.
  See Condition (C3) in \citet{Yang2013Graphical}.

 With assumption {\bf SE} on the population quantity, the
  following lemma, adopted from Corollary 4 in
  \cite{Belloni2013Least}, quantifies the sparse eigenvalues of the
  sample quantity $\EE_n\sbr{\Gamma(x_i)}$.

\begin{lemma}
\label{lemma:sample_sparse_eigenvalue}
Suppose assumption {\bf SE} is satisfied. Suppose there exist $K_n$ such that $\varphi_1(x_i)$ and $\varphi_2(x_i)$ are bounded: $\sup_i \| \varphi_1(x_i) \|_\infty \leq K_n$ and $\sup_i \| \varphi_2(x_i) \|_\infty \leq K_n$ a.s. If the sample size satisfies 
\begin{equation*}
K_n^2 \cdot m\log{p} \cdot \log^2(m\log{p}) \cdot \log{n} \cdot \log{(p \vee n)} = o(n \phi_{\min}^2 / \phi_{\max}),
\end{equation*}
then the event 
\begin{equation*}
\Ecal_{\rm SE} = \cbr{
\frac{\phi_{\min}}{2}
  \leq \phi_{-}\big(m \cdot \log n, \EE_n\sbr{\Gamma(x_i)}\big) 
  \leq \phi_{+}\big(m \cdot \log n, \EE_n\sbr{\Gamma(x_i)}\big)
  \leq 2 \phi_{\max} }
\end{equation*}
holds with probability at least $1-o(1)$.
\end{lemma}

 Lemma \ref{lemma:sample_sparse_eigenvalue} ensures that the
  sparse eigenvalues of the sample quantity $\EE_n\sbr{\Gamma(x_i)}$
  are well-behaved provided that $\varphi_1(x_i)$ and $\varphi_2(x_i)$
  can be upper bounded, and the sample size is reasonably large. The
  scale of the upper bound $K_n$ depends on the sufficient statistics
  $\varphi(x)$, and can be verified for concrete models.  For example,
  for the Gaussian graphical model, a standard result on the Gaussian
  tail bound gives $K_n = C\cdot(\log n + \log p)^{1/2}$ with high
  probability.  As another example, Proposition 4 in
  \cite{Yang2013Graphical} shows that, under mild conditions,
  $K_n = C\cdot(\log n + \log p)$ with high probability when the
  sufficient statistics of the conditional density are given by
  $x_a, x_b$ and $x_{a}x_{b}$, which includes a wide range of
  applications, such as exponential graphical model, and Poisson
  graphical model.  For models with more general sufficient
  statistics, we can modify the proof of Proposition 4 in
  \cite{Yang2013Graphical} to obtain the corresponding rate on $K_n$,
  under suitable assumptions.

\vspace{2mm}
Let $r_{j\theta} = \norm{\hat \theta^{ab} - \theta^{ab,*}}_j$ and
$r_{j\gamma} = \norm{\hat \gamma^{ab} - \gamma^{ab,*}}_j$, for
$j\in\{1,2\}$, be the rates of estimation in steps 1 and 2,
respectively. Under the assumption {\bf SE}, on the event
$$
 \Ecal_{\theta} = \cbr{ \norm{\EE_n\sbr{\Gamma(x_i) \theta^{ab,*} + g(x_i)}
  }_\infty \leq \frac{\lambda_1}{2} },
$$
we have that $r_{1\theta} \lesssim m\lambda_1/\phi_{\min}$ and
$r_{2\theta} \lesssim c_2\sqrt{m}\lambda_1/\phi_{\min}$. 
Similarly, on the
event
$$
 \Ecal_{\gamma} = \cbr{ \norm{\EE_n\sbr{
      \eta_{1i}\vpaxi[,-ab]+\eta_{2i}\vpbxi[,-ab] }}_\infty \leq
  \frac{\lambda_2}{2} },
$$
we have that $r_{1\gamma} \lesssim m\lambda_2/\phi_{\min}$ and
$r_{2\gamma} \lesssim \sqrt{m}\lambda_2/\phi_{\min}$, using results of
\cite{negahban2010unified}.   In order to ensure that
  $\Ecal_{\theta}$ and $\Ecal_{\gamma}$ hold with high-probability,
  one needs to choose appropriate $\lambda_1$ and $\lambda_2$.  This
  calculation is specific to the model at hand.  For example, if the
  vectors
  \begin{equation}
    \label{eq:sub_gaussian:example}
    \Gamma(x_i) \theta^{ab,*} + g(x_i)
    \quad\text{and}\quad
    \eta_{1i}\vpaxi[,-ab]+\eta_{2i}\vpbxi[,-ab]
  \end{equation}
  have sub-Gaussian components, then by taking
  $\lambda_1, \lambda_2 \propto \sqrt{\log p/n}$, the events
  $\Ecal_{\theta}$ and $\Ecal_{\gamma}$ hold with probability at least
  $1 - c_1p^{-c_2}$ \citep{Yang2013Graphical, negahban2010unified}.
  For other distributions, we may need to choose larger $\lambda_1$
  and $\lambda_2$.  See also Lemma 9 in \citet{Yang2013Graphical}. 

\vspace{3mm}
The following result establishes a Bahadur representation for
$\tilde \theta_{ab}$.
 
\begin{theorem}
  \label{thm:main}
  Suppose that assumptions {\bf M} and {\bf SE} hold. Define $w^*$ with $w^{*}_{ab} = 1$ and
$w^{*}_{-ab} = -\gamma^{ab,*}$, where $\gamma^{ab,*}$ is given in the
assumption {\bf M}. On the event
  $\Ecal_\gamma \cap \Ecal_\theta$, we have that 
  \begin{equation}
    \label{eq:bahadur}
    \begin{aligned}
\sqrt{n}\cdot \rbr{\tilde \theta_{ab} - \theta_{ab}^{*}} 
& = - \sigma_n^{-1} 
\cdot\sqrt{n} \EE_n \sbr{w^{* \top}\rbr{\Gamma(x_i)\theta^{ab,*} + g(x_i)} } 
 + \Ocal\rbr{\phi_{\max}^2\phi_{\min}^{-4} \cdot \sqrt{n}\lambda_1\lambda_2  m },    
\end{aligned}
\end{equation}
  where $\sigma_n = \EE_n\sbr{\eta_{1i}\vpaxi[,ab] + \eta_{2i}\vpbxi[,ab]}$.
\end{theorem}

Theorem~\ref{thm:main} is deterministic in nature. It establishes a
representation that holds on the event
$\Ecal_\gamma\cap\Ecal_\theta\cap\Ecal_{\rm SE}$, which in many cases
holds with overwhelming probability. We will show that under suitable
conditions the first term converges to a normal distribution. The
following assumption is a regularity condition needed even in a low
dimensional setting for asymptotic normality of the score matching
estimator \citep{Forbes2013Linear}.

\paragraph{Assumption {\bf R}.}
$\EE_{q^{ab}}\sbr{\norm{\Gamma(X_a,X_b,x_{-ab})\theta^{ab,*}}^2}$ and
$\EE_{q^{ab}}\sbr{\norm{g(X_a,X_b,x_{-ab})}^2}$ are finite for all
values of $x_{-ab}$ in the domain.

\vspace{3mm} Theorem~\ref{thm:main} and Lemma~\ref{lem:L4}
(Appendix~\ref{sec:technical_proofs}) together give the following
corollary:

\begin{corollary}
\label{corollary:normality}
  Suppose that the conditions of Theorem~\ref{thm:main} hold. In
  addition, suppose the assumption {\bf R} holds,
  $\sqrt{n} \lambda_1\lambda_2 m = o(1)$ and 
  $\PP\rbr{\Ecal_\gamma\cap\Ecal_\theta\cap\Ecal_{\rm SE}} \rightarrow 1$.
  Then we have 
  $$
  \sqrt{n} (\tilde \theta_{ab} - \theta_{ab}^*)
  \longrightarrow_D N(0, V_{ab}) , 
  $$
  where
  $V_{ab} = \rbr{\EE\sbr{\sigma_n}}^{-2}\cdot\Var\rbr{w^{* \top}\rbr{\Gamma(x_i)\theta^{ab, *} + g(x_i)}}$ and
  $\sigma_n$ is as in Theorem~\ref{thm:main}.
\end{corollary}

When the vectors in \eqref{eq:sub_gaussian:example} are
  sub-Gaussian, we choose
  $\lambda_1, \lambda_2 \propto \sqrt{\log p/n}$, so that the sample
  complexity is given by $(m\log{p})^2/n = o(1)$.  For other
  distributions, we may need a larger sample size to bound the error
  term in \eqref{eq:bahadur}. We see that the variance $V_{ab}$
depends on the true $\theta^{ab,*}$ and $\gamma^{ab,*}$, which are
unknown. In practice, we estimate $V_{ab}$ using the following
consistent estimator $\hat V_{ab}$,
\begin{equation}
\label{eq:variance_est}
\hat V_{ab} = 
e_{ab}^\top
\rbr{\EE_n\sbr{\Gamma(x_i)}_{\tilde M}}^{-1}
\cdot Z \cdot 
\rbr{\EE_n\sbr{\Gamma(x_i)}_{\tilde M}}^{-1}
e_{ab},
\end{equation}
with
\begin{equation*}
Z={\EE_n\sbr{\rbr{\Gamma(x_i)\tilde \theta^{ab} + g(x_i)}_{\tilde M}\rbr{\Gamma(x_i)\tilde \theta^{ab} + g(x_i)}_{\tilde M}^\top}},
\end{equation*}
and $e_{ab}$ being a canonical vector with $1$ in the position of
element $ab$ and $0$ elsewhere. 
The consistency of this variance estimator is provided in the appendix. 
Using this estimate, we can construct
a confidence interval with asymptotically nominal coverage. In
particular,
\[
\lim_{n \rightarrow \infty}
\sup_{\theta^*\in\Theta} \PP_{\theta^*}\rbr{ \theta_{ab}^* \in \tilde \theta_{ab} \pm z_{\kappa/2} \cdot \sqrt{\hat V_{ab} / n} } = \kappa.
\]
In the next section, we outline the proof of
Theorem~\ref{thm:main}. Proofs of other technical results are
relegated to appendix.

\subsection{Proof  of Theorem~\ref{thm:main}}
\label{sec:proof-sketch-theorem}

We first introduce some auxiliary estimates.
Let $\tilde \gamma^{ab}$ be a minimizer of the following constrained 
problem
\begin{equation}
  \label{eq:estimation:step2:S}
  \begin{aligned}
    &\min_{\gamma} \ \,
    \EE_n\sbr{
      \rbr{\vpaxi[,ab]-\vpaxi[,-ab]^\top \gamma}^2 + 
      \rbr{\vpbxi[,ab]-\vpbxi[,-ab]^\top \gamma}^2 
    } \ \\
    & \text{ s.t. }\ \, {\rm supp}(\gamma) \subseteq \tilde M \backslash (a,b),
  \end{aligned}
\end{equation}
where $\tilde M$ is defined in the step 3 of the procedure.  Essentially,
$\tilde \gamma^{ab}$ is the refitted estimator from step 2 constrained
to have the support on $\tilde M\backslash (a,b)$.  Let $\tilde w \in \RR^{s'}$ with
$\tilde w_{ab} = 1$,
$\tilde w_{\tilde M \backslash (a,b) } = - \tilde{\gamma}_{\tilde{M} \backslash (a,b)}$ and zero
elsewhere. The solution $\tilde \theta^{ab}$ satisfies the first order optimality
condition
$
  \rbr{ \EE_n\sbr{\Gamma(x_i)} \tilde \theta^{ab}  +  \EE_n[g(x_i)] }_{\tilde M} = 0$.
Multiplying by $\tilde w$, it follows that 
\begin{equation}
  \label{eq:kkt3:2}
\begin{aligned}
  & \tilde w^{\top}\rbr{ \EE_n\sbr{\Gamma(x_i)} \tilde \theta^{ab}  +  \EE_n[g(x_i)] } \\
  = &
  \rbr{\tilde w - w^*}^\top \EE_n\sbr{\Gamma(x_i)} \rbr{\tilde \theta^{ab} - \theta^{ab,*}}  +  \rbr{\tilde w - w^*}^\top \rbr{\EE_n\sbr{\Gamma(x_i)\theta^{ab,*} + g(x_i) }}  \\
  & \qquad  +
  w^{* \top} \EE_n\sbr{\Gamma(x_i)} \rbr{\tilde \theta^{ab} - \theta^{ab,*}}  + w^{* \top} \rbr{\EE_n\sbr{\Gamma(x_i)\theta^{ab,*} + g(x_i)} } \\
  \triangleq & L_1 + L_2+L_3+L_4 = 0.
\end{aligned}
\end{equation}
From Lemma~\ref{lem:L1} and Lemma~\ref{lem:L2} (Appendix~\ref{sec:technical_proofs}), we have that
$$
 \abr{L_1 + L_2} \lesssim \phi_{\max}^2\phi_{\min}^{-4} \cdot
\lambda_1\lambda_2 m.
$$
Using Lemma~\ref{lem:L3}, the term $L_3$ can be written as
$$
 L_3 = \EE_n\sbr{\eta_{1i}\vpaxi[,ab] + \eta_{2i}\vpbxi[,ab]} \rbr{\tilde
  \theta_{ab} - \theta_{ab}^{ab,*}} +
\Ocal\rbr{\phi_{\max}^{1/2}\phi_{\min}^{-2}\cdot\lambda_1\lambda_2 m}.
$$ 
Putting all the pieces together, we can rewrite \eqref{eq:kkt3:2} as
$$
\sigma_n \rbr{\tilde \theta_{ab} - \theta_{ab}^{ab,*}} = - w^{* \top} \rbr{\EE_n\sbr{\Gamma(x_i)\theta^{ab,*} + g(x_i)} } + \Ocal\rbr{\lambda_1\lambda_2 m}.
$$
with $\sigma_n = \EE_n\sbr{\eta_{1i}\vpaxi[,ab] + \eta_{2i}\vpbxi[,ab]}$. This completes the proof.


\subsection{Theoretical Results for Non-negative Data}
\label{sec:theoretical_nonnegative}

In this section we provide the theoretical results for non-negative data obtained by modifying the assumptions according to the scoring rule for non-negative data.

\paragraph{\emph{Assumption {\bf M}'.}} The parameter vector $\theta^*$ is sparse,
with
$|{\rm supp}(\theta^{ab,*})| \ll n$.  Let
  \begin{equation}
  \label{eq:estimation:step2:true_nonnegative}
  \begin{aligned}
    \gamma^{ab,*} &= \arg\min_\gamma \
    \EE\sbr{
      (\tilde \varphi_{1,ab}(x_i) - \tilde \varphi_{1,-ab}(x_i)^\top \gamma)^2 + 
      (\tilde \varphi_{2,ab}(x_i) - \tilde \varphi_{2,-ab}(x_i)^\top \gamma)^2 
    },
  \end{aligned}
\end{equation}
with $\tilde\varphi_1, \tilde\varphi_2$ defined in \eqref{eq:tilde_phi_nonnegative}. 
Let $\eta_{1i} = \tilde \varphi_{1,ab}(x_i) - \tilde \varphi_{1,-ab}(x_i)^\top \gamma^{ab,*}$ and
$\eta_{2i} = \tilde \varphi_{2,ab}(x_i) - \tilde \varphi_{2,-ab}(x_i)^\top \gamma^{ab,*}$, for
$i \in [n]$.  The vector $\gamma^{ab,*}$ is sparse with
$|{\rm supp}(\gamma^{ab,*})| \ll n$.  Let
$m = |{\rm supp}(\theta^{ab,*})| \vee |{\rm supp}(\gamma^{ab,*})|$.

\paragraph{\emph{Assumption {\bf SE}'.}} We have
\[
\phi_{\min} 
  \leq \phi_{-}(m \cdot \log n, {\EE\sbr{\Gamma_\ell(x_i)}}) 
  \leq \phi_{+}(m \cdot \log n, {\EE\sbr{\Gamma_\ell(x_i)}})
  \leq \phi_{\max},
\]
where $0 < \phi_{\min} \leq \phi_{\max} < \infty$.

\paragraph{\emph{Assumption {\bf R}'.}}
$\EE_{q^{ab}}\sbr{\norm{\Gamma_\ell(X_a,X_b,x_{-ab})\theta^{ab,*}}^2}$ and
$\EE_{q^{ab}}\sbr{\norm{g_\ell(X_a,X_b,x_{-ab})}^2}$ are finite for all
values of $x_{-ab}$ in the domain.


Denote the modified events as
$$
 \Ecal_{\theta} = \cbr{ \norm{\EE_n\sbr{\Gamma_\ell(x_i) \theta + g_\ell(x_i)}
  }_\infty \leq \frac{\lambda_1}{2} }
$$
and
$$
 \Ecal_{\gamma} = \cbr{ \norm{\EE_n\sbr{
      \eta_{1i}\tilde \varphi_{1,-ab}(x_i)+\eta_{2i}\tilde \varphi_{2,-ab}(x_i) }}_\infty \leq
  \frac{\lambda_2}{2} }.
$$
We have the asymptotic normality for the estimator on non-negative data.

\nocite{hahn2018regularization}
\nocite{hahn2019efficient}
\nocite{he2018xbart}
\nocite{he2020stochastic}

\begin{corollary}
  Suppose that assumptions {\bf M'}, {\bf SE}', and {\bf R'} hold. Define $w^*$ with $w^{*}_{ab} = 1$ and
$w^{*}_{-ab} = -\gamma^{ab,*}$, where $\gamma^{ab,*}$ is given in the
assumption {\bf M}'. 
  In
  addition, suppose 
  $\sqrt{n} \lambda_1\lambda_2 m = o(1)$ and $\PP\rbr{\Ecal_\gamma\cap\Ecal_\theta\cap\Ecal_{\rm SE}} \rightarrow 1$ where
  \begin{equation*}
\Ecal_{\rm SE} = \cbr{
\frac{\phi_{\min}}{2}
  \leq \phi_{-}\big(m \cdot \log n, \EE_n\sbr{\Gamma_\ell(x_i)}\big) 
  \leq \phi_{+}\big(m \cdot \log n, \EE_n\sbr{\Gamma_\ell(x_i)}\big)
  \leq 2 \phi_{\max} }.
\end{equation*}
  Then we have 
  $$
  \sqrt{n} (\tilde \theta_{ab} - \theta_{ab}^*)
  \longrightarrow_D N(0, V_{ab}) , 
  $$
  with the variance term
  $$
  V_{ab} = \rbr{\EE\sbr{\sigma_n}}^{-2}\cdot\Var\rbr{w^{* \top}\rbr{\Gamma_\ell(x_i)\theta^{ab} + g_\ell(x_i)}}
  $$ 
  where 
  $\sigma_n = \EE_n\sbr{ \eta_{1i}\tilde \varphi_{1,ab}(x_i)+\eta_{2i}\tilde \varphi_{2,ab}(x_i) }$.
\end{corollary}

\section{Relaxing the Sparsity Assumption on the Inverse of Hessian}
\label{sec:relaxation_L1}

 For general models, the sparsity condition on $\gamma^{ab,*}$
  could be violated. For example, for the non-negative Gaussian
  graphical model with $\Sigma = \Omega = I_p$, by direct calculation
  we obtain that almost all the components of $\gamma^{ab,*}$ take the
  same value, which is approximately $1/p$. Therefore $\gamma^{ab,*}$
  is neither sparse, nor approximately sparse (see Section
  \ref{sec:experiments_synthetic} for details).  Instead, it only
  satisfies a weaker condition $\|\gamma^{ab,*}\|_1 \leq 2$.  This
  constant $L_1$ norm condition is studied in \cite{ma2017inter}.
  Since $\gamma^{ab,*}$ is dense, we cannot select sparse support in
  Step 2; and therefore Step 3 is no longer valid when $p > n$.

We relax the sparsity condition on $\gamma^{ab,*}$ to a
  constant $L_1$ condition, and modify our procedure.  We apply the
  debias method in \cite{ma2017inter}. Specifically, recall that the
  scoring rule is
  \begin{equation}
    S^{ab}(x, \theta)   = \frac 12 \theta^\top \Gamma(x) \theta + \theta^\top g(x) + c(x),
  \end{equation}
  and the gradient with respect to $\theta$ is
  \begin{equation}
    \nabla  S^{ab}(x, \theta)  =  \Gamma(x) \theta + g(x). 
  \end{equation}
  We obtain an estimator $\hat\theta^{ab}$ using Step 1, which satisfies
  \begin{equation}
    \nabla S^{ab} \big( x, \hat\theta^{ab} \big) - \nabla S^{ab} \big( x, \theta^{ab, *} \big)  = \Gamma(x) \big( \hat\theta^{ab} - \theta^{ab, *} \big).
  \end{equation}
  Multiplying by some matrix $M$ on both sides and rearranging terms, we obtain
  \begin{equation}
    \label{eq:split_general_debias}
    \hat\theta^{ab} - M \cdot \nabla S^{ab} \big( x, \hat\theta^{ab} \big) = \theta^{ab, *} - M \cdot \nabla S^{ab} \big( x, \theta^{ab, *} \big) + \big( I - M \cdot \Gamma(x) \big) \big( \hat\theta^{ab} - \theta^{ab, *} \big).
  \end{equation}
  The empirical version of \eqref{eq:split_general_debias} is 
  \begin{multline}
    \label{eq:split_general_debias_empirical}
    \hat\theta^{ab} - M \cdot \EE_n[ \nabla S^{ab} \big( x_i, \hat\theta^{ab} \big)] \\
    = \theta^{ab, *} - M \cdot \EE_n[ \nabla S^{ab} \big( x_i, \theta^{ab, *} \big) ] + \big( I - M \cdot \EE_n[\Gamma(x_i)] \big) \big( \hat\theta^{ab} - \theta^{ab, *} \big).
  \end{multline}
  Rather than using Step 3 in the procedure described in Section~\ref{sec:methodology},
  we define the left hand side as the proposed estimator:
\begin{equation}
\label{eq:estimator_L1_relax}
\tilde \theta^{ab} = \hat\theta^{ab} - M \cdot \EE_n[ \nabla S^{ab} \big( x_i, \hat\theta^{ab} \big) ]
= \hat\theta^{ab} - M \cdot \frac{1}{n} \sum_{i=1}^n \Gamma(x_i) \hat\theta^{ab} + g(x_i).
\end{equation}
Notice that the first term in the right hand side of
\eqref{eq:split_general_debias_empirical} is the true value.  Suppose
$M$ is an approximate inverse of $\EE_n[ \Gamma(x) ]$, then the third
term in the right hand side of
\eqref{eq:split_general_debias_empirical} would be negligible.  For
the second term, we see that
$\EE_n[ \nabla S^{ab} \big( x_i, \theta^{ab, *} \big) ]$ is an average
of $n$ i.i.d. samples. If it is independent of $M$, then this second
term is asymptotically normal, and the coordinate $ab$ of
$\tilde \theta^{ab}$ is the desired estimator, similar to the
three-step procedure described in Section~\ref{sec:methodology}.  We
construct $M$ following the procedure in \cite{ma2017inter}.  We first
split the data into two parts and estimate $\hat\theta^{ab}$ on the
first part, while $M$ is estimated on the second part.  For notation
simplicity, let $\{x_i\}_{i=1}^n$ denote observations on the first
part and $\{x_i'\}_{i=1}^n$ on the second part.  We estimate $M$ by
solving the following convex program:
\begin{equation}
\begin{aligned}
& \text{minimize} \quad \|M\|_\infty \\
& \text{subject to} \,\,\,  \left\| \, I - M \cdot \EE_n[ \Gamma(x_i') ] \, \right\|_{\max} \leq \lambda_2.
\end{aligned}
\end{equation}
By selecting appropriate $\lambda_2$, the solution $M$ will be an
approximate inverse of $\EE_n[ \Gamma(x_i') ]$ and, hence, an
approximate inverse of $\EE_n[ \Gamma(x_i) ]$.  On the other hand,
since we estimate $M$ based on second part of the data,
$\{x_i'\}_{i=1}^n$, it is independent of
$\EE_n[ \nabla S^{ab} \big( x_i, \theta^{ab, *} \big) ]$.  Let $M^*$
be the population version of $M$. We see that the column $ab$ of $M^*$
(denoted as $M^*_{ab}$) corresponds to $w^*$ up to a constant, where
$w^*$ is defined in Theorem \ref{thm:main} with $w^{*}_{ab} = 1$ and
$w^{*}_{-ab} = -\gamma^{ab,*}$.  For non-negative Gaussian graphical
model with $\Sigma = \Omega = I_p$, a simple calculation shows that for large
$p$, we have $\| M^*_{ab} \|_1 \leq 1.5/(1-\frac{2}{\pi}) < 5$. We
then see that the bounded $L_1$ norm condition on $M^*_{ab}$ is
satisfied.

  To establish asymptotic normality of the modified procedure, we define the
  following event
  $$
  \Ecal_{\gamma}' = \cbr{
    \left\|\, I - M^* \cdot \EE_n[ \Gamma(x_i) ] \, \right\|_{\max}
    \leq {\lambda_2} }.
  $$
  For example, when $\varphi_1(x)$ and $\varphi_2(x)$ are sub-Gaussian
  vectors, modification of Lemma D.1 in \cite{ma2017inter} gives us
  that if $\lambda_2 \asymp \sqrt{\frac{\log p}{n}}$, then
  $\PP\rbr{ \Ecal_{\gamma} } \rightarrow 1$.  By the proof of Lemma
  \ref{lem:refit}, we have that
  $\| \hat\theta^{ab} - \theta^{ab, *} \|_1 \lesssim \lambda_1 m$.
  This shows that the third term of
  \eqref{eq:split_general_debias_empirical} is of order
  $m \cdot \log p / n$.  Suppose $(m\log{p})^2/n = o(1)$, we then
  obtain a similar result as in Corollary \ref{corollary:normality}.
  It is also straightforward to see that the variance given by
  \eqref{eq:split_general_debias_empirical} is asymptotically the same
  as $V_{ab}$ in Corollary \ref{corollary:normality}. We conclude
  with the following Corollary for sub-Gaussian distribution.

\begin{corollary}
\label{corollary:L_1_relaxation}
Suppose that assumptions {\bf SE} and {\bf R} hold.  Furthermore,
suppose $\| M^*_{ab} \|_1 \leq C$. If $(m\log{p})^2/n = o(1)$ and
$\PP\rbr{\Ecal_\gamma'\cap\Ecal_\theta\cap\Ecal_{\rm SE}} \rightarrow
1$, then the estimator $\tilde \theta^{ab}$ in
\eqref{eq:estimator_L1_relax} satisfies
  $$
  \sqrt{n} (\tilde \theta_{ab} - \theta_{ab}^*)
  \longrightarrow_D N(0, V_{ab}) , 
  $$
  where
  $V_{ab} = \Var\big( {M_{ab}^*}^\top \big(\Gamma(x_i) \theta^{ab, *} + g(x_i) \big) \big) $. 
\end{corollary}

\section{Simultaneous Inference}
\label{sec:simultaneous}

In the last two sections, we have developed a procedure for
constructing a consistent and asymptotically normal estimate of a
single edge parameter. In this section, we develop a procedure for
simultaneous hypothesis testing of all edges connected to a specific
node. We adopt the Gaussian multiplier bootstrap
\citep{Chernozhukov2013Gaussian} to our setting.
In this section we focus on the case where $\Xcal = \RR^p$. The analysis can be straightforwardly extended to non-negative data. 

For a fixed node $a \in V$, we would like to test the null hypothesis
\begin{equation}
\label{eq:H0_simultaneous}
H_0: \theta_{ab}^* = \breve\theta_{ab} \quad \text{for all } b \in V_a =  \{1, \ldots, p\} \backslash \{a\},
\end{equation}
for some values $\breve\theta_{ab}$ versus the alternative 
\begin{equation}
\label{eq:H1_simultaneous}
H_1: \theta_{ab}^* \neq \breve\theta_{ab} \quad \text{for some } b \in V_a = \{1, \ldots, p\} \backslash \{a\}.
\end{equation}
We propose the following test statistic 
\begin{equation}
\label{eq:def_test_statistic_simultaneous}
\max_{b \in V_a} \sqrt{n} \abr{ \tilde \theta_{ab} - \breve \theta_{ab} },
\end{equation}
where $\tilde \theta_{ab}$ is obtained by the three step procedure
described in Section \ref{sec:methodology}. The null hypothesis will
be rejected for large values of the test statistic. Using the
$\ell_\infty$ statistics will allow us to have power against
alternatives that change few of the coordinates of
$\breve\theta_{ab}$. In order to use the test statistic in practice,
we need to be able to accurately compute the critical value of the
test statistic in a high-dimensional setting. To that end, we describe
a multiplier bootstrap method that will allow us to obtain an
accurate critical value to the test statistic in~\eqref{eq:def_test_statistic_simultaneous}.

For each $b \in V_a$ and $i \in \{1, \ldots, n\}$, denote
\begin{equation}
\label{eq:def_xib_tilde}
\tilde z_{iab} = -\sigma_{n,ab}^{-1} \cdot \tilde w_{ab}^{\top} \Big(\Gamma_{ab}(x_i) \breve \theta^{ab} + g_{ab}(x_i) \Big),
\end{equation}
where
$\sigma_{n,ab} = \EE_n\sbr{\eta_{1iab}\vpaxi[,ab] + \eta_{2iab}\vpbxi[,ab]}$ as
defined in Theorem~\ref{thm:main}. 
We use the subscript $ab$ to highlight that all of these terms depend on the node $a$ and $b$.
Let $e_i$, $i = 1,\ldots, n$, be a
sequence of independent standard Gaussian random variables and
independent of data. We define the multiplier bootstrap statistic as
\begin{equation}
\label{eq:def_W_tilde}
\tilde W = \max_{b \in V_a} \frac{1}{\sqrt n} \sum_{i=1}^n \tilde z_{iab} e_{i}
\end{equation}
and compute the bootstrap critical value as the $(1-\alpha)$ quantile
of $\tilde W$
\[
  c_{\tilde W}(\alpha) = \inf\{t \in \RR: \PP(\tilde W \leq t) \geq
  1-\alpha\}.
\]
Importantly, note that the quantile of the multiplier bootstrap
statistic can be estimated using a Monte-Carlo method. We will show
that the quantiles of $\tilde W$ approximate the quantiles of our
test statistic.

Define
\begin{equation}
\label{eq:def_xib}
z_{iab} = -\sigma_{ab}^{-1} \cdot w_{ab}^{*\top}\Big(\Gamma(x_i)\theta^{ab,*} + g(x_i) \Big),
\end{equation} 
as the counterpart to $\tilde z_{iab}$, where $\sigma_{ab} = \EE[\sigma_{n,ab}]$.
In order to establish our main theoretical result on simultaneous
inference, we need the following regularity condition.

\paragraph{\emph{Assumption RR.}} Define
$\gamma_{abc}(x_i) = z_{iab}z_{iac} - \EE(z_{iab}z_{iac})$.  There exist
$\eta_n$ and $\tau_n^2$, such that for any $b,c \in V_a$, we have
$\|\gamma_{abc}(x_i)\|_{\infty} \leq \eta_n$ and
$\frac 1n \sum_{i=1}^n \EE \gamma^2_{abc}(z_i) \leq \tau_n^2$ with
probability at least $1-n^{-c_1}$. Moreover, uniformly for $b \in V_a$,
we have $c_0 \leq \frac 1n \sum_{i=1}^n \EE z_{iab}^2 \leq C_0$ for
some $0 < c_0 < C_0$.

The assumption RR imposes very mild technical conditions and is standard
for a large number of models
when the sample size is large enough.
Part of the conditions are adopted from
\cite{Chernozhukov2013Gaussian} in order to apply the theoretical
results on the Gaussian multiplier bootstrap.

\begin{theorem}
  \label{thm:simultaneous}
    
  Suppose the assumptions {\bf M}, {\bf SE}, {\bf R} and {\bf RR} are satisfied, 
  and the events $\Ecal_\gamma\cap\Ecal_\theta\cap\Ecal_{\rm SE}$ hold for each $b \in V_a$.
  Furthermore, suppose there exists a constant $\epsilon > 0$,
  such that
\begin{equation}
\label{eq:asmp_regime_simultaneous}
\frac 1n \Big[ (\tau_n^2 + \eta_n) \log p + (m\log{p})^2 + \log(pn)^7 \Big] = o(n^{-\epsilon}).
\end{equation}
Then, under the null hypothesis, we have
\begin{equation}
\sup_{\alpha \in (0,1)} \bigg| \PP\Big(\max_{b \in V_a} \sqrt{n} ( \tilde \theta_{ab} - \breve \theta_{ab} ) \geq c_{\tilde W}(\alpha) \Big) - \alpha  \bigg| = o(1).
\end{equation}
\end{theorem}

The proof of Theorem \ref{thm:simultaneous} is provided in the appendix. 
Since
\[
  | \tilde \theta_{ab} - \breve \theta_{ab} | = \max\{\tilde
\theta_{ab} - \breve \theta_{ab} , \breve \theta_{ab} - \tilde
\theta_{ab} \},
\]
it is straightforward to obtain the following
corollary for the test statistic
in~\eqref{eq:def_test_statistic_simultaneous}.

\begin{corollary}
\label{corollary_simultaneous}
Suppose the conditions in Theorem \ref{thm:simultaneous} are
satisfied.  Then, under the null hypothesis, we have
\begin{equation}
\sup_{\alpha \in (0,1)} \bigg| \PP\Big(\max_{b \in V_a} \sqrt{n} | \tilde \theta_{ab} - \breve \theta_{ab} | \geq c_{\overline W}(\alpha) \Big) - \alpha  \bigg| = o(1),
\end{equation}
where
\begin{equation}
\label{eq:def_W_overline}
\overline W = \max_{b \in V_a} \frac{1}{\sqrt n} \bigg| \sum_{i=1}^n \tilde z_{iab} e_{i} \bigg|
\end{equation}
and the bootstrap critical value is defined as
\[
  c_{\overline W}(\alpha) = \inf\{t \in \RR: \PP(\overline W \leq t) \geq 1-\alpha\}.
\]
\end{corollary}

 We remark that we are not aiming for a tight bound on the
  sample complexity.  For commonly used models, we always have that
  $\gamma_{abc}(x_i)$ in Assumption {\bf RR} converges to 0 at a model
  specific rate.  Theorem \ref{thm:simultaneous} is valid as long as
  the sample size is large enough, so that the sample complexity
  condition in \eqref{eq:asmp_regime_simultaneous} is satisfied.  

Based on Corollary~\ref{corollary_simultaneous}, we reject the null
hypothesis if the test
statistic~\eqref{eq:def_test_statistic_simultaneous} is greater than
$c_{\overline W}(\alpha)$. This gives us a valid simultaneous test for
all the edges connected to some node $a \in V$ with asymptotic Type I
error equal to $\alpha$.

\subsection{Applications of Simultaneous Testing}
\label{sec:application}

In this section, we show three concrete applications of our proposed procedure. Specifically, we consider
\begin{enumerate}[topsep=0pt,itemsep=-1ex,partopsep=1ex,parsep=1ex]
\item testing for isolated node; 
\item support recovery;
\item testing for difference between graphical models.
\end{enumerate}

\paragraph{\emph{Testing for isolated node.}}
For a specific node $a \in V$, we would like to test whether it is
isolated in the graph. This specific structural question translates
into whether the variable $X_a$ is conditionally independent with all the other
nodes. In this case, we would like to test the null hypothesis
\begin{equation}
\label{eq:H0_isolated}
H_0: \theta_{ab}^* = 0 \quad \text{for all } b \in V_a =  \{1, \ldots, p\} \backslash \{a\},
\end{equation}
versus the alternative 
\begin{equation}
\label{eq:H1_isolated}
H_1: \theta_{ab}^* \neq 0 \quad \text{for some } b \in V_a = \{1, \ldots, p\} \backslash \{a\}.
\end{equation}
We can directly apply our simultaneous inference procedure with $\breve\theta_{ab} = 0$.

\paragraph{\emph{Support recovery.}}
For a specific node $a \in V$, we would like to estimate the support
of $a$ defined as
$\text{supp}(a) = \{b \in V_a, \theta_{ab}^* \neq 0\}$. Let $S^*$ be
the true support and we focus on distributions with sub-Gaussian components.  For each node $b \in V_a$, let $\tau_{ab}$ be a
threshold that we set as 
\[
  \tau_{ab} = \sqrt{2\hat V_{ab} \log p/n},
\]
where $\hat V_{ab}$ is the variance estimator defined in \eqref{eq:variance_est}.
We can estimate the support $S^*$ by thresholding the values
$\tilde \theta_{ab}$ that are smaller than $\tau_{ab}$. In particular,
the support recovery procedure return the following support set
\begin{equation}
\label{eq:support_set}
\hat S(\tau_{ab}) = \{ b \in V_a, | \tilde \theta_{ab} | > \tau_{ab} \}.
\end{equation}

We have the following result on the support recovery.
\begin{corollary}
Suppose that the values $\theta_{ab}^{*}$ on the true support are
bounded from below as 
\begin{equation*}
|\theta_{ab}^*| > \sqrt{\frac{8\hat V_{ab} \log p}{n}}, \qquad \text{for all } b \in S^*.
\end{equation*}
Then 
\begin{equation}
\inf \PP \big( \hat S(\tau_{ab}) = S^* \big) \xrightarrow{ n \to \infty } 1,
\end{equation}
where the infimum is taken over all data generating procedures that satisfy the minimum signal strength condition.
\end{corollary}
The proof follows in a similar way to the proof of Proposition 3.1 in
\citet{Zhang2014Simultaneous} and is omitted here. The result shows
that we are able to consistently recover the support of any node with
overwhelming probability.

\paragraph{\emph{Testing the difference between graphical models.}} We
consider a two-sample problem in which we wish to test whether the
parameters of two graphical models, with the same set of nodes and
belonging to the same exponential family of the form
in~\eqref{eq:pairwise}, are the same. For example, we may have the
data for the same set of nodes collected in different time periods, and
we want to test whether the graph structure changes over time.  As
another example, consider functional brain connectivity. It is of
interest to test whether brain connectivity is the same for
the healthy subjects and people with a certain disorder. 

Formally, suppose there are two densities $p_{\theta_{ab,1}^*}$ and
$p_{\theta_{ab,2}^*}$ of the form in~\eqref{eq:pairwise}, indexed by
parameter vectors $\theta_{ab,1}^*$ and $\theta_{ab,2}^*$.  Given
$n_1$ i.i.d.~samples $\{x_{i,1}\}_{i\in[n_1]}$ from
$p_{\theta_{ab,1}^*}$ and $n_2$ i.i.d.~samples
$\{x_{i,2}\}_{i\in[n_2]}$ from $p_{\theta_{ab,2}^*}$, we would like to
test the null hypothesis
\begin{equation}
\label{eq:H0_difference}
H_0: \theta_{ab,1}^* = \theta_{ab,2}^* \quad \text{for all }  a, b \in V \times V,
\end{equation}
versus the alternative 
\begin{equation}
\label{eq:H1_difference}
H_1: \theta_{ab,1}^* \neq \theta_{ab,2}^* \quad \text{for some } a, b \in V \times V.
\end{equation}

In order to create a test statistic for the difference, we first apply
the three step procedure on each group of observations. That is, we
obtain the estimators $\tilde \theta_{ab,1}$, $\tilde \theta_{ab,2}$
and estimates of their variances $\hat V_{ab,1}, \hat
V_{ab,2}$. According to the Bahadur representation \eqref{eq:bahadur}
in Theorem~\ref{thm:main}, we have
\begin{equation}
\label{eq:bahadur1}
\sqrt{n_1}\cdot \rbr{\tilde \theta_{ab,1} - \theta_{ab,1}^{*}} = - \hat \sigma_{n,ab,1}^{-1} 
\cdot\sqrt{n_1} \EE_{n_1} \sbr{w_{ab,1}^{*\top}\rbr{\Gamma_{ab}(x_{i,1})\theta_1^{ab,*} + g_{ab}(x_{i,1})} } + o_\PP(1),
\end{equation}
and
\begin{equation}
\label{eq:bahadur2}
\sqrt{n_2}\cdot \rbr{\tilde \theta_{ab,2} - \theta_{ab,2}^{*}} = - \hat \sigma_{n,ab,2}^{-1} 
\cdot\sqrt{n_2} \EE_{n_2} \sbr{w_{ab,2}^{*\top}\rbr{\Gamma_{ab}(x_{i,2})\theta_2^{ab,*} + g_{ab}(x_{i,2})} } + o_\PP(1).
\end{equation}
We propose to use the following test statistic 
\begin{equation}
\sqrt{n_1 + n_2} \cdot \max_{a, b \in V \times V} |\tilde \theta_{ab,1} - \tilde \theta_{ab,2}|,
\end{equation}
which will allow us to identify sparse changes in parameter values.
We reject the null hypothesis for large values of the test statistic
above. Next, we describe how to estimate the quantiles of the test
statistic using the multiplier bootstrap.

Denote 
\begin{equation}
\label{eq:def_xib_tilde1}
\tilde z_{iab,1} = -\sigma_{n,ab,1}^{-1} \cdot \tilde w_{ab,1}^{\top} \Big(\Gamma_{ab}(x_{i,1}) \tilde \theta^{ab}_1 + g_{ab}(x_{i,1}) \Big),
\end{equation}
and
\begin{equation}
\label{eq:def_xib_tilde2}
\tilde z_{iab,2} = -\sigma_{n,ab,2}^{-1} \cdot \tilde w_{ab,2}^{\top} \Big(\Gamma_{ab}(x_{i,2}) \tilde \theta^{ab}_2 + g_{ab}(x_{i,2}) \Big).
\end{equation}
We generate two sequences of independent standard Gaussian random
variables
\[
  e_{i,j} \sim N(0,1) \qquad \text{for } i = 1,\ldots, n_j, \text{ and } j=1,2,
\]
that are independent
of data as well.  The multiplier bootstrap statistic is defined as
\begin{equation}
\label{eq:def_W_overline_diff}
\overline W = \frac{1}{\sqrt{n_1 + n_2}} \cdot \max_{a,b \in V \times V} \abr{  \rbr{1+\frac{n_2}{n_1}}\sum_{i=1}^{n_1} \tilde z_{iab,1} e_{i,1} - \rbr{1+\frac{n_1}{n_2}}\sum_{i=1}^{n_2} \tilde z_{iab,2} e_{i,2} }
\end{equation}
and
\[
  c_{\overline W}(\alpha) = \inf\{t \in \RR: \PP(\overline W \leq t) \geq 1-\alpha\}
\]
is the bootstrap critical value.

Similar to Corollary \ref{corollary_simultaneous}, under the null hypothesis, we have
\begin{equation}
\sup_{\alpha \in (0,1)} \abr{ \PP\rbr{ \sqrt{n_1 + n_2} \cdot \max_{a,b \in V \times V} \abr{\tilde \theta_{ab,1} - \tilde \theta_{ab,2}}
 \geq c_{\overline W}(\alpha) } - \alpha  } = o(1).
\end{equation}
This gives us a valid procedure for testing whether the parameters of
two graphical models are the same or not.

A recent paper \citep{kim2019two} proposed a different inference procedure that directly estimates the parameters of the differential network. 
\citet{Xia2015Testing} studied the two sample problem in the context
of Gaussian graphical models and proposed the following
test statistic
\begin{equation}
T = \max_{a, b \in V \times V} \, \frac{\rbr{\tilde \theta_{ab,1} - \tilde \theta_{ab,2}}^2}{\hat V_{ab,1} + \hat V_{ab,2}} 
\end{equation}
and showed that under the null hypothesis
the limiting distribution of the test statistic satisfies
\begin{equation}
\label{eq:test_diff_limit}
\PP \big(T - 2\log p + \log\log p \leq t \big) \to \exp \big\{ (-2 \pi)^{-\frac 12} \exp(-t/2) \big\},
\qquad \text{as }n \rightarrow \infty.
\end{equation}
Unfortunately, the convergence to the extreme value distribution is
rather slow and, as a result, the critical values based on the
limiting approximation are not accurate for finite samples.  In
comparison, our multiplier bootstrap procedure provides non-asymptotic
approximation to quantiles of the test statistic.  Furthermore, the
approximation quality improves polynomially with the sample size and,
as a result, provides a good performance for small and moderate sample
sizes. 

Extending the above described inferential procedure to differential 
networks with latent variables \citep{Na2019Estimating} and
differential functional graphical models 
\citep{Zhao2019Direct, Zhao2020FuDGE} is 
left for future work.

\section{Extension to General $L$}
\label{sec:general_L}

 So far we have assumed that the number of parameters
  corresponding to an edge is $L = 1$. In this section we extend our
  results to general $L$. Throughout the section, we treat $L$ as a
  fixed quantity. Recall that $t_{ab}^{(l)}$, $l\in[L]$, represent
  sufficient statistics.

\paragraph{\emph{Inference for a fixed edge.}}
For a fixed index $(a, b)$,
the parameter of interest is the $L$ dimensional vector,
$\theta_{ab}^{[L]} = \big[\theta_{ab}^{(1)}, \ldots, \theta_{ab}^{(L)}\big]$.
There is no edge between $a$ and $b$ in the corresponding conditional
independence graph if and only if $\theta_{ab}^{(1)}=\cdots=\theta_{ab}^{(L)}=0$. 
Following the same procedure as in Section \ref{sec:methodology}, we have the logarithm of conditional density as
\begin{align*}
  \log q^{ab}_\theta(x) = \dotp{\theta^{ab}}{\vpx} 
  - \Psi^{ab}(\theta, x_{-ab})  
  + h^{ab}(x),
\end{align*}
where $\theta^{ab} \in \RR^{s'}$, with $s' = 2K+2(p-2)L+L$, is the part
of the vector $\theta$ corresponding to
$\left\{ \theta_{a}^{(k)}, \theta_b^{(k)} \right\}_{k\in[K]}$,
$\left\{\theta_{ac}^{(l)},\theta_{bc}^{(l)}\right\}_{l\in[L],c\in-ab}$, and $\left\{\theta_{ab}^{(l)}\right\}_{l\in[L]}$; 
 and
$\vpx = \varphi^{ab}(x) \in \RR^{s'}$ is the corresponding vector of
sufficient statistics
\[
  \left\{t_a^{(k)}(x_a), t_b^{(k)}(x_b)\right\}_{k\in[K]},\
  \left\{ t_{ac}^{(l)}(x_a,x_c), t_{bc}^{(l)}(x_b,x_c) \right\}_{l\in[L], c\in-ab},
  \text{ and }
  t_{ab}^{(l)}(x_a,x_b)_{l\in[L]}.
\]

For notation simplicity, for a given node $c \in -ab$, denote
  $\theta^{ac} \in \RR^{L}$ as the stack of
  $\left\{\theta_{ac}^{(l)}\right\}$ for $l \in [L]$; similarly,
  denote $\theta^{bc} \in \RR^{L}$ as the stack of
  $\left\{\theta_{bc}^{(l)}\right\}$. Let
  $\theta^{ab, -{\rm group}}$ denote the stack of
  $\left\{ \theta_{a}^{(k)}, \theta_b^{(k)} \right\}_{k\in[K]}$ and
  $\left\{\theta_{ab}^{(l)}\right\}_{l\in[L]}$, which are the parameters
  in $\theta^{ab}$ without group structure.  We define $\gamma^{ac}$,
  $\gamma^{bc}$, and $\gamma^{ab, -{\rm group}}$ similarly.  Let
  $E(a,b)$ denote the index set of the parameters
  corresponding to the edge $(a,b)$.
  Figure \ref{illustration} presents an illustrative
  example with $L=K=2$, $p=6$, and $(a,b) = (1,2)$.

\begin{figure}[t]
\begin{center}
\includegraphics[width=13cm]{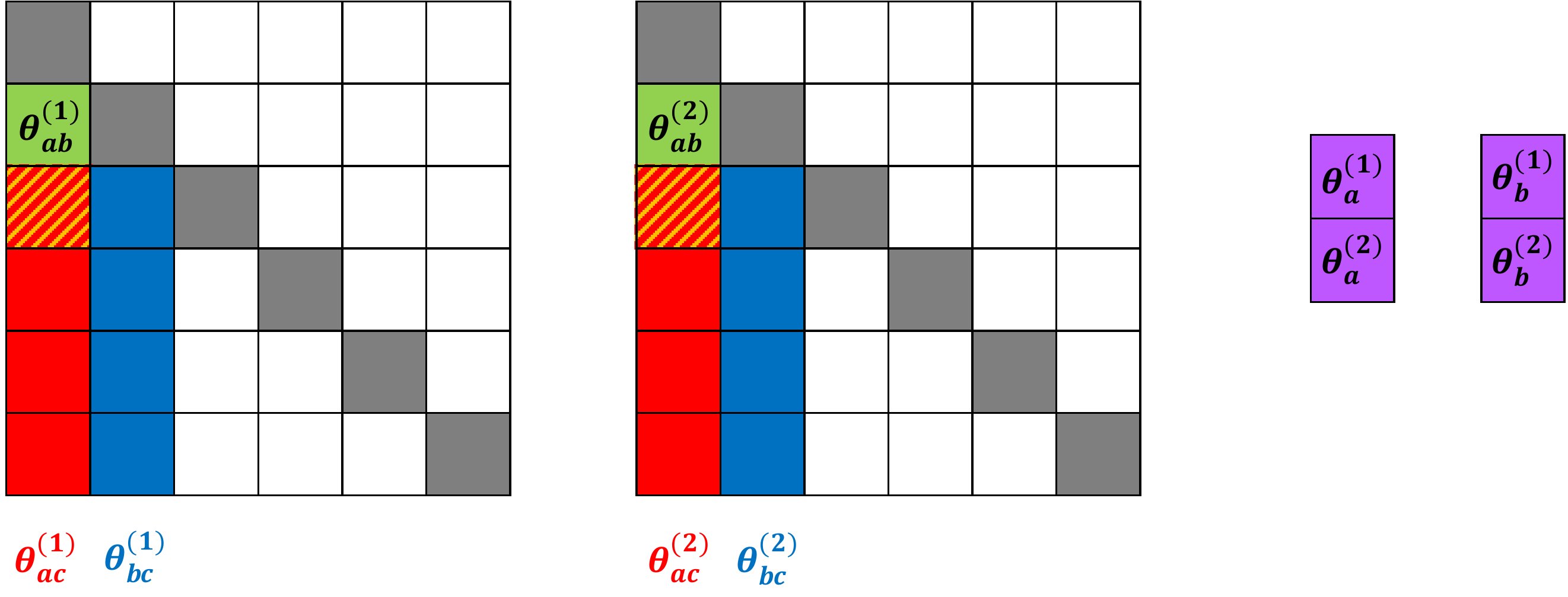}
\caption{An illustrative example with $L=K=2$, $p=6$, and $(a,b) = (1,2)$.
  The green cells are the parameters of interest: $\left\{\theta_{ab}^{(l)}\right\}_{l\in[L]}$;
  the red cells represent $\left\{\theta_{ac}^{(l)}\right\}_{l\in[2],c\in-ab}$;
  the blue cells represent $\left\{\theta_{bc}^{(l)}\right\}_{l\in[2],c\in-ab}$;
  the purple cells represent $\left\{ \theta_{a}^{(k)}, \theta_b^{(k)} \right\}_{k\in[2]}$. 
  These parameters constitute $\theta^{ab} \in \RR^{s'}$.
  The green and purple cells correspond to $\theta^{ab, -{\rm group}}$.
  The striped red cells correspond to $\theta^{ac} = \left\{\theta_{ac}^{(l)}\right\}_{l \in [2]}$ with $c = 3$. 
  Finally, the white cells are parameters not used in the estimation,
  while the gray cells are zero diagonal values. }
\label{illustration}
\end{center}
\end{figure}

   We modify the
   three step procedure in Section \ref{sec:methodology} as follow.

\paragraph{\emph{Step 1:}} We find a pilot estimator of $\theta^{ab}$ by solving the 
following program 
\begin{equation}
  \label{eq:estimation_general_L}
\begin{aligned}
  \hat \theta^{ab} 
  & =  \arg\min_{\theta \in \RR^{s'}} \ \EE_n\sbr{S^{ab}(x_i, \theta)} +
  \lambda_{1} \bigg( \| \theta^{ab, -{\rm group}} \|_1 +
  \sum_{c \in -ab} \Big( \| \theta^{ac} \|_2 +  \| \theta^{bc} \|_2 \Big)  \bigg), 
\end{aligned}  
\end{equation}
where
\begin{equation}
\| \theta^{ab, -{\rm group}} \|_1 = \sum_{l=1}^L | \theta_{ab}^{(l)} |   + \sum_{k=1}^K  |\theta_{a}^{(k)}| +  |\theta_{b}^{(k)}|
\end{equation}
and $\lambda_1$ is a tuning parameter.
Since $L > 1$, we use the group Lasso penalty to estimate $\hat \theta^{ab}$. 
Let $\hat M_1$ be the support of $\hat \theta^{ab}$:
\begin{equation}
\hat M_1 = {\rm supp} ( \hat\theta^{ab, -{\rm group}} ) \bigcup   \{ E(a,c) \mid \| \hat\theta^{ac} \|_2 \neq 0 \} \bigcup \{ E(b,c) \mid \| \hat\theta^{bc} \|_2 \neq 0 \}.
\end{equation}

\paragraph{\emph{Step 2:}} For $l \in [L]$, let $\hat \gamma^{abl} \in \RR^{s'-1} $ be a minimizer of 
\begin{multline}
  \label{eq:estimation:step2_general_L}
     \sum_{l \in [L]} \frac 12 \EE_n\Big[
      (\vpaxi[,abl]-\vpaxi[,-abl]^\top \gamma^{abl})^2 + 
      (\vpbxi[,abl]-\vpbxi[,-abl]^\top \gamma^{abl})^2 
    \Big]  \\
     + \lambda_{2} \bigg(  \sum_{l \in [L]} \| \gamma^{abl, -\rm{group}} \|_1  + \sum_{c \in -ab} \Big( \| \gamma^{ac} \|_2 +  \| \gamma^{bc} \|_2 \Big)   \bigg),
   \end{multline}
   where $\lambda_2$ is a tuning parameter. Let $\hat M_2$ be the union of the support of $\hat \gamma^{abl}$:
\begin{equation}
\hat M_2 = \bigcup_{l \in [L]}{\rm supp} ( \hat\gamma^{abl, -{\rm group}} ) \bigcup  \{ E(a,c) \mid \| \hat\gamma^{ac} \|_2 \neq 0 \} \bigcup \{ E(b,c) \mid \| \hat\gamma^{bc} \|_2 \neq 0 \}.
\end{equation}

\paragraph{\emph{Step 3:}}
Let $\tilde M = E(a,b) \cup \hat M_1 \cup \hat M_2$.
We obtain our estimator as a solution to the following 
program
\begin{equation}
  \label{eq:estimation:step3_general_L}
\begin{aligned}
  \tilde \theta^{ab} 
  & = \arg\min_\theta \ \EE_n\sbr{ S^{ab}(x_i, \theta) }
  \qquad \text{s.t.}\quad {\rm supp}(\theta) \subseteq \tilde M.
\end{aligned}  
\end{equation}
Our estimator of $\theta_{ab}^{[L]}$ is $\tilde \theta_{ab}^{[L]} \in \RR^L$, a block of 
$\tilde \theta^{ab}$.

\paragraph{\emph{Asymptotic Normality.}}
For each $l \in [L]$, define $w_l^* \in \RR^{s'}$ with $w^{*}_{abl} = 1$ and
$w^{*}_{-abl} = -\gamma^{abl,*}$, where $\gamma^{abl,*}$ is the population version of $\hat \gamma^{abl}$. 
Define 
\begin{equation}
\eta_{1il} = \vpaxi[,abl]-\vpaxi[,-abl]^\top \gamma^{abl,*} ~~~\text{and}~~~ \eta_{2il} = \vpbxi[,abl]-\vpbxi[,-abl]^\top \gamma^{abl,*},
\end{equation}
and 
\begin{equation}
\sigma_{n, l} =   \EE_n\sbr{\eta_{1il}\vpaxi[,abl] + \eta_{2il}\vpbxi[,abl]}.
\end{equation}
Let $u^*_l = w_l^* / \sigma_{n, l}$ and $U^* \in \RR^{s' \times L}$ as the stack of $u^*_l$: $U^* = [u_1^*, \ldots, u_L^*]$.
Similar to Theorem \ref{thm:main}, 
we obtain the Bahadur representation for $\tilde \theta_{ab}^{[L]} \in \RR^L$ as: 
  \begin{equation}
    \label{eq:bahadur_general_L}
    \begin{aligned}
\sqrt{n}\cdot \rbr{\tilde \theta_{ab}^{[L]} - \theta_{ab}^{*[L]}}
& = - \sqrt{n} \EE_n \sbr{U^{* \top}\rbr{\Gamma(x_i)\theta^{ab,*} + g(x_i)} } + \Delta,
\end{aligned}
\end{equation}
where $\| \Delta \|_\infty = \Ocal\rbr{\phi_{\max}^2\phi_{\min}^{-4} \cdot \sqrt{n}\lambda_1\lambda_2  m }$.
Furthermore, under similar conditions as in Section \ref{sec:asympt-norm-estim}, we obtain 
\begin{equation}
\label{eq:asymptotic_normal_general_L}
  \sqrt{n} \Big( \tilde \theta_{ab}^{[L]} - \theta_{ab}^{*[L]} \Big)  \longrightarrow_D N(0, V_{ab}) , 
\end{equation}
where $V_{ab} \in \RR^{L \times L}$ is the covariance matrix defined as 
$V_{ab} = \Var\big( U^{* \top}(\Gamma(x_i)\theta^{ab,*} + g(x_i)) \big)$.
From \eqref{eq:asymptotic_normal_general_L} we can construct a multivariate confidence interval with asymptotically nominal coverage as before.

\paragraph{\emph{Simultaneous inference.}}
For simultaneous inference, with a fixed node $a \in V$, we would like to test the null hypothesis
\begin{equation}
\label{eq:H0_simultaneous_general_L}
H_0: \theta_{ab}^{*( l)} = \breve\theta_{ab}^{\,(l)} \quad \text{for all } l \in \{1, \ldots, L\} \text{  and  } b \in V_a =  \{1, \ldots, p\} \backslash \{a\},
\end{equation}
for some fixed $\breve\theta_{ab}$ versus the alternative 
\begin{equation}
\label{eq:H1_simultaneous_general_L}
H_1: \theta_{ab}^{*( l)}  \neq \breve\theta_{ab}^{\,(l)} \quad \text{for some } l \in \{1, \ldots, L\} \text{  and  } b \in V_a = \{1, \ldots, p\} \backslash \{a\}.
\end{equation}
Again, the test involves a large number of parameters, $(p-1)L$.

First, note that we can directly apply the procedure developed in
Section 6.  By ignoring the covariance structure of
$\theta_{ab}^{(1)}, \ldots, \theta_{ab}^{(L)}$, we can directly use
the Gaussian multiplier bootstrap.  Specifically, for each
$b \in V_a$, we obtain the Bahadur representation in
\eqref{eq:bahadur_general_L}.  Next, we stack the resulting $p-1$
vectors into a $(p-1)L$ dimensional vector and perform the Gaussian
multiplier bootstrap method to calculate the test statistic and
critical values.  Since $L$ is an absolute constant, all the analysis
in Section \ref{sec:simultaneous} remains valid.  However, such a procedure disregards the
group structure on parameters and ignores the off-diagonal elements of
the covariance matrix $V_{ab}$ when constructing the test and
computing the critical values.

An alternative approach is based on the moderate deviation result
for the $\chi^2$-test developed in \cite{Liu2013Carmer}.
Here, we outline the procedure and refer to \cite{Liu2013Carmer} for technical details. 
First, for each $b \in V_a$, we define
\begin{equation}
T_{nb}^2 = n \cdot \Big( \tilde \theta_{ab}^{[L]} - \breve\theta_{ab}^{[L]} \Big)^\top \cdot ( V_{ab} )^{-1} \cdot \Big( \tilde \theta_{ab}^{[L]} - \breve\theta_{ab}^{[L]} \Big).
\end{equation}
It follows from \eqref{eq:asymptotic_normal_general_L}
that the limiting distribution of $T_{nb}^2$ is $\chi^2_L$.
Under mild conditions, Theorem 2.2 of \cite{Liu2013Carmer} shows that
\[
  \frac{\PP\rbr{T_{nb}^2 \geq x^2}}{\PP\rbr{\chi^2_L \geq x^2}} \rightarrow 1,\quad \text{as }n \rightarrow \infty
\]
uniformly for $x \in [0, o(n^{1/6}))$. 
This motivates the following test statistic
\begin{equation}
\max_{b \in V_a} \, T_{nb}^2.
\end{equation}
We obtain the critical value $y_\alpha$ that satisfies 
\begin{equation*}
(p-1) \cdot \PP \rbr{ \chi^2_L \geq y_\alpha } = - \log(1-\alpha).
\end{equation*}
The null hypothesis is rejected if
$ \max_{b \in V_a} T_{nb}^2 \geq y $.  We can prove that the
asymptotic Type I error is $\alpha$ under the null only when the
dependency among $T_{nb}^2$ is weak.  We refer to \cite{Liu2013Carmer}
for technical details.  The disadvantage of this approach is that, the
terms $T_{nb}^2$ are correlated across $b \in V_a$, which is ignored
when computing the critical value.  Despite ignoring the group
structure, the approach based on multiplier bootstrap can control the
Type I error better with small sample sizes. See Section \ref{sec:experiments_synthetic} for experimental results.

\section{Simulations}
\label{sec:experiments_synthetic}

In this section, we illustrate the finite sample properties of our
inference procedure on several synthetic data sets.  We generate data
from four different Exponential family distributions that were
introduced in Section~\ref{sec:ExpoGM}. The first and third example
involve Gaussian node-conditional distributions, for which we use
regularized score matching. For the second and fourth setting where
the node-conditional distributions follow Truncated Gaussian and
Exponential distribution, respectively, we use regularized
non-negative score matching procedure. Following the recommendation in
\cite{Yu2018Graphical}, we set $\ell_a(x) = \log(x + 1)$ for the
non-negative settings. In each example, we report the mean coverage
rate of 95\% confidence intervals for several coefficients averaged
over 500 independent simulation runs.

\paragraph{\emph{Gaussian graphical model.}}

For the Gaussian setting, we have $X \sim N(0,\Sigma)$ with precision matrix $\Omega = \Sigma^{-1} = (\theta_{ab})$. 
Without loss of generality, say we are interested in $\theta_{12}$. We have
$$
\theta^* = (\theta^*_{11}, \theta^*_{12},\ldots,\theta^*_{1p},\theta^*_{22},\theta^*_{23},\ldots,\theta^*_{2p})^T,
$$
$$
\varphi(x) = \Big(-\frac 12 x_1^2, -x_1x_2,\ldots,-x_1x_p,-\frac 12x_2^2,-x_2x_3,\ldots,-x_2x_p\Big)^T,
$$
$$
\varphi_1(x) = (-x_1,-x_2,\ldots,-x_p,0,\ldots,0)^T,
$$
$$
\varphi_2(x) = (0,-x_1,0,\ldots,0,-x_2,-x_3,\ldots,-x_p)^T,
$$
$$
g(x) = (-1,0,0,\ldots,0,-1,0,\ldots,0)^T,
$$
where for $g$ the second `$-1$' is at location $p+1$. Now we have
\begin{equation*}
  \begin{aligned}
    \gamma^{ab,*} &= \arg\min \
    \EE[
      (\vpaxi[,ab]-\vpaxi[,-ab]^T\gamma)^2 + 
      (\vpbxi[,ab]-\vpbxi[,-ab]^T\gamma)^2 
    ] \\
    &= \arg\min \
    \EE[
      (x_2-(x_1,x_3,\ldots,x_p,0,\ldots,0)^T\gamma)^2 + 
      (x_1-(0,\ldots,0,x_2,x_3,\ldots,x_p)^T\gamma)^2 
    ] .
  \end{aligned}
\end{equation*}

We can see that $\gamma^{ab,*}$ can be partitioned into first $p-1$ elements and last $p-1$ elements: $\gamma^{ab,*} = [\gamma^{ab,*}_1;\gamma^{ab,*}_2]$. The two parts can be optimized separately.
Moreover, both the population quantity $\varphi_1(x)\varphi_1(x)^\top$ and $\varphi_2(x)\varphi_2(x)^\top$ are the covariance matrix $\Sigma$ after rearranging terms and ignoring zero components. 
Assumption {\bf SE} is satisfied with most of the commonly used covariance matrices with full rank. 
Moreover, we can verify that $\gamma^{ab,*}_1$ and $\gamma^{ab,*}_2$ are proportional to the second and first column of the precision matrix $\Omega$. Therefore, assumption {\bf M} is satisfied when the columns of the precision matrix $\Omega$ are sparse.

For the experiment, we set diagonal entries of $\Omega$ as
$\theta_{jj} = 1$. The sparsity pattern of the precision matrix corresponds to
the the 4-nearest neighbor graph and the non-zero coefficients
are set as $\theta_{j, j-1} = \theta_{j-1, j} = 0.5$ and
$\theta_{j, j-2} = \theta_{j-2, j} = 0.3$. 
We set the sample size $n = 300$ and
vary the number of nodes $p$.  Table~\ref{simu1} shows the empirical
coverage rate for different values of $p$ for four chosen
coefficients. As is evident from the table, the coverage probabilities
for the unknown coefficient is remarkably close to nominal. 


\begin{table}[!h]
  \center
    \begin{tabular}{ccccc}
\hline
& $\theta_{1,2}$ & $\theta_{1,3}$ & $\theta_{1,4}$ & $\theta_{1,10}$ \\\hline
$p = 50$ & 95.4\% & 92.4\% & 93.8\% & 93.2\% \\
$p = 200$ & 94.6\% & 92.4\% & 92.6\% & 94.0\%\\
$p = 400$ & 94.6\% & 94.8\% & 92.6\% & 93.8\% \\\hline 
    \end{tabular}
    \caption{Empirical Coverage for Gaussian Graphical Model}
      \label{simu1}
\end{table}%

\paragraph{\emph{Non-negative Gaussian.}}

For simplicity we first consider score matching for non-negative Gaussian model with $\ell(x) = x^2$. 
Following the setting and notation in the previous paragraph, we have
\begin{equation*}
\begin{aligned}
\tilde\varphi_1(x) &= x_1 \cdot \varphi_1(x) = x_1 \cdot (-x_1,-x_2,\ldots,-x_p,0,\ldots,0)^T, \\
\tilde\varphi_2(x) &= x_2 \cdot \varphi_2(x) = x_2 \cdot (0,-x_1,0,\ldots,0,-x_2,-x_3,\ldots,-x_p)^T.
\end{aligned}
\end{equation*}
As before, $\gamma^{ab, *}$ is separable into two parts; we focus on one to obtain
\begin{equation*}
\gamma^{ab, *}_2 = \sbr{ \EE \,\, x_1^2 \cdot 
 \begin{pmatrix}
  x_1^2 & x_1x_3 & \cdots & x_1x_p \\
  x_1x_3 & x_3^2 & \cdots & x_3x_p \\
  \vdots  & \vdots  & \ddots & \vdots  \\
  x_1x_p & x_3x_p & \cdots & x_p^2
 \end{pmatrix}
 }^{-1}
 \cdot
 \sbr{ \EE \,\, x_1^2 x_2 \cdot 
 \begin{pmatrix}
  x_1 \\
  x_3 \\
  \vdots    \\
  x_p
 \end{pmatrix}
 }.
\end{equation*}
We can see that it contains expectations, such as $x_1^2x_3x_4$, which are hard to calculate explicitly, in addition to the matrix inversion. 
To the best of our knowledge, this calculation is intractable. 
If we instead use generalized score matching with $\ell(x) = \log(x+1)$, the calculation would be more complicated.

One exception is when the precision matrix $\Omega = I_p$, which means $x_i$ follows i.i.d. non-negative standard normal distribution.
Using the moments $\EE[x] = \sqrt{2/\pi}$, $\EE[x^2] = 1$, $\EE[x^3] = \sqrt{8/\pi}$, $\EE[x^4] = 3$, we can calculate $\gamma^{ab, *}$ explicitly. 
It turns out that the two parts in $\gamma^{ab, *}$ are the same. 
All their components take the same value at approximately $1/p$, except for one component that takes the value approximately $1.6/p$. Therefore, we can see that the sparsity assumption on $\gamma^{ab, *}$ is violated. 
It instead only satisfies a weaker condition that $\| \gamma^{ab, *} \|_1 \leq 2$ for large $p$. 
Similarly, we can calculate that $\| M^*_{ab} \|_1 \leq 5$ for large $p$. 
We then follow the debias method in Section \ref{sec:relaxation_L1} to construct confidence intervals.

For the simulation, we use the same setting as for the Gaussian graphical model with $\theta_{j, j-1} = \theta_{j-1, j} = 0.3$ and
$\theta_{j, j-2} = \theta_{j-2, j} = 0.1$. 
We set $\ell_a(x) = \log(x + 1)$, and use the minimax tilting method to generate the data \citep{Botev2017normal}. 
We first support the bounded $L_1$ norm condition of $M^*$ through experiments with a small $p = 20, 50$ and large $n$. 
Here we focus on the edge $(a,b)=(1,2)$; results for other edges are similar, and are therefore omitted. 
Since we have enough samples, we estimate $M$ as the exact inverse of the empirical quantity $\EE_n[ \Gamma(x_i') ]$. 
Table \ref{simu_verify_M_NNG} shows the average mean and maximum of the $L_1$ norm of $M$ on column $ab$, based on 500 independent simulation runs with different sample sizes. 
This shows that the $L_1$ norm of the column $ab$ of $M^*$ would be bounded from above.
These experimental results indicate that the bounded $L_1$ norm condition of $M^*$ is reasonable.

Table~\ref{simu_TN} shows the empirical coverage rate for various choices of $p$ and $n$. Note that since we are doing sample splitting, the real sample size is $2n$. 
We observe that by using the debias method, we can obtain nominal coverage rate even for relatively large $p$ with small~$n$. 

\begin{table}[!h]
  \center
\begin{tabular}{ccccc}
\hline
& $n = 500$ & $n = 2000$ & $n = 10000$ & $n = 50000$ \\\hline
averaged mean, $p = 20$ & 13.01 & 11.42 & 11.16 & 11.10 \\
averaged max, $p = 20$   & 17.84 & 13.46 & 11.95 & 11.52 \\\hline 
averaged mean, $p = 50$ & 24.71 & 15.19 & 12.90 & 12.72 \\
averaged max, $p = 50$   & 32.65 & 17.86 & 14.13 & 13.17 \\\hline 
\end{tabular}
\caption{Averaged mean and max of the $L_1$ norm of $M$, for Non-negative Gaussian}
  \label{simu_verify_M_NNG}
\end{table}%

\begin{table}[!h]
  \center
\begin{tabular}{ccccc}
\hline
& $\theta_{1,2}$ & $\theta_{1,3}$ & $\theta_{1,4}$ & $\theta_{1,10}$ \\\hline
$p = 100, n = 150$ & 94.2\% & 93.8\% & 95.0\% & 92.4\%\\
$p = 200, n = 300$ & 95.2\% & 96.6\% & 94.8\% & 94.6\%\\
$p = 300, n = 500$ & 94.8\% & 95.8\% & 95.0\% & 94.4\% \\\hline 
\end{tabular}
\caption{Empirical Coverage for Non-negative Gaussian, using debias method}
  \label{simu_TN}
\end{table}%




\paragraph{\emph{Normal conditionals.}} For the experiment, we 
consider a special case of normal conditionals with $L = 1$ parameter matrix, whose density is 
\begin{equation}
p(x ; B, \beta, \beta^{(2)}) \propto
  \exp \cbr{
    \sum_{a \neq b} \beta_{ab} x_a^2x_b^2 +
    \sum_{a \in V} \beta_a^{(2)} x_a^2 + 
    \sum_{a \in V} \beta_ax_a },\quad x \in \RR^{p}.
\end{equation}

This distribution is also considered in \citet{Lin2015High}.
We set
$\beta_j = 0.4$, $\beta_j^{(2)} = -2$, and we use a 4 nearest neighbor
lattice dependence graph with interaction matrix:
$\beta_{j, j-1} = \beta_{j-1, j} = -0.2$ and
$\beta_{j, j-2} = \beta_{j-2, j} = -0.2$.  Since the univariate
marginal distributions are all Gaussian, we generate the data using a
Gibbs sampler. The first 500 samples were discarded as `burn in' step,
and of the remaining samples, we keep one in three.

We first support the assumption {\bf M} through experiments with a small $p = 20$ and large $n$. 
Here we focus on the edge $(a,b)=(1,2)$; results for other edges are similar, and are therefore omitted. 
We estimate $\hat\gamma^{ab}$ as in Step 2, but without the $L_1$ regularization term since we have enough samples. 
For normal conditionals, we have $\hat\gamma^{ab} \in \RR^{2p} = \RR^{40}$.
There are five components in $\hat\gamma^{ab}$ with relatively large non-zero values (not decreasing with $n$), and we calculate the mean and maximum absolute value of the remaining 35 components. 
Table \ref{simu_verify_M_NC} shows the average mean and maximum absolute values of these 35 components, based on 500 independent simulation runs with different sample sizes. 
This suggests that the population quantity $\gamma^{ab, *}$ would be close to a sparse vector, with an infinite amount of samples.
These experimental results indicate that assumption {\bf M} is reasonable, at least in an approximately sparse version.

We then set the number
of samples $n = 500$, and follow the proposed three-step procedure to calculate the coverage rate. Table~\ref{simu2} shows the empirical coverage
rate for $p=100$ and $p=300$ nodes. Again, we see that our inference
algorithm behaves well on the above Normal Conditionals Model.

\begin{table}[!h]
  \center
\begin{tabular}{ccccc}
\hline
& $n = 500$ & $n = 2000$ & $n = 10000$ & $n = 50000$ \\\hline
average mean & $4.3 \times 10^{-3}$ & $2.7 \times 10^{-3}$ & $1.4 \times 10^{-3}$ & $0.7 \times 10^{-3}$ \\
average max   & $9.7 \times 10^{-3}$ & $8.4 \times 10^{-3}$ & $6.9 \times 10^{-3}$ & $5.5 \times 10^{-3}$ \\\hline 
\end{tabular}
\caption{Average mean and max on the 35 components, for Normal Conditionals}
  \label{simu_verify_M_NC}
\end{table}%

\begin{table}[!h]
\center
\begin{tabular}{ccccc}
\hline
& $\beta_{1,2}$ & $\beta_{1,3}$ & $\beta_{1,4}$ & $\beta_{1,10}$ \\\hline
$p = 100$ & 93.2\% & 93.4\% & 94.6\% & 95.0\%\\
  $p = 300$ & 93.2\% & 93.0\% & 92.6\% &93.0\%\\\hline
\end{tabular}
  \caption{Empirical Coverage for Normal Conditionals}
\label{simu2}
\end{table}%

\paragraph{\emph{Exponential graphical model.}} We choose $\theta_j = 2$, and a 2 nearest neighbor dependence graph with
$\theta_{j, j-1} = \theta_{j-1, j} = 0.3$. 
We again first support the assumption {\bf M} through experiment with a small $p = 20$ and large $n$, where we focus on the edge $(a,b)=(1,2)$ and use a Gibbs sampler to generate data. 
For exponential graphical model, we have $\hat\gamma^{ab} \in \RR^{2p-2} = \RR^{38}$.
There are four components in $\hat\gamma^{ab}$ with relatively large non-zero values (not decreasing with $n$), and we calculate the mean and maximum absolute value of the remaining 34 components. 
Table \ref{simu_verify_M_EGM} shows the average mean and maximum absolute values of these 34 components, based on 500 independent simulation runs with different sample sizes. 
This suggests that the population quantity $\gamma^{ab, *}$ would be close to a sparse vector, with an infinite amount of samples.
Once again, this experiment results indicate that assumption {\bf M} is reasonable, at least in an approximately sparse version.

We then set $n = 1000$ and the empirical coverage rate and
histograms of estimates of four selected coefficients are presented in
Table~\ref{simu3} and Figures~\ref{fig:histogram} for $p=100$ and
$p=300$, respectively.

\begin{table}[!h]
  \center
\begin{tabular}{ccccc}
\hline
& $n = 500$ & $n = 2000$ & $n = 10000$ & $n = 50000$\\\hline
average mean & $3.6 \times 10^{-3}$ & $2.2 \times 10^{-3}$ & $0.9 \times 10^{-3}$ & $0.4 \times 10^{-3}$ \\
average max   & $9.4 \times 10^{-2}$ & $6.8 \times 10^{-3}$ & $3.8 \times 10^{-3}$ & $1.2 \times 10^{-3}$ \\\hline 
\end{tabular}
\caption{Average mean and max on the 34 components, for Exponential Graphical Model}
  \label{simu_verify_M_EGM}
\end{table}%

\begin{table}[!h]
\center
\begin{tabular}{ccccc}
\hline
& $\theta_{1,2}$ & $\theta_{1,3}$ & $\theta_{1,4}$ & $\theta_{1,10}$ \\\hline
$p = 100$ & 94.2\% & 91.6\% & 92.6\% & 92.4\%\\
$p = 300$ & 92.6\% & 92.0\% & 92.2\% & 92.4\%\\\hline
\end{tabular}
\caption{Empirical Coverage for Exponential Graphical Model}
\label{simu3}
\end{table}%

\begin{figure*}[!t]
        \centering
        \begin{subfigure}[b]{0.235\textwidth}
            \centering
            \includegraphics[width=\textwidth]{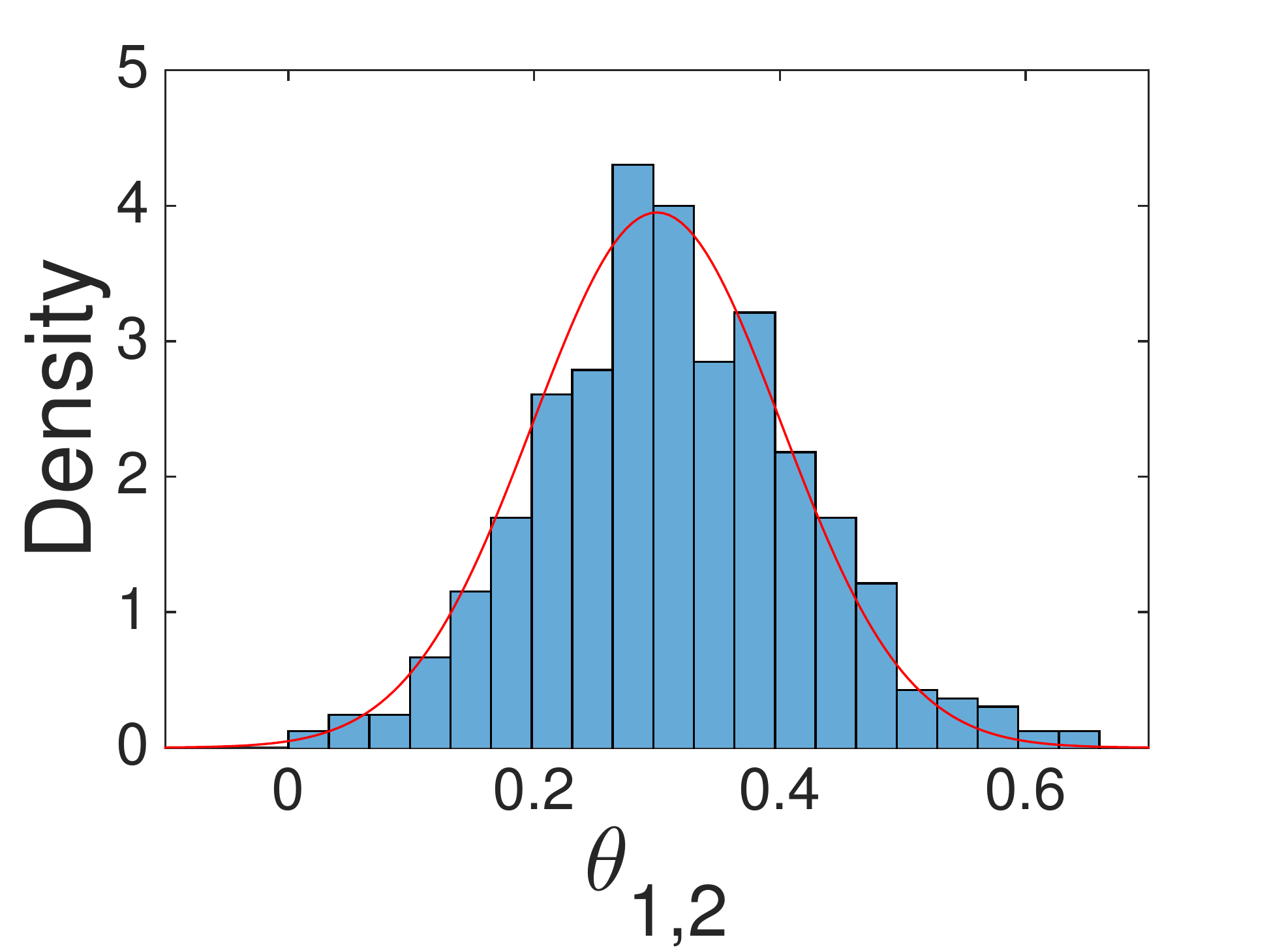} 
        \end{subfigure}
         \begin{subfigure}[b]{0.235\textwidth}
            \centering
            \includegraphics[width=\textwidth]{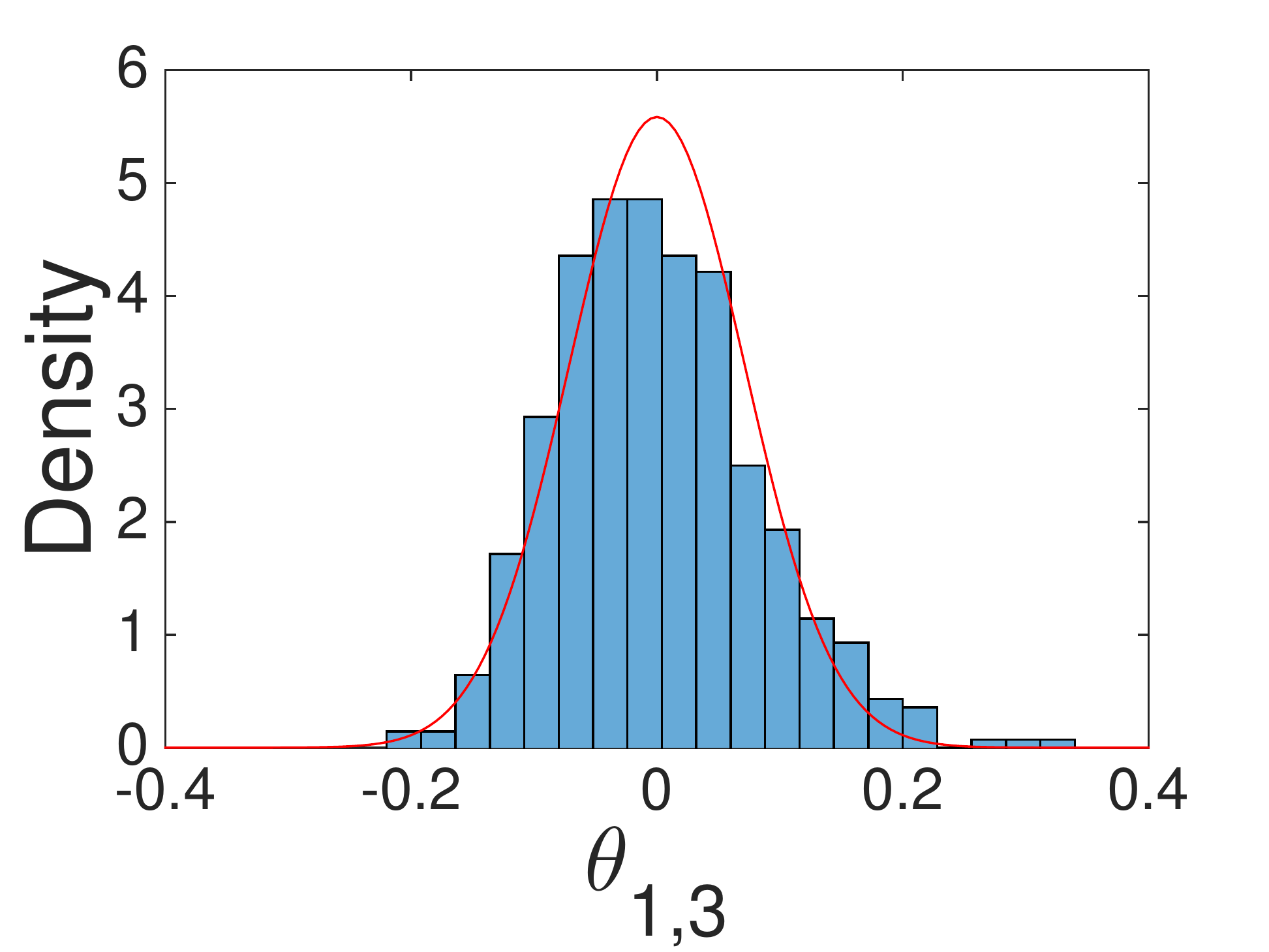}
        \end{subfigure}
         \begin{subfigure}[b]{0.235\textwidth}
            \centering
            \includegraphics[width=\textwidth]{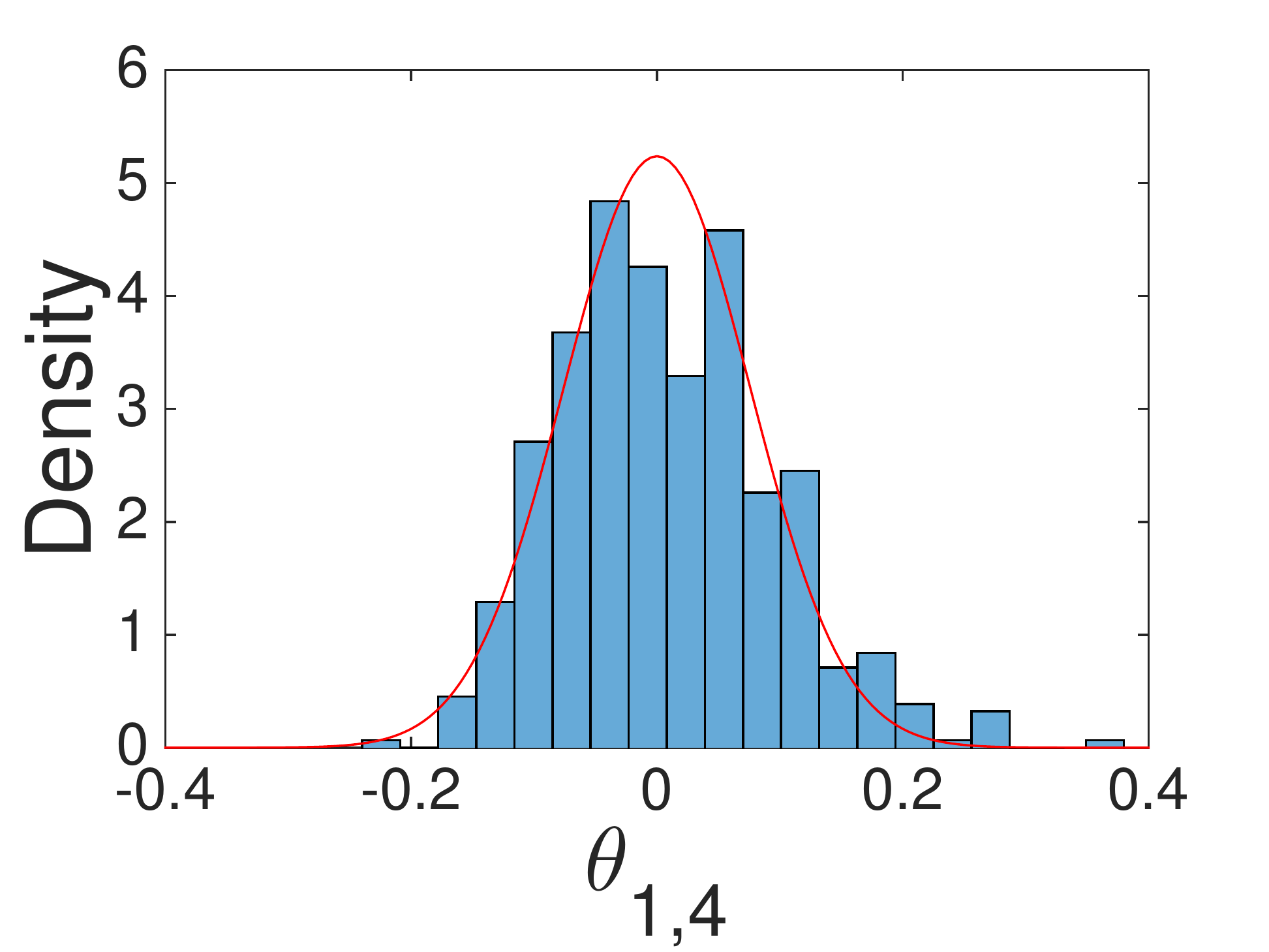}
        \end{subfigure}
        \begin{subfigure}[b]{0.235\textwidth}  
            \centering 
            \includegraphics[width=\textwidth]{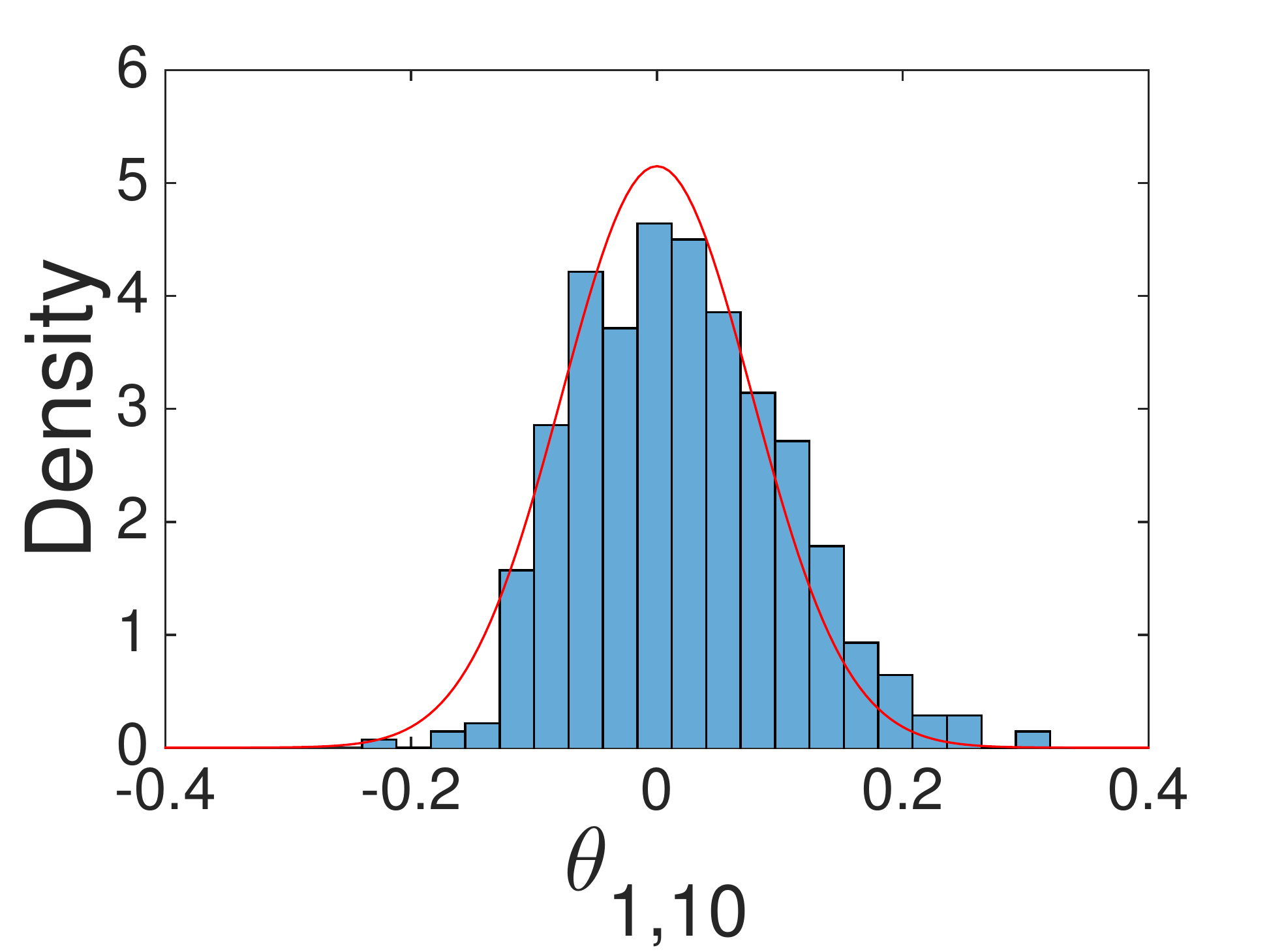}
        \end{subfigure}
        
        \hfill
        
        \vskip\baselineskip
        \begin{subfigure}[b]{0.235\textwidth}   
            \centering 
            \includegraphics[width=\textwidth]{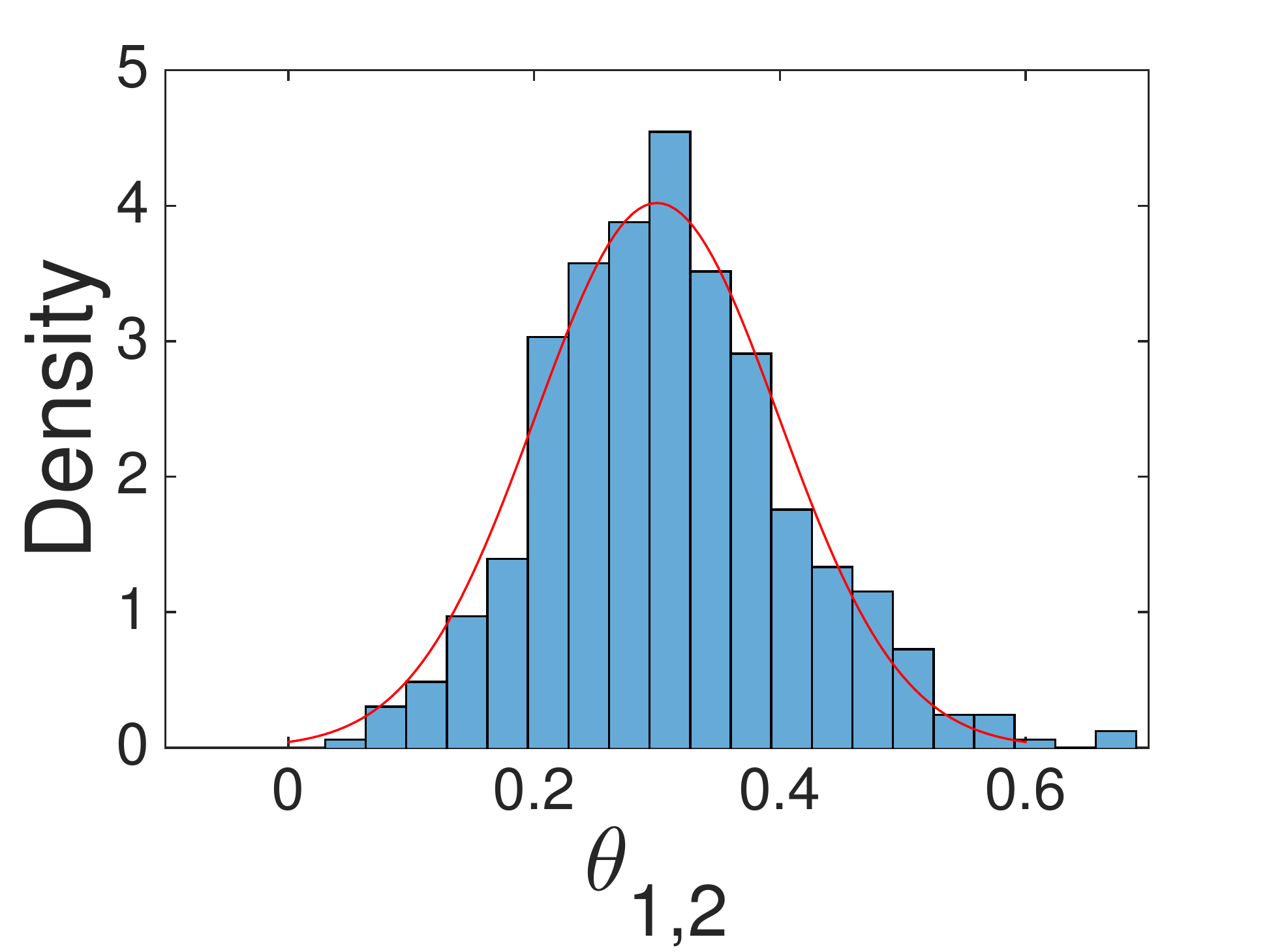}
        \end{subfigure}
        \begin{subfigure}[b]{0.235\textwidth}   
            \centering 
            \includegraphics[width=\textwidth]{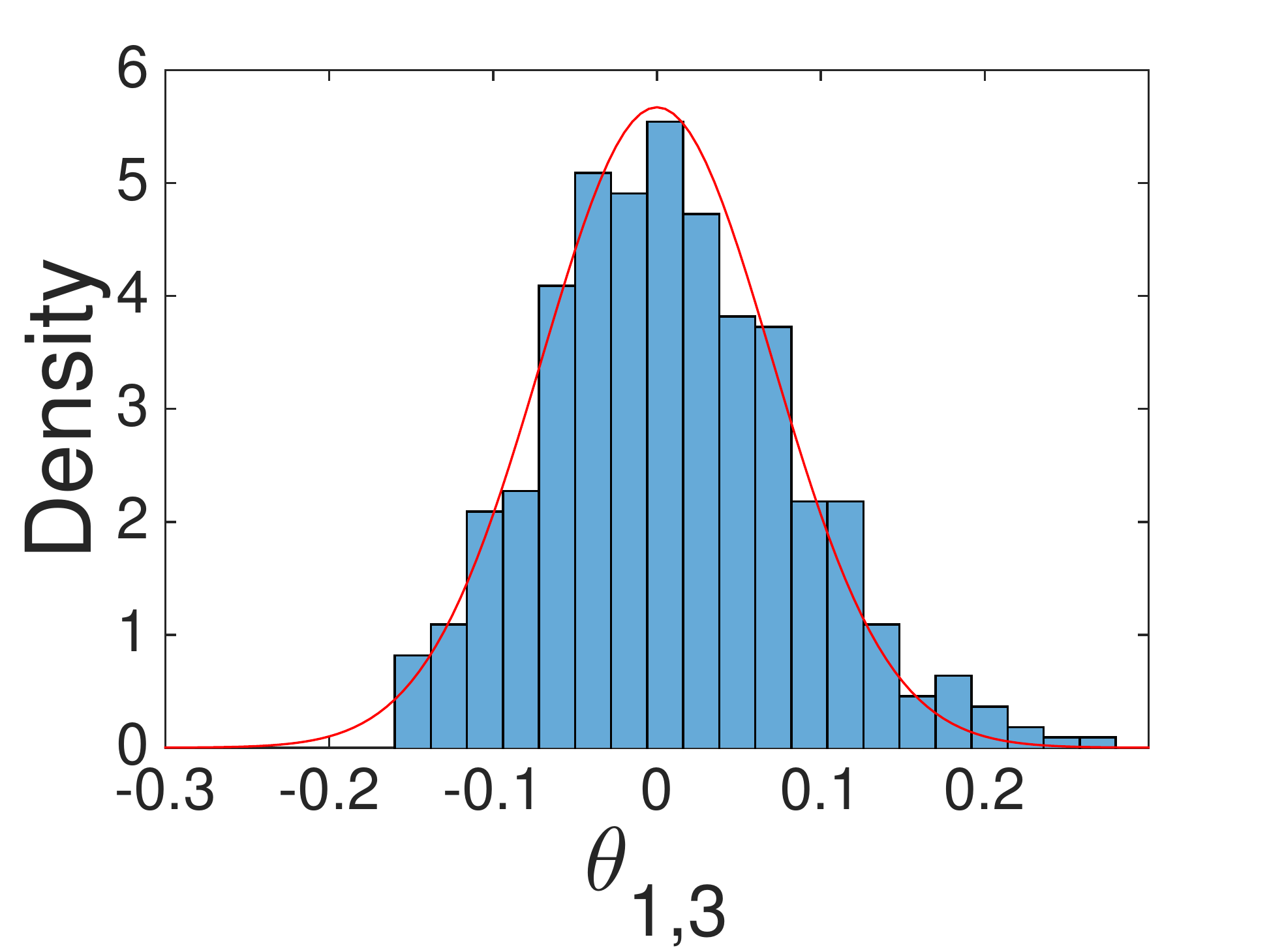}
        \end{subfigure}
        \begin{subfigure}[b]{0.235\textwidth}   
            \centering 
            \includegraphics[width=\textwidth]{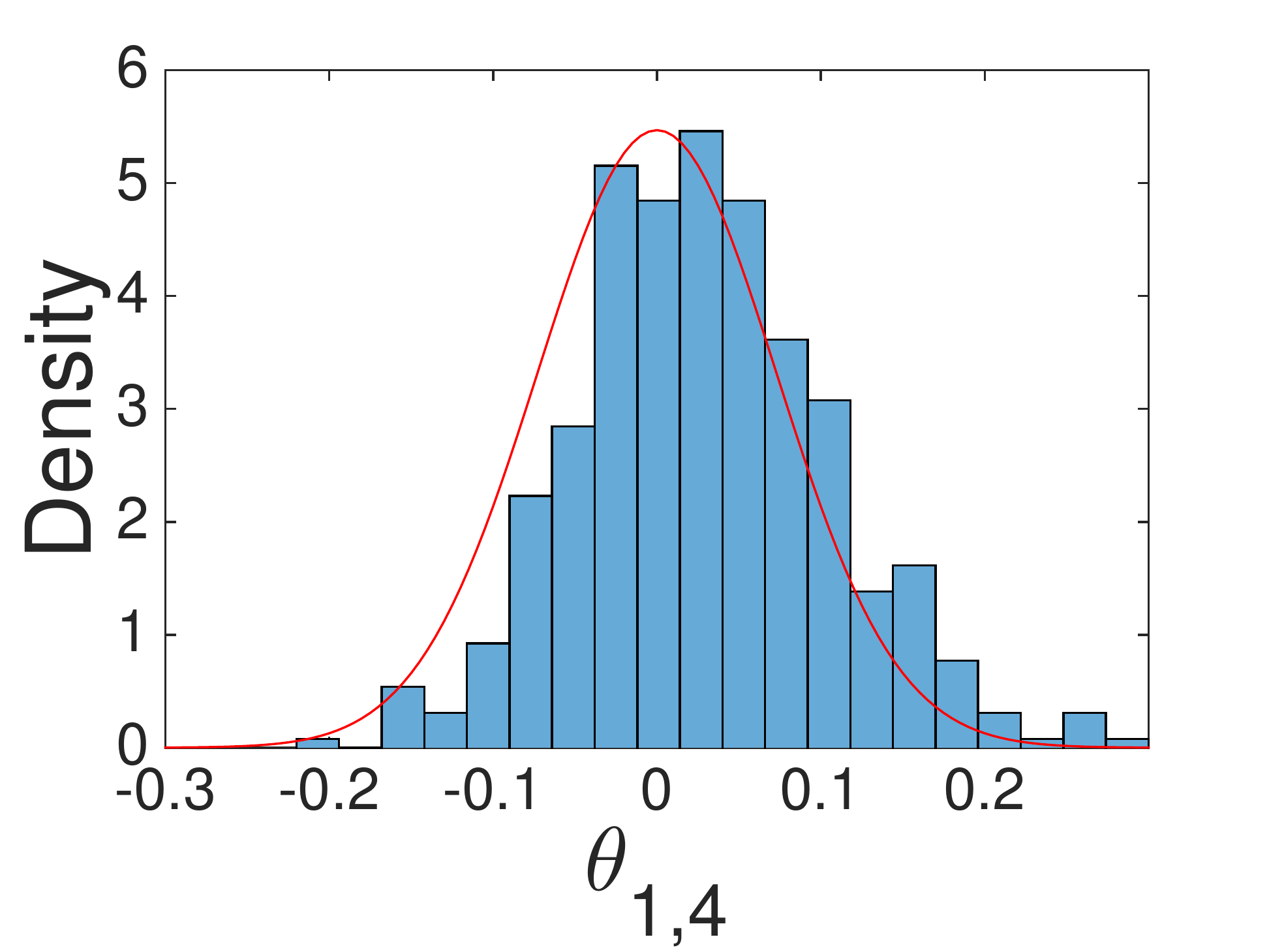}
        \end{subfigure}
        \begin{subfigure}[b]{0.235\textwidth}   
            \centering 
            \includegraphics[width=\textwidth]{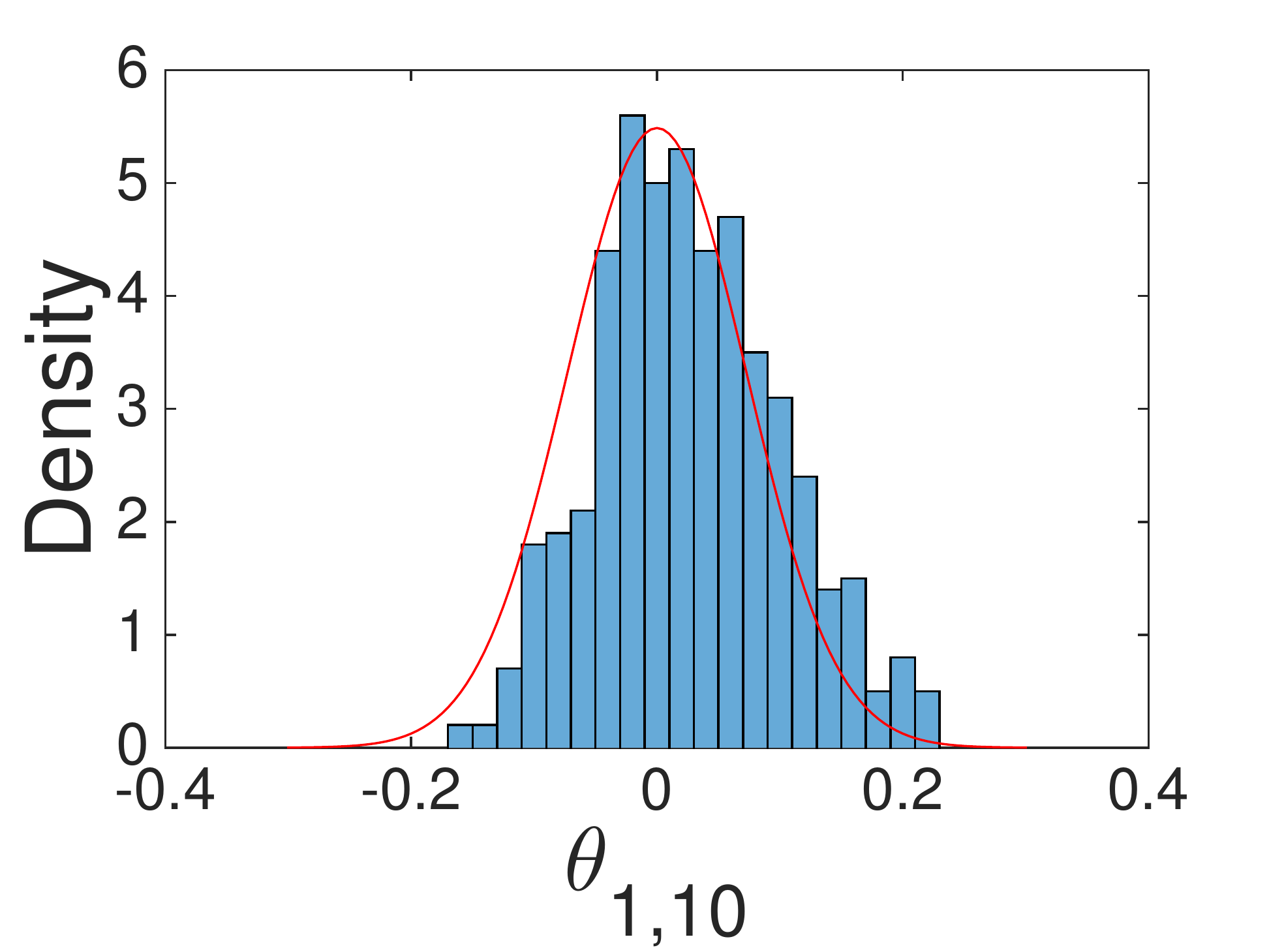}
        \end{subfigure}
        \caption{Histograms for $\theta$ for exponential graphical
          model. The first row corresponds to $p = 100$ and the second row to $p = 300$.}
        \label{fig:histogram}
\end{figure*}

We can see from the simulations here that we
need more samples for inference based on non-negative score matching to be valid, compared to regular score matching.
The results are still impressive as the sample size is
small relative to the total number of parameters in the model.
Moreover, by using the generalized score matching with
$\ell_a(x) = \log(x + 1)$, we get more accurate empirical coverage
compared to the original score matching, which uses $\ell_a(x) = x^2$.
The histograms in Figures~\ref{fig:histogram} show that the fitting is
quite good, but to get a better estimation and hence better coverage,
we would need more samples.

\paragraph{\emph{Simultaneous inference.}}
We then apply the simultaneous inference procedure to test for all
the edges connected to some node $a \in V$. Since the sample
complexity \eqref{eq:asmp_regime_simultaneous} for simultaneous
inference is large, we set $p = 50$.  
For hypothesis testing, we focus on the first node and 
we would like to test the null hypothesis
\begin{equation}
H_0: \theta_{1b}^* = \breve\theta_{1b} \quad \text{for all } b \in V_1 =  \{2, \ldots, p\} ,
\end{equation}
versus the alternative 
\begin{equation}
H_1: \theta_{1b}^* \neq \breve\theta_{1b} \quad \text{for some } b \in V_1 = \{2, \ldots, p\} .
\end{equation}
We set the designed Type I error as $\alpha = 0.05$ and we consider Gaussian and
Non-negative Gaussian settings as before. 
Table \ref{simultaneous_rerun_L1} shows the empirical Type I error under the null $\breve\theta_{1b} = \theta_{1b}^*$ with different choices of sample size. 
We see that our procedure works well as
long as we have enough data.

\begin{table}[!h]
\begin{center}
\begin{tabular}{cccccc}
\hline
& $n = 500$ & $n = 800$ & $n = 1000$ & $n=2000$ & $n = 5000$ \\\hline
Gaussian & 0.082 & 0.074 & 0.042 & 0.052 & 0.048\\
Non-negative Gaussian & 0.072 & 0.062 & 0.054 & 0.040 & 0.046  \\\hline
\end{tabular}
\end{center}
\caption{Empirical Type I error of simultaneous test}
\label{simultaneous_rerun_L1}
\end{table}%

\paragraph{\emph{Simultaneous inference with general $L$.}}
We finally consider the simultaneous inference with general $L$. We consider the normal conditionals model with density
\begin{equation*}
  p(x ; \Theta^{(1)}, \Theta^{(2)}, \eta, \beta) \propto
  \exp \cbr{
    \sum_{a \neq b} \Theta_{ab}^{(2)}x_a^2x_b^2 +
    \sum_{a \neq b} \Theta_{ab}^{(1)}x_ax_b +
    \sum_{a \in V} \eta_a x_a^2 + 
    \sum_{a \in V} \beta_ax_a },\, x \in \RR^{p}.
\end{equation*}

This corresponds to $L = K = 2$. We apply the two methods in Section \ref{sec:general_L} to test for all
the edges connected to some node $a \in V$. 
We set $p = 50$ and the designed Type I error $\alpha = 0.05$. 
For hypothesis testing, we focus on the first node (i.e., $a = 1$). 
Table \ref{simultaneous_rerun_L1_general_L} shows the empirical Type I error under the null with different choices of sample sizes. 
We see that both methods work well as long as we have enough data.

\begin{table}[!h]
\begin{center}
\begin{tabular}{ccccc}
\hline
& $n = 1000$ & $n = 2000$ & $n = 4000$ & $n=6000$ \\\hline
Gaussian multiplier bootstrap & 0.076  & 0.058 & 0.054 & 0.048 \\
 Moderate deviation method   & 0.182 & 0.092 & 0.068  & 0.056 \\\hline
\end{tabular}
\end{center}
\caption{Empirical Type I error of simultaneous test with general $L$}
\label{simultaneous_rerun_L1_general_L}
\end{table}%

\section{Protein Signaling Dataset}
\label{sec:experiments_real}

In this section we apply our algorithm to a protein signaling flow
cytometry data set, which contains the presence of $p = 11$ proteins
in $n = 7466$ cells \citep{Sachs2005Causal}.
\citet{Yang2013Graphical} fit exponential and Gaussian graphical
models to the data set. 

Figure~\ref{Real} shows the network structure after applying our
method to the data using an Exponential Graphical Model.  We learn the
structure directly from the data as well as provide confidence intervals using the Exponential Graphical
Model, rather than log-transforming the data and fitting Gaussian
graphical model as was done in \citet{Yang2013Graphical}.
To infer
the network structure, we calculate the $p$-value for each pair of
nodes, and keep the edges with $p$-values smaller than 0.01.
Estimated negative conditional dependencies are shown via red
edges. Recall that the exponential graphical model restricts the edge
weights to be non-negative, hence only negative dependencies can be
estimated. From the figure we see that PKA is a major protein
inhibitor in cell signaling networks. This result is consistent with
the estimated graph structure in \cite{Yang2013Graphical}, as well as
in the Bayesian network of \cite{Sachs2005Causal}. In addition, we
find significant dependency between PKC and PIP3.

\begin{figure}[!t]
\begin{center}
\includegraphics[width=10cm]{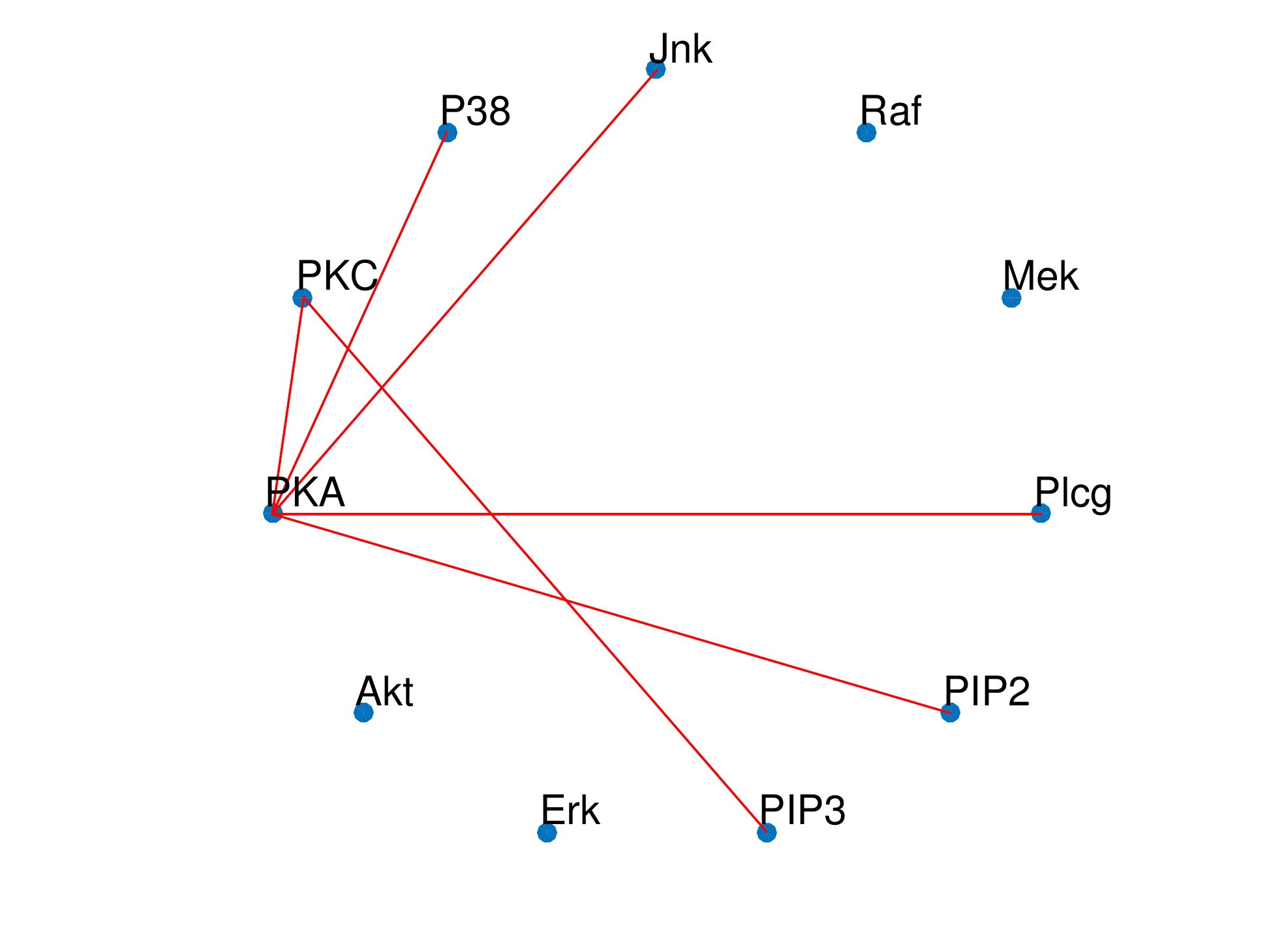}
\caption{Estimated Structure of Protein Signaling Dataset}
\label{Real}
\end{center}
\end{figure}

\section{Conclusion}
\label{sec:conclusion}

Motivated by applications in Biology and Social Networks, 
much progress has been made in statistical learning models and methods for networks with a large number of nodes. 
Graphical models provide a powerful and flexible modeling framework for such networks to uncover the dependency among nodes. 
As a result, there is a vast literature on estimation and inference algorithms for high dimensional Gaussian graphical models, as well as 
more general graphical models in the exponential family. 
As a disadvantage of most of these works, the normalizing constant (partition function) of the conditional densities is usually computationally intractable and without closed-form formula.
Score matching estimators provide a way to address this issue, but so far all the existing works on score matching focus on estimation problem for high-dimensional graphical models without statistical inference.
In this paper, we fill this gap by proposing a novel estimator using the score matching method that is asymptotically normal, which allows us to build statistical inference for a single edge of the graph.
Moreover, we propose the procedure on simultaneous testing on all the edges connected to some specific node in the graph, using the Gaussian multiplier bootstrap method. 
This procedure can be used to test if certain nodes are isolated or not, recover the support of the graph, and test the difference between two graphical models. 
There are a number of interesting and important directions that will be explored in future.
For example, developing inferential techniques based on score matching for 
multi-attribute graphical models 
\citep{kolar13multiatticml, Kolar2014Graph}, 
graphical models with confounders \citep{Geng2019Partially,Geng2018Joint},
time-varying graphical models 
\citep{Zhou08time, Kolar2010Estimating, kolar2011time},
networks with jumps
\citep{kolar10estimating} and conditional graphical models \citep{kolar10nonparametric},
as well as data with missing values \citep{kolar10nonparametric}.
It is also of interest to incorporate constraints in the model and perform constrained inference \citep{Yu2019Constrained}.
Finally, our method is developed for continuous data and developing results for 
discrete valued data is also of interest.

\acks{We are extremely grateful to the associate editor, Jie Peng, and two anonymous reviewers for their insightful comments that helped improve this paper. 
This work is partially supported by an IBM Corporation Faculty
  Research Fund and the William S. Fishman Faculty Research Fund at
  the University of Chicago Booth School of Business. This work was
  completed in part with resources provided by the University of
  Chicago Research Computing Center.}

\appendix

\section{Technical proofs}
\label{sec:technical_proofs}

We first establish a bound on the size of $\hat m_1 = \abr{\hat M_1}$
and $\hat m_2 = \abr{\hat M_2}$ in the following lemma.

\begin{lemma}
  \label{lem:size_m1}
  Assume the conditions of Theorem~\ref{thm:main} are satisfied. Then 
  \[\hat m_1 + \hat m_2 \lesssim  \phi_{\max}\phi_{\min}^{-2} m.\]
\end{lemma}
\begin{proof}
From the KKT  conditions we have that $\hat \theta^{ab}$ satisfies
\[
\EE_n\sbr{ \Gamma(x_i) \hat \theta^{ab} + g(x_i) } + \lambda_1 \cdot \hat\tau = 0,
\]
where $\hat \tau \in \partial\|\hat\theta^{ab}\|_1$. Restricted to
$\hat M_1$, we have (elementwise)
\[
\abr{\rbr{\EE_n\sbr{ \Gamma(x_i) \hat \theta^{ab} + g(x_i) }}_{\hat M_1}} = \lambda_1.
\]
Computing the $\ell_2$ norm on both sides, 
\begin{equation*}
\begin{aligned}
\sqrt{\hat m_1} \cdot \lambda_1 &= 
\bignorm{\rbr{\EE_n\sbr{ \Gamma(x_i) \hat \theta^{ab} + g(x_i) }}_{\hat M_1}}_2 \\
& \leq \bignorm{\rbr{\EE_n\sbr{ \Gamma(x_i) \rbr{\hat \theta^{ab} - \theta^{ab,*}}}}_{\hat M_1}}_2 +
\bignorm{\rbr{\EE_n\sbr{ \Gamma(x_i) \theta^{ab,*} + g(x_i) }}_{\hat M_1}}_2 \\
& \triangleq L_1 + L_2.
\end{aligned}
\end{equation*}
For the first term we have that 
\[
\begin{aligned}
L_1 & \leq \phi_{+}(\hat m_1 + m, \EE_n\sbr{ \Gamma(x_i) }) \cdot r_{2\theta}\\
& \lesssim
\phi_{+}(\hat m_1 + m, \EE_n\sbr{ \Gamma(x_i) }) \cdot \phi_{\min}^{-1}\cdot \lambda_1\sqrt{m},
\end{aligned}
\]
using \cite{negahban2010unified}. For the second term, we have that
\[
L_2 \leq \sqrt{\hat m_1}\cdot \lambda_1/2.
\]
Combining the two bounds, we obtain
\[
\sqrt{\hat m_1} \lesssim  
\phi_{+}(\hat m_1 + m, \EE_n\sbr{ \Gamma(x_i) }) \cdot \phi_{\min}^{-1}\sqrt{m}.
\]
Now, proceeding as in the proof of Theorem~3 in \cite{Belloni2013Least}, we establish that 
\[
\hat m_1 \lesssim  \phi_{\max}\phi_{\min}^{-2} m.
\]
The proof for $\hat m_2$ is similar.
\end{proof}

Our next result establishes bounds on $\tilde \theta^{ab} - \theta^{ab,*}$.

\begin{lemma}
  \label{lem:refit}
  Assume the conditions of Theorem~\ref{thm:main} are satisfied. Then 
\[
\begin{aligned}
\norm{\tilde \theta^{ab} - \theta^{ab,*}}_2
& \lesssim \phi_{\max}^{1/2}\phi_{\min}^{-2}\cdot\lambda_1\sqrt{ m }, \\
\norm{\tilde \theta^{ab} - \theta^{ab,*}}_1
& \lesssim \phi_{\max}^{1/2}\phi_{\min}^{-2}\cdot\lambda_1 m .
\end{aligned}
\]

\end{lemma}

\begin{proof}
From the KKT  conditions we have that $\hat \theta^{ab}$ satisfies
\[
\EE_n\sbr{ \Gamma(x_i)_{\hat M_1} } \hat \theta^{ab}_{\hat M_1} + \EE_n\sbr{ g(x_i)_{\hat M_1} } 
+ \lambda_1 \cdot {\rm sign}(\hat\theta^{ab}_{\hat M_1}) = 0,
\]
while $\tilde \theta^{ab}$ satisfies
\[
\EE_n\sbr{ \Gamma(x_i)_{\tilde M} } \tilde \theta^{ab}_{\tilde M} + \EE_n\sbr{ g(x_i)_{\tilde M} } = 0.
\]
Combining these two equations we have 
\[
\EE_n\sbr{ \Gamma(x_i)_{\tilde M} } \rbr{\tilde \theta^{ab}_{\tilde M} - \hat \theta^{ab}_{\hat M_1}}
 = \lambda_1 \cdot {\rm sign}(\hat\theta^{ab}_{\hat M_1})
\]
and
\[
\phi_{\min}\cdot \norm{\tilde \theta^{ab}_{\tilde M} - \hat \theta^{ab}_{\hat M_1}}_2 \leq 
\bignorm{
\EE_n\sbr{ \Gamma(x_i)_{\tilde M} } \rbr{\tilde \theta^{ab}_{\tilde M} - \hat \theta^{ab}_{\hat M_1}}
}_2
 = \lambda_1 \sqrt{\hat m_1} .
\]
Therefore, using \cite{negahban2010unified},
\[
\norm{\tilde \theta^{ab} - \theta^{ab,*}}_2
\leq \norm{\tilde \theta^{ab} - \hat \theta^{ab,*}}_2 + \norm{\hat \theta^{ab} - \theta^{ab,*}}_2
\lesssim \phi_{\min}^{-1}\cdot\lambda_1\sqrt{\hat m_1}.
\]
Combining with Lemma~\ref{lem:size_m1}, we obtain
\[
\begin{aligned}
\norm{\tilde \theta^{ab} - \theta^{ab,*}}_2
\lesssim \phi_{\max}^{1/2}\phi_{\min}^{-2}\cdot\lambda_1\sqrt{ m }
\quad\text{and}\quad
\norm{\tilde \theta^{ab} - \theta^{ab,*}}_1
\lesssim \phi_{\max}^{1/2}\phi_{\min}^{-2}\cdot\lambda_1 m .
\end{aligned}
\]
\end{proof}

A similar result can be established for
$\tilde \gamma^{ab} - \gamma^{ab,*}$, which we state without proof, as
it is analogous to the proof of Lemma~\ref{lem:refit}.
\begin{lemma}
  \label{lem:refit:gamma}
  Assume the conditions of Theorem~\ref{thm:main} are satisfied. Then 
\[
\begin{aligned}
\norm{\tilde \gamma^{ab} - \gamma^{ab,*}}_2
& \lesssim  \phi_{\max}^{1/2}\phi_{\min}^{-2}\cdot\lambda_2\sqrt{ m }, \\
\norm{\tilde \gamma^{ab} - \gamma^{ab,*}}_1
& \lesssim  \phi_{\max}^{1/2}\phi_{\min}^{-2}\cdot\lambda_2 m .
\end{aligned}
\]

\end{lemma}

To simplify notation later, let
$\tilde{r}_{j\theta} = \norm{\tilde \theta^{ab} - \theta^{ab,*}}_j$
and
$\tilde r_{j\gamma} = \norm{\tilde \gamma^{ab} - \gamma^{ab,*}}_j$,
for $j\in\{1,2\}$.
\begin{lemma}
\label{lem:L1}
Under the conditions of Theorem~\ref{thm:main}, we have
\[
  \abr{ \rbr{\tilde w - w^*}^\top \EE_n\sbr{\Gamma(x_i)} \rbr{\tilde \theta^{ab} - \theta^{ab,*}} }
  \lesssim  \phi_{\max}^2 \phi_{\min}^{-4} \cdot \lambda_1\lambda_2 {m}.
\]
\end{lemma}
\begin{proof}
 Let $\Scal_k$ be the set of $k$-sparse vectors in the unit ball,
\[
\Scal_k = \cbr{u \in \RR^p : \norm{u}_2\leq 1, \norm{u}_0 \leq k}.
\]
Abusing the notation, let $\norm{\cdot}_{\Scal_k}$ denote the sparse
spectral norm for matrices, that is,
\[
\norm{M}_{\Scal_k} = \max_{u,v\in\Scal_k}u^\top Mv.
\]
Using Lemma 4.9 of \cite{Barber2015ROCKET},
\[|u^\top M v| \leq 
\left(\norm{u}_2 + \norm{u}_1/\sqrt{k}\right)\cdot
\left(\norm{v}_2 + \norm{v}_1/\sqrt{k}\right)\cdot 
\sup_{u',v'\in\Scal_k}|u'^\top M v'|
\]
for any fixed matrix $M\in\RR^{p\times p}$ and vectors $u,v\in\RR^p$,
and any $k\geq 1$. With this, we have
\begin{equation*}
\begin{aligned}
    \rbr{\tilde w - w^*}^\top  \EE_n\sbr{\Gamma(x_i)} \rbr{\tilde \theta^{ab} - \theta^{ab,*}} 
    &\leq \norm{\EE_n\sbr{\Gamma(x_i)}}_{\Scal_{\tilde m}} \cdot
          \rbr{\tilde r_{2\gamma} + \tilde r_{1\gamma}/\sqrt{\tilde m}} 
          \cdot
          \rbr{\tilde r_{2\theta} + \tilde r_{1\theta}/\sqrt{\tilde m}} \\
    &\lesssim  \phi_{\max}^2\phi_{\min}^{-4} \cdot \lambda_1\lambda_2  m ,
  \end{aligned}
\end{equation*}
where the second line follows from the assumption {\bf SE}, and
Lemma~\ref{lem:refit} and  Lemma~\ref{lem:refit:gamma}.
\end{proof}

\begin{lemma}
\label{lem:L2}
Under the conditions of Theorem~\ref{thm:main}, we have
\[
    \abr{
      \rbr{\tilde w - w^*}^\top  \rbr{ \EE_n\sbr{\Gamma(x_i)}\theta^{ab,*} + \EE_n[g(x_i)] }
    } \lesssim \phi_{\max}^{1/2}\phi_{\min}^{-2}\cdot\lambda_1\lambda_2 m.
\]
\end{lemma}
\begin{proof}
Using H\"older's inequality, we have  
\[
\abr{
  \rbr{\tilde w - w^*}^\top  
  \rbr{ \EE_n\sbr{\Gamma(x_i)}\theta^{ab,*} + \EE_n[g(x_i)] }
}
\leq \tilde r_{1\gamma} \cdot 
\norm{\EE_n\sbr{\Gamma(x_i)}\theta^{ab,*} + \EE_n[g(x_i)]}_\infty .
\]
On the event $\Ecal_\theta$, we have
$\norm{\EE_n\sbr{\Gamma(x_i)}\theta^{ab,*} + \EE_n[g(x_i)]}_\infty
\leq \lambda_1/2$.
Finally, using Lemma~\ref{lem:refit:gamma}, we
conclude that
\[
\abr{ \rbr{\tilde w - w^*}^\top  \rbr{ \EE_n\sbr{\Gamma(x_i)}\theta^{ab,*}
    + \EE_n[g(x_i)] } } \lesssim
    \phi_{\max}^{1/2}\phi_{\min}^{-2}\cdot\lambda_1\lambda_2 m.
\]
\end{proof}

\begin{lemma}
\label{lem:L3}
Under the conditions of Theorem~\ref{thm:main}, we have
\begin{equation*}
\begin{aligned}
w^{* \top} \EE_n\sbr{\Gamma(x_i)} \rbr{\tilde \theta^{ab} - \theta^{ab,*}} =
\EE_n\sbr{\eta_{1i}\vpaxi[,ab] + \eta_{2i}\vpbxi[,ab]} & \rbr{\tilde \theta_{ab} - \theta_{ab}^{ab,*}} \\
& + \Ocal\rbr{\phi_{\max}^{1/2}\phi_{\min}^{-2}\cdot\lambda_1\lambda_2 m}.
\end{aligned}
\end{equation*}

\end{lemma}
\begin{proof}
  We have that
\begin{equation*}
\begin{aligned}
w^{* \top} \EE_n\sbr{\Gamma(x_i)} \rbr{\tilde \theta^{ab} - \theta^{ab,*}}  
& = \EE_n\sbr{\rbr{\eta_{1i}\vpaxi+\eta_{2i}\vpbxi}^\top } \rbr{\tilde \theta^{ab} - \theta^{ab,*}}     \\
& = \EE_n\sbr{\eta_{1i}\vpaxi[,ab]+\eta_{2i}\vpbxi[,ab]} \rbr{\tilde \theta_{ab}^{ab} - \theta_{ab}^{ab,*}}     \\
&\quad  + \EE_n\sbr{\rbr{\eta_{1i}\vpaxi[,-ab]+\eta_{2i}\vpbxi[,-ab]}^\top } \rbr{\tilde \theta_{-ab}^{ab} - \theta_{-ab}^{ab,*}}.     
\end{aligned}
\end{equation*}
For the second term, we have
\begin{equation*}
\begin{aligned}
&\abr{\EE_n\sbr{\rbr{\eta_{1i}\vpaxi[,-ab]+\eta_{2i}\vpbxi[,-ab]}^\top } \rbr{\tilde \theta_{-ab}^{ab} - \theta_{-ab}^{ab}}} \\
& \qquad \leq \tilde r_{1\theta} \cdot \norm{\EE_n\sbr{\eta_{1i}\vpaxi[,-ab]+\eta_{2i}\vpbxi[,-ab]}}_{\infty} \\
&\qquad \leq \tilde r_{1\theta} \cdot \lambda_2/2,
\end{aligned}
\end{equation*}
since we are working on the event $\Ecal_\gamma$.  Since
$\tilde{r}_{1\theta} \leq
\phi_{\max}^{1/2}\phi_{\min}^{-2}\cdot\lambda_1 m$, combining with the
display above, the proof is complete.
\end{proof}

\begin{lemma}
\label{lem:L4}
Under the assumptions {\bf M} and {\bf R}, we have that 
\[
\sqrt{n} \cdot w^{* \top}\rbr{ \EE_n\sbr{\Gamma(x_i)\theta^{ab,*} +
  g(x_i)} } \longrightarrow_D N\rbr{0, H(\theta^*)},
\]
where
$H(\theta^*) = \Var\rbr{w^{* \top}\rbr{ {\Gamma(x_i)\theta^{ab,*} + g(x_i)} }}$.
\end{lemma}
\begin{proof}
Let $Z_i = w^{* \top}\rbr{\Gamma(x_i)\theta^{ab,*} + g(x_i)}$. 
Then
\begin{align}
  \sqrt{n} \cdot w^{* \top}\rbr{ \EE_n\sbr{\Gamma(x_i)\theta^{ab,*} + g(x_i)} }
  = \frac{1}{\sqrt{n}} \sum_i Z_i.
\end{align}
From \cite{Forbes2013Linear}, we have that $\EE[Z_i] = 0$ and
$\Var(Z_i)$ is finite.  An application of the central limit theorem
completes the proof.
\end{proof}

\begin{lemma}
  \label{lem:variance_consistent}  
  The variance estimator $\hat V_{ab}$ is consistent, $\hat V_{ab} \rightarrow_P V_{ab}$.
\end{lemma}

\begin{proof}
  The variance estimator is obtained by using the second sample
  moment, and replacing true $\theta^{ab, *}, \gamma^{ab,*}$ with
  $\tilde \theta^{ab}, \tilde \gamma^{ab}$. We show the consistency of
  $\hat V_{ab}$ by showing the consistency of the estimator for
  $\sigma_n$ and
  $\text{Var}\big( w^{*,T}(\Gamma(x_i)\theta^{ab,*} + g(x_i)) \big)$,
  respectively.

  \paragraph{Step 1.} We can write
\begin{equation*}
\begin{aligned}
\sigma_n &= \mathbb{E}_n \big[\eta_{1i}\varphi_{1,ab}(x_i) + \eta_{2i}\varphi_{2,ab}(x_i)\big] \\
&= \mathbb{E}_n \big[w^{*,\T}\varphi_1(x_i) \cdot \varphi_{1,ab}(x_i) + w^{*,\T}\varphi_2(x_i) \cdot \varphi_{2,ab}(x_i)] \\
&= w^{*\T} \cdot \mathbb{E}_n [\Gamma(x_i)] \cdot e_{ab}.
\end{aligned}
\end{equation*}
Let
$\sigma = \EE[\sigma_n] = w^{*\T} \cdot \mathbb{E} [\Gamma(x_i)]
\cdot e_{ab}$ denote the population version of $\sigma_n$ and
$\tilde \sigma_n = \tilde w^{\top} \cdot \mathbb{E}_n [\Gamma(x_i)] \cdot
e_{ab}$ the sample version. With high probability we have that 
\begin{equation*}
\begin{aligned}
| \tilde \sigma_n - \sigma |  &\leq | \tilde \sigma_n - \sigma_n | + | \sigma_n - \sigma | \\
&\leq \Big|(\tilde w - w^*)^{\T} \cdot \mathbb{E}_n [\Gamma(x_i)] \cdot e_{ab} \Big| 
	+ \Big| {w^*}^{\T} \cdot \big[ \mathbb{E}_n [\Gamma(x_i)] - \mathbb{E} [\Gamma(x_i)] \big] \cdot e_{ab} \Big| \\
& \leq \|\tilde w - w^*\|_1 \cdot \big\|\mathbb{E}_n [\Gamma(x_i)] \cdot e_{ab} \big\|_\infty
	+ \| w^* \|_1 \cdot \big\| \big[ \mathbb{E}_n [\Gamma(x_i)] - \mathbb{E} [\Gamma(x_i)] \big] \cdot e_{ab} \big\|_\infty\\
& \lesssim  \lambda_2 m \cdot (C + \sqrt{\log p/n}) + m \cdot \sqrt{\log p/n} = o_P(1).
\end{aligned}
\end{equation*}

\paragraph{Step 2.} We estimate the variance of $w^{*\top}\big(\Gamma(x_i)\theta^{ab,*} + g(x_i)\big)$.
Since
\[
  \EE\sbr{w^{*\top}\big(\Gamma(x_i)\theta^{ab,*} + g(x_i)\big)} = 0,
\]
we can use the second sample moment to estimate the variance. As above, we plug in $\tilde \theta^{ab}$ and
$\tilde \gamma^{ab}$, to obtain that 
\begin{equation*}
\begin{aligned}
& \Bigg|  \mathbb{E}_n\bigg\{ \tilde w^\T\Big(\Gamma(x_i)\tilde\theta^{ab} + g(x_i)\Big) \bigg\}^2 - \mathbb{E}_n\bigg\{ w^{*\T}\Big(\Gamma(x_i)\theta^{ab,*} + g(x_i)\Big) \bigg\}^2\Bigg|  \\
& \qquad =  \Bigg| \EE_n \bigg\{ \tilde w^\T\big(\Gamma(x_i)\tilde\theta^{ab} + g(x_i)\big) - w^{*\T}\big(\Gamma(x_i)\theta^{ab,*} + g(x_i)\big) \bigg\}  \\
& \qquad  \quad \quad \quad  \quad \quad \quad \cdot \bigg\{\tilde w^\T\big(\Gamma(x_i)\tilde\theta^{ab} + g(x_i)\big) + w^{*\T}\big(\Gamma(x_i)\theta^{ab,*} + g(x_i)\big) \bigg\}  \Bigg| \\
& \qquad  \lesssim   \EE_n \bigg| \tilde w^{\T} \rbr{\Gamma(x_i) \tilde \theta^{ab}  +  g(x_i) }     - w^{*\T} \rbr{\Gamma(x_i)\theta^{ab,*} + g(x_i) }  \bigg| \\
& \qquad \lesssim  \EE_n \bigg| (\tilde w - w^*)^\top \rbr{\Gamma(x_i)\theta^{ab,*} + g(x_i) } + \tilde w^\top \Gamma(x_i)(\tilde \theta^{ab} - \theta^{ab,*}) \bigg| \\
& \qquad \lesssim { \|\tilde w - w^*\|_1 \cdot \EE_n\Big\|\Gamma(x_i)\theta^{ab,*} + g(x_i)\Big\|_\infty +  
	\|\tilde \theta^{ab} - \theta^{ab,*}\|_1 \cdot \EE_n\Big\| \tilde w^\top \Gamma(x_i) \Big\|_\infty }\\
& \qquad = o_P(1).
\end{aligned}
\end{equation*}
Combining the results of the two steps, completes the proof.
\end{proof}

\paragraph{Proof of Theorem \ref{thm:simultaneous}}
Denote  
\begin{equation}
\label{eq:def_W0}
W_0 = \max_{b \in V_a} \frac{1}{\sqrt n} \sum_{i=1}^n z_{iab} e_{i}
\end{equation}
as the counterpart to $\tilde W$.  Let
\begin{equation}
\label{eq:def_T0}
T_0 = \max_{b \in V_a} \frac{1}{\sqrt n} \sum_{i=1}^n z_{iab} \quad \text{and} \quad \tilde T = \max_{b \in V_a} \frac{1}{\sqrt n} \sum_{i=1}^n \tilde z_{iab}. 
\end{equation}
Denote 
\begin{equation}
\Delta = \max_{b,c \in V_a} \bigg| \frac 1n \sum_{i=1}^n \gamma_{abc}(x_i) \bigg|,
\end{equation}
where $\gamma_{abc}(x_i)$ is defined in assumption {\bf RR}. 
In order to apply Theorem 3.2 in \cite{Chernozhukov2013Gaussian},
we check the following conditions:
\begin{enumerate}
\item $\PP(\Delta \geq n^{-c}) \leq n^{-c}$.
\item $\PP( |T_0 - \tilde T| \geq n^{-c} ) \leq p^{-c}$.
\item With probability at least $1-p^{-c}$, $\PP_e( |W_0 - \tilde W| \geq n^{-c} ) \leq n^{-c}$.
Here $\PP_e$ denotes the probability with respect to $\{e_i\}_{i=1}^n$, conditionally on the observed data.
\end{enumerate}

We verify the first condition by applying Lemma A.1 in
\citet{Geer2008High}. By the definition of $\gamma_{abc}(x_i)$, clearly
we have $\EE\sbr{\gamma_{abc}(x_i)} = 0$. Together with assumption {\bf RR},
we apply Lemma A.1 in \cite{Geer2008High} and obtain
\begin{equation*}
\EE[\Delta] \leq \sqrt{\frac{4\tau_n^2\log{(2p)}}{n}} + \frac{2\eta_n\log{(2p)}}{n}.
\end{equation*}
According to \eqref{eq:asmp_regime_simultaneous}, for sufficiently
large $n$, we have $\EE[\Delta] \leq n^{-2c}$, for some $c > 0$. By
Markov inequality,
\[
  \PP(\Delta \geq n^{-c}) \leq n^c \cdot \EE[\Delta] \leq n^{-c},
\]
which verifies the first condition.

Next, we verify the second condition. For a fixed $b \in V_a$, under
the null, we have
\begin{equation*}
\begin{aligned}
\bigg| \frac{1}{\sqrt n} \sum_{i=1}^n z_{iab} - \frac{1}{\sqrt n} \sum_{i=1}^n \tilde z_{iab} \bigg| &\leq \sqrt{n} \bigg| (\sigma_{ab}^{-1} - \sigma_{n,ab}^{-1}) \cdot  w_{ab}^{*\T} \Big( \EE_n\big[\Gamma_{ab}(x_i)  \theta^{ab,*} + g_{ab}(x_i)\big] \Big) \bigg| \\
& \qquad + \sqrt{n} \bigg| \sigma_{n,ab}^{-1} \cdot (w_{ab}^* - \tilde w_{ab})^\T \Big( \EE_n\big[\Gamma_{ab}(x_i)  \theta^{ab,*} + g_{ab}(x_i)\big] \Big) \bigg| \\
&\leq \sqrt{n} C\cdot\lambda_1\lambda_2 m \\
&\leq n^{-c},
\end{aligned}
\end{equation*}
with probability at least $1-p^{-c-1}$, where the second inequality
comes from the consistency of $\sigma_n$, Lemma \ref{lem:L2},
and Lemma \ref{lem:L4}.  We then have
\begin{equation*}
\begin{aligned}
\PP( |T_0 - \tilde T| \geq n^{-c} ) & \leq \PP\Big(\bigcup_{b \in V_a} \Big\{\frac{1}{\sqrt n} \big|\sum_{i=1}^n z_{iab} - \sum_{i=1}^n \tilde z_{iab} \big| \geq n^{-c} \Big\} \Big) \\
&\leq \sum_{b \in V_a} \PP\Big( \frac{1}{\sqrt n} \big | \sum_{i=1}^n z_{iab} - \sum_{i=1}^n \tilde z_{iab} \big| \geq n^{-c} \Big) \\
&\leq p \cdot p^{-c-1} = p^{-c},
\end{aligned}
\end{equation*}
which verifies the second condition.

Finally, we verify the third condition. We have
\begin{equation}
\begin{aligned}
\label{eq:diff_W}
\PP_e( |W_0 - \tilde W| \geq n^{-c} ) & \leq \PP_e\Big(\max_{b \in V_a} \Big\{\frac{1}{\sqrt n} \big|\sum_{i=1}^n (z_{iab} - \tilde z_{iab}) e_i \big| \Big\} \geq n^{-c}  \Big). \\
\end{aligned}
\end{equation}
Denote
$Z_b = \frac{1}{\sqrt n} \sum_{i=1}^n (z_{iab} - \tilde z_{iab})e_i$.
Under the null we have
\begin{equation*}
\begin{aligned}
z_{iab} - \tilde z_{iab} &= \Big[ (\sigma_{ab}^{-1} - \sigma_{n,ab}^{-1}) \cdot  w_{ab}^{* \top} \big( \Gamma_{ab}(x_i)  \theta^{ab,*} + g_{ab}(x_i)\big) \Big]  \\
& \qquad\qquad\qquad + \Big[ \sigma_{n,ab}^{-1} \cdot (w_{ab}^* - \tilde w_{ab})^\top \big( \Gamma_{ab}(x_i)  \theta^{ab,*} + g_{ab}(x_i) \big) \Big].
\end{aligned}
\end{equation*}
According to Lemma A.1 in \cite{Chernozhukov2013Gaussian}, we have
\begin{equation*}
\EE\bigg[ \frac 1n \Big\|\sum_{i=1}^n \Big(\Gamma_{ab}(x_i)  \theta^{ab,*} + g_{ab}(x_i)\Big)e_i \Big\|_\infty \bigg] \lesssim \sigma_0\sqrt{\frac{\log p}{n}} + \frac{M\log p}{n},
\end{equation*}
uniformly for each $b \in V_a$, where
\begin{equation}
\sigma_0^2 = \max_{j} \frac 1n \sum_{i=1}^n \Big[ \Big(\Gamma_{ab}(x_i)  \theta^{ab,*} + g_{ab}(x_i)\Big)e_i \Big]_j^2,
\end{equation}
and 
\begin{equation}
M^2 = \EE \bigg[ \max_{i}  \Big\| \Big(\Gamma_{ab}(x_i)  \theta^{ab,*} + g_{ab}(x_i)\Big)e_i \Big\|_\infty \bigg]^2. 
\end{equation}
We then have
\begin{equation*}
\begin{aligned}
\EE |Z_b| &\leq \frac{1}{\sqrt n} \Big((\sigma_{ab}^{-1} - \sigma_{n,ab}^{-1}) \cdot \|w_{ab}^*\|_1 + \sigma_{n,ab}^{-1} \cdot \|w_{ab}^* - \tilde w_{ab}\|_1\Big) \\
& \qquad\qquad \times \EE \bigg[ \Big\|\sum_{i=1}^n \Big(\Gamma_{ab}(x_i)  \theta^{ab,*} + g_{ab}(x_i)\Big)e_i \Big\|_\infty \bigg] \\
&\leq \frac{C}{\sqrt n} \cdot \lambda m \cdot \bigg(\sigma_0\sqrt{\frac{\log p}{n}} + \frac{M\log p}{n}\bigg) \cdot n \\
&\leq n^{-2c},
\end{aligned}
\end{equation*}
uniformly for each $b \in V_a$ with probability at least $1-p^{-c}$, where the second inequality
comes from the consistency of $\sigma_n$ and Lemma
\ref{lem:refit:gamma}. Applying Markov inequality again, we obtain
\[
  \PP_e(|Z_b| \geq n^{-c}) \leq n^c \cdot \EE|Z_b| \leq n^{-c}.
\]
uniformly for each $b \in V_a$ with probability at least $1-p^{-c}$. 
Plugging back to \eqref{eq:diff_W}, we obtain
\begin{equation}
\begin{aligned}
\PP_e( |W_0 - \tilde W| \geq n^{-c} ) &\leq  \PP_e\Big( \max_{b \in V_a} |Z_b| \geq n^{-c} \Big) \leq n^{-c}
\end{aligned}
\end{equation}
with probability at least $1-p^{-c}$, which verifies the third condition.

With the three conditions verified and assumption {\bf RR}, we apply
Theorem 3.2 in \cite{Chernozhukov2013Gaussian} to obtain
\begin{equation*}
\sup_{\alpha \in (0,1)} \bigg| \PP\Big(\max_{b \in V_a} \sqrt{n} ( \tilde \theta_{ab} - \breve \theta_{ab} ) \geq c_{\tilde W}(\alpha) \Big) - \alpha  \bigg| = o(1),
\end{equation*}
which completes the proof.

\bibliography{19-383}

\end{document}